\documentclass[aps,pra,reprint,groupedaddress]{revtex4-2}

\usepackage{amsmath,amssymb}
\usepackage{dsfont}
\newcommand{\Id}{\mathds{1}}
\usepackage{graphicx}
\usepackage{booktabs}
\usepackage{xcolor}
\usepackage{bm}
\usepackage{tikz}
\usetikzlibrary{arrows.meta,positioning}
\usepackage{hyperref}
\usepackage{amsthm}
\newtheorem{lemma}{Lemma}
\usepackage{braket}
\usepackage[normalem]{ulem}
\usepackage[most]{tcolorbox}

\begin{document}

\title{A Hardware-Native Trapped-Ion Hybrid Qudit--Oscillator Instruction Set}

\author{Masoud Hakimi Heris}
\affiliation{Department of Electrical and Computer Engineering, North Carolina State University, Raleigh, NC, USA}
\author{Yuan Liu}
\affiliation{Department of Electrical and Computer Engineering, North Carolina State University, Raleigh, NC, USA}
\affiliation{Department of Computer Science, North Carolina State University, Raleigh, NC, USA}
\affiliation{Department of Physics, North Carolina State University, Raleigh, NC, USA}
\author{Frank Mueller}
\affiliation{Department of Computer Science, North Carolina State University, Raleigh, NC, USA}

\date{\today}

\begin{abstract}
Scaling trapped-ion processors by adding ions remains challenging while
maintaining sufficient fidelity of long and wide circuits.
A complementary route exploits the
multilevel internal structure of each ion to increase the computational
Hilbert space per ion. Hybrid qubit--oscillator architectures already
retain bosonic modes as native oscillators, and we extend this paradigm by
compressing the discrete subsystem into a multilevel register as well.
We introduce a hardware-native hybrid qudit--oscillator instruction set
for a general $d$-level internal register coupled to bosonic modes and
prove that three of its native operation families suffice for universality
of one qudit coupled to several modes, with no synthesized mode-coupling
operation required. Adding an entangling operation between ions extends
universality to a connected register of several qudits. The architecture
natively supports both phase-space and sideband instruction sets with
cross-compilation between them and is implemented in the open-source
framework HybridLane.

As a representative application, we reproduce the dissipative vibronic
dynamics of a three-chromophore complex previously simulated on a hybrid
qubit--oscillator superconducting platform. Moving from a
one-qubit-per-state register to a single qudit stores the entire
electronic manifold in one ion whenever it fits within the usable levels,
so the computational register shrinks from one ion per chromophore to a
single ion, every electronic coupling becomes a local intra-ion
transition, and the lab-frame Hamiltonian compiles directly into the ISA
without a frame transformation.
Multitone drives and higher-rank conditional operations
further reduce the number of control blocks. For the three-chromophore
benchmark, one Trotter step uses $34$ ISA operations against $32$ for the
qubit register. Multitone compilation reduces the qudit implementation to
$27$ control blocks. Under the stated sequential timing model informed by
measured control parameters from a $^{40}$Ca$^+$ qudit processor, one
Trotter step is projected to run about ten times faster on the qudit. The
resulting dynamics are validated against the published dissipative
reference trajectories, with an average RMSE of approximately $1.1\%$.
When the manifold exceeds one ion, it is spread over several ions with a
two-ion exchange between them, and for long chains under this timing
model the projected gain settles at a constant set by the number of
couplings that stay inside one ion.
\end{abstract}

\maketitle


\section{Introduction}
\label{sec:introduction}
Reaching computational regimes beyond classical capabilities requires
increasing processor scale without sacrificing operation fidelity. For
trapped-ion systems, maintaining their exceptional control performance as
the computational register grows remains an important scaling
challenge~\cite{ransford2026, pino2021}. A complementary route enlarges the
computational space within each ion instead, by exploiting its rich internal
level structure~\cite{low2025qudit, ringbauer2022}.

Molecular vibronic dynamics~\cite{macdonell2021, valahu2023}, lattice gauge theories~\cite{davoudi2021, pagano2025}, and spin-boson transport models~\cite{so2024} all pair a small number of electronic or spin states with a large number of bosonic modes. On qubit-only hardware, each bosonic mode must instead be truncated and digitally encoded into a register, and both qubit count and depth grow with the cutoff~\cite{macdonell2021, vu2025}. Hybrid processors avoid this by keeping the modes native and putting the discrete states in a register of their own~\cite{liu2026isa, crane2024, kemper2025}.

Trapped ions are a natural fit. Internal electronic levels are long-lived qudits~\cite{ringbauer2022, low2020, low2025qudit, aksenov2023}, motional modes supply the oscillators~\cite{wineland1998, leibfried2003, sutherland2021}, and the shared motional bus connects any pair of registers without SWAP chains~\cite{molmer1999, sorensen2000}. A molecule has a manifold of electronic states dressed by vibrational modes, while an ion has a manifold of internal states dressed by motional modes. Distributing the electronic manifold across qubits discards that correspondence.

This correspondence has a direct resource consequence. Let
$d_{\rm eff}$ denote the number of internal levels of a chosen ion that are
experimentally usable for coherent computation under the required control
operations. A single-excitation manifold containing $N_c$ excited states and
one reference state can be encoded entirely within one ion whenever
\begin{equation}
N_c+2\le d_{\rm eff},
\end{equation}
where the extra level is an auxiliary sublevel reserved for the
level-selective mode operations (Sec.~\ref{sec:encoding}). Our
$^{40}$Ca$^+$ layout holds at most five excited states per ion.
In that regime, a one-qubit-per-state representation requires $N_c$ 
computational qubits, whereas the qudit representation requires one
computational ion. The resulting advantage is therefore not primarily a
reduction in logical operation count, but more efficient use of the
discrete-variable hardware. As such, a spatially distributed register is replaced by
a spectrally addressed multilevel register within one ion while the bosonic
degrees of freedom remain native. When the manifold does not fit, it is
spread over several ions and the compression becomes a constant factor
per ion (Sec.~\ref{sec:multiion}). This is the resource
principle motivating the hybrid qudit--oscillator instruction set.

Universal control of a bosonic mode was settled by Lloyd and Braunstein, who showed that Gaussian operations plus one non-Gaussian element generate any transformation~\cite{lloyd1999}. The same criterion, phrased through Lie-algebra closure, runs through the control-theoretic literature~\cite{brif2010}. Neither says anything about hardware, or about a discrete register.

Experiments then built those operations in hardware. In circuit QED, dispersive coupling to an ancilla qubit gives universal oscillator control through number-selective phase gates~\cite{krastanov2015, heeres2015} or, faster, through conditional displacements~\cite{eickbusch2022}. On trapped ions, spin-dependent forces coupling an internal transition to the motion supply both the Gaussian and the non-Gaussian operations~\cite{sutherland2021, katz2023, gan2020}, and the same interactions encode logical qubits directly in the motion of a single ion~\cite{fluhmann2019}. Throughout, the discrete register is a qubit.

Qudit control on ions has matured on a separate track. Universal single-qudit operations have been demonstrated on ${}^{40}$Ca$^+$~\cite{ringbauer2022} and ${}^{171}$Yb$^+$~\cite{aksenov2023, zalivako2024}, native two-qudit entangling gates up to $d=5$~\cite{hrmo2023}, and control and readout of higher-dimensional registers on ${}^{137}$Ba$^+$~\cite{low2025qudit, shi2025}. The compilation advantages over qubits are by now well catalogued~\cite{kiktenko2025}. These developments leave open a trapped-ion computational architecture that
combines qudit control with qudit--oscillator operations and establishes the
resulting instruction set as a universal Hamiltonian-generation model
(Fig.~\ref{fig:landscape}).

Liu \textit{et al.} organized hybrid oscillator--qubit processors into instruction set architectures, cataloguing the phase-space, Fock-space, and sideband sets and their costs~\cite{liu2026isa}. The framework is deliberately platform-agnostic, and its primitives, in particular a native beam splitter, come from superconducting cavities. Hybrid oscillator--qudit registers have also been studied from a
coding-theoretic perspective, where stabilizer states, stabilizer codes, and
symplectic operations are constructed for a $d$-level system coupled to a
mode~\cite{chakraborty2026}. That work develops the structure of the hybrid
phase space rather than a hardware instruction set for a specific platform.

A concrete application built on a qubit register is the simulation of
molecular vibronic dynamics, where a few electronic states couple to
vibrational modes~\cite{vu2025}. Ref.~\cite{vu2025} distributes the
electronic manifold across a one-hot qubit register and works in a displaced
frame, where the vibronic coupling takes the form
$\sigma_z^R(\hat a+\hat a^\dagger)$ and matches the conditional displacement
used by that implementation. In the qudit encoding, the complete electronic manifold is stored within one
ion, so $\Pi_R$ is a native level projector, and electronic couplings between
encoded states become local intra-ion transitions. 

\begin{figure}[t]\centering
\includegraphics[width=\columnwidth]{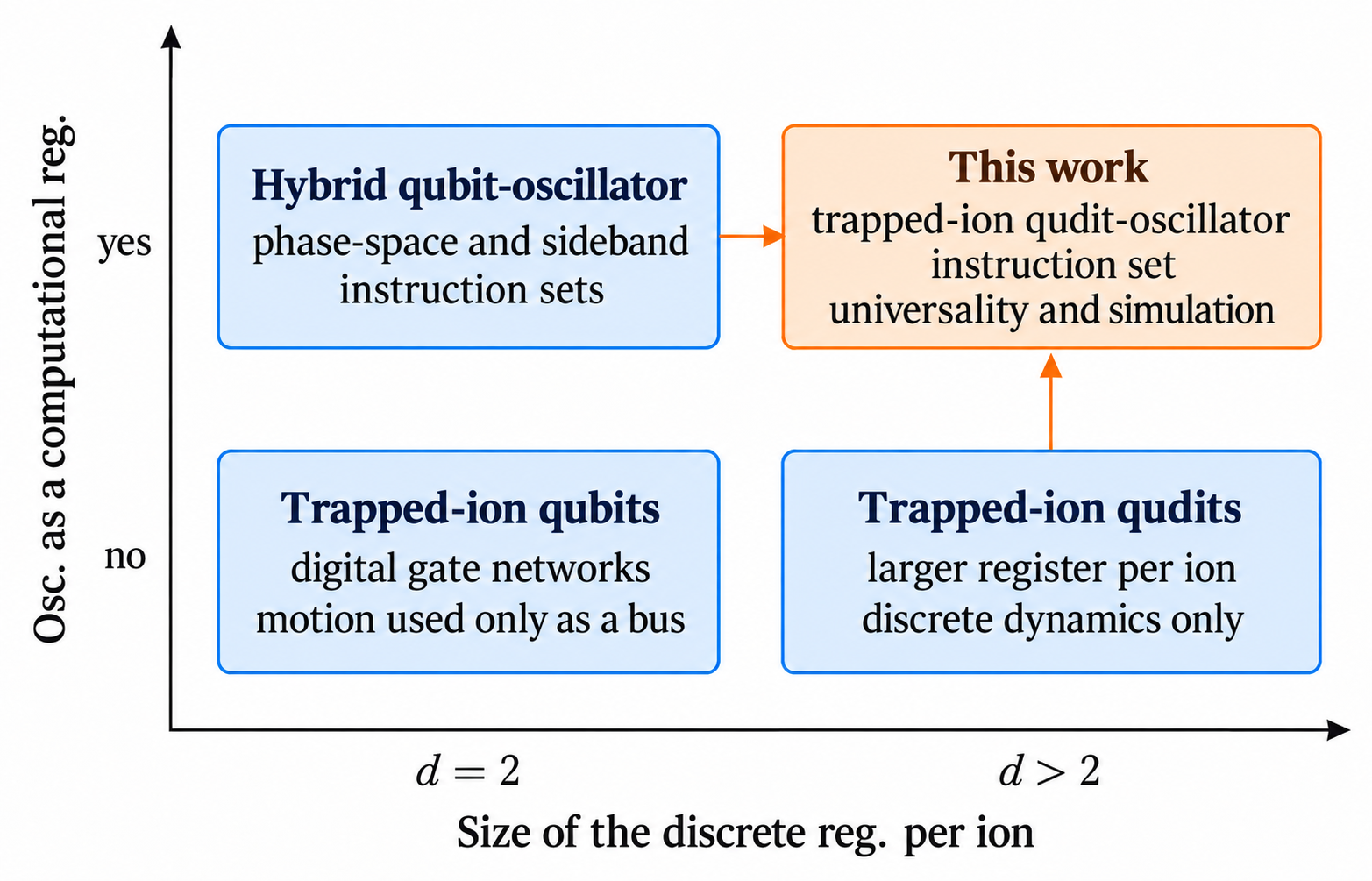}
\caption{Position of this work. Hybrid instruction sets pair a qubit with one
or more oscillators, while trapped-ion qudit control enlarges the discrete
register but keeps the dynamics discrete. This work connects these
directions as follows: A trapped-ion instruction set in which a $d$-level qudit and its
motional modes are both computational registers, with level-selective
couplings between them and a universality guarantee.}
\label{fig:landscape}
\end{figure}

We build the register the molecule already resembles. The primitives are the
operators in which such problems are naturally written as level-selective
phases, transitions between individual levels, and displacements and
rotations of a mode conditioned on a single level. The lab-frame vibronic
Hamiltonian therefore maps onto them term by term, with no displaced-frame
transformation. Five ion-native primitives provide local qudit--oscillator control. Independently of that application, three of them, the qudit transition ($X_E$), qudit phase shift ($Z_d$), and
level-conditioned displacement ($C_{\Pi_z}D$),
suffice to generate Hamiltonians formed from arbitrary qudit operators tensored with polynomial functions of the quadratures of all modes, so multimode universality requires no synthesized operation. An ancilla-mediated synthesized beam splitter is added as an efficient Gaussian mode-mixing primitive, and $C_{\Pi_x}D$ places the phase-space and sideband models of Ref.~\cite{liu2026isa} within the same hardware-native architecture. We implement the architecture in HybridLane~\cite{liu_hybridlane2025, pennylane2018}, recover the dissipative dynamics of the three-chromophore model of Ref.~\cite{vu2025} on a single ququart, and quantify both the benefits and limitations of the qudit register.

Sections~\ref{sec:isa}--\ref{sec:crosscompile} introduce the physical setting, instruction set, universality, and cross-compilation; Sections~\ref{sec:hybridlane}--\ref{sec:results} present the implementation and vibronic validation; Sections~\ref{sec:resources}--\ref{sec:conclusion} analyze resources, limitations, and conclusions.



\section{A Native Hybrid Qudit--Oscillator Instruction Set}
\label{sec:isa}
The instruction set consists of five hardware-native ISA primitives acting on
the qudit and its motional modes, namely the qudit phase shift ($Z_d$), the qudit
transition ($X_E$), the level-conditioned displacement ($C_{\Pi_z}D$), the
level-conditioned mode rotation ($C_{\Pi}R$), and the transition-conditioned displacement ($C_{\Pi_x}D$).
Here, ``hardware-native ISA primitive'' means a logical operation constructed
directly from trapped-ion internal-state and spin--motion controls, rather
than through a digital decomposition into two-qubit gates~\cite{liu2026isa}.
We use ``ISA operation'' for one
logical operation of a listed primitive or of the standard unconditional
mode rotation $R(\theta)=\exp(-i\theta\hat n)$, its pulse-level implementation
is platform-dependent. The beam splitter introduced below is a synthesized
ISA operation rather than a native primitive. The two-ion exchange, used
only when the manifold spans several ions, is described at the end of this
section. Table~\ref{tab:isa} collects the five single-ion primitives and
their physical control mechanisms.

The set instantiates the phase-space instruction set of Liu \textit{et al.}~\cite{liu2026isa}, consisting of the conditional displacement (CD), the mode rotation ($R_\phi$), and the beam splitter (BS), $\{CD,R_\phi,BS\}$,
with the discrete register lifted from a qubit to a qudit. The mode rotation $R_\phi$ is retained as the unconditional rotation $R(\theta)$, $X_E$ and $Z_d$ additionally supply native $\mathrm{SU}(d)$ control of the discrete register, $C_{\Pi_z}D$ is the qudit-level realization of the entangling primitive $CD$, and $BS$ is synthesized from second-order spin--motion interactions~\cite{sutherland2021, gan2020} (Sec.~\ref{sec:bs}). $C_{\Pi}R$, $C_{\Pi_x}D$, and the synthesized beam splitter plays no role in the universality argument of Sec.~\ref{sec:universality}, which rests on $X_E$, $Z_d$, and $C_{\Pi_z}D$ alone. They are kept because $C_{\Pi}R$ and $C_{\Pi_x}D$ compile the dispersive and link-modulated Hamiltonian terms directly (Sec.~\ref{sec:ham}), $C_{\Pi_x}D$ makes the sideband model native (Sec.~\ref{sec:crosscompile}), and the beam splitter provides efficient Gaussian mode mixing compatible with the phase-space instruction set of Ref.~\cite{liu2026isa}.

\begin{table*}[t]
  \centering
  \begin{tabular}{ll@{\hspace{1.2em}}l@{\hspace{1.2em}}l}
    \toprule
    Gate & Generator & Generalized generator & Physical basis \\
    \midrule

    $Z_d(\theta,\ell)$
    & $\ket{\ell}\!\bra{\ell}$
    & $\displaystyle \sum_{\ell}\theta_\ell\ket{\ell}\!\bra{\ell}$
    & virtual-$Z$ \\

    $X_E(\theta,\ell_i,\ell_j)$
    & $\ket{\ell_i}\!\bra{\ell_j}+\mathrm{h.c.}$
    & $\displaystyle
      \sum_{(i,j)\in E}J_{ij}
      \left(\ket{i}\!\bra{j}+\mathrm{h.c.}\right)$
    & carrier drive \\

    $C_{\Pi_z}D(\alpha,\ell)$
    & $i\ket{\ell}\!\bra{\ell}\!\otimes\!(\alpha\hat a^\dagger-\alpha^*\hat a)$
    & $\displaystyle
      i\sum_{\ell\in S}\ket{\ell}\!\bra{\ell}
            \otimes(\alpha\hat a^\dagger-\alpha^*\hat a)$
    & state-dependent force \\

    $C_{\Pi}R(\theta,\ell)$
    & $\ket{\ell}\!\bra{\ell}\otimes\hat n$
    & $\displaystyle
      \sum_{\ell\in S}\ket{\ell}\!\bra{\ell}\otimes\hat n$
    & dispersive sideband shift \\

    $C_{\Pi_x}D(\theta,\ell_i,\ell_j)$
    & $(\ket{\ell_i}\!\bra{\ell_j}+\mathrm{h.c.})\otimes\hat x$
    & $\displaystyle
    \sum_{(i,j)\in E}
    \theta_{ij}
    \left(\ket{i}\!\bra{j}+\mathrm{h.c.}\right)\otimes\hat x$
    & red/blue sidebands \\
    
    \bottomrule
  \end{tabular}

  \caption{Native ISA primitives, their generalized multilevel forms, and physical control mechanisms. Here $E$ denotes the set of simultaneously addressed qudit transitions and $S$ the selected set of qudit levels.}

  \label{tab:isa}
\end{table*}

\subsection{Background}
\label{sec:background}
\emph{Qudits and qumodes.} A qudit is a $d$-level register with basis $\ket{0},\dots,\ket{d-1}$, generalizing the qubit ($d=2$) and carrying the operator algebra $\mathfrak{su}(d)$~\cite{ringbauer2022, low2020, kiktenko2025}. A qumode is a bosonic mode with annihilation operator $\hat a$, number states $\ket{n}$, and quadratures $\hat x=\hat a+\hat a^\dagger$ and $\hat p=i(\hat a^\dagger-\hat a)$~\cite{lloyd1999, chakraborty2026}. The hybrid register occupies the joint space $\mathbb{C}^d\otimes\mathcal{F}$, where $\mathcal{F}$ is the bosonic Fock space (numerical simulations use a finite cutoff), and its operations couple discrete transitions of the qudit to phase-space motion of the mode~\cite{liu2026isa, crane2024}.

\emph{Trapped-ion interactions.} The five ISA primitives are built from
established trapped-ion internal-state and spin--motion controls
~\cite{wineland1998,leibfried2003,low2020,sutherland2021}. A resonant drive
on an internal transition acts as a carrier gate rotating two qudit levels.
A drive tuned to a red or blue motional sideband gives the
Jaynes--Cummings or anti--Jaynes--Cummings coupling between that transition
and a mode~\cite{sutherland2021}. Bichromatic drives provide
state-dependent motional forces~\cite{molmer1999,sorensen2000,low2020}, and the same field configuration yields a $\sigma_z$ force by choice of the beat-note frequency~\cite{bazavan2023},
while far-detuned sideband interactions produce phonon-number-dependent
dispersive shifts~\cite{wineland1998,katz2023,choi2026}.
The following subsections assemble these physical interactions into the
instruction set.

\emph{Universality.} Here, universality refers to Hamiltonian
generation, where the available generators and their repeated
commutators generate Hamiltonians formed from arbitrary qudit
operators tensored with finite polynomials in the oscillator
quadratures~\cite{lloyd1999, brif2010}. In infinite dimensions this
does not mean exact synthesis of every unitary in finite
depth. Section~\ref{sec:universality} establishes the
polynomial-Hamiltonian result, while Sec.~\ref{sec:resources}
addresses the separate question of cost.

\subsection{$Z_d$: Qudit Phase Shift}

$Z_d$ applies a controllable phase to a selected qudit level,
\begin{equation}
Z_d(\theta,\ell) = \exp\!\left(-i\theta\,\ket{\ell}\!\bra{\ell}\right),
\end{equation}
leaving every other level untouched. It is the qudit generalization of the qubit virtual-$Z$ gate~\cite{mckay2017}, implemented as a virtual phase update on the drives addressing level $\ell$ rather than as a physical transition. The $Z_d$ projectors span the diagonal sector of $\mathfrak{u}(d)$, while their traceless differences span the diagonal sector of $\mathfrak{su}(d)$.

More generally, independent phase shifts on several qudit levels can be
combined exactly into a single diagonal operation,
\begin{equation}
Z_d(\boldsymbol{\theta})
=
\exp\!\left[
-i\sum_{\ell=0}^{d-1}\theta_\ell
\ket{\ell}\!\bra{\ell}
\right].
\end{equation}
As $Z_d$ is realized virtually, the generalized operation is implemented by
coordinated phase updates of the frequency components addressing the relevant
transitions in the subsequent multitone control waveform, rather than by
additional physical pulses.

\subsection{$X_E$: Qudit Transition}
$X_E$ coherently couples two levels of the qudit,
\begin{equation}
X_E(\theta,\ell_i,\ell_j) = \exp\!\left[-i\theta\left(\ket{\ell_i}\!\bra{\ell_j}+\ket{\ell_j}\!\bra{\ell_i}\right)\right],
\end{equation}
driving a Rabi rotation within the $\{\ket{\ell_i},\ket{\ell_j}\}$ subspace while leaving the remaining levels untouched. This is realized by a resonant carrier drive on the corresponding internal transition~\cite{ringbauer2022, low2020}, with $\theta$ set by the pulse area and the level pair selected by the drive frequency. $X_E$ supplies the off-diagonal generators. Together with $Z_d$, it closes on $\mathfrak{su}(d)$, giving native single-qudit control~\cite{aksenov2023} (Sec.~\ref{sec:sud}).

More generally, multiple qudit transitions can be driven simultaneously by a
multitone carrier field,
\begin{equation}
X_{E}(\boldsymbol{\theta})
=
\exp\!\left[
-i\sum_{(i,j)\in E}
\theta_{ij}
\left(
\ket{i}\!\bra{j}+\ket{j}\!\bra{i}
\right)
\right],
\end{equation}
where each addressed transition is selected by its frequency, while its
rotation angle and phase are controlled by the corresponding tone.
When the addressed transitions share a level, as $A\!\leftrightarrow\!B$ and
$A\!\leftrightarrow\!C$ do, their generators do not commute, so the summed exponential above is not
equal to the product of the two sequential $X_E$ exponentials. In that case
the multitone waveform is synthesized to reproduce the desired sequential
$X_E$ unitary, so the block is a compiled control primitive rather than an
algebraic identity. Compatible $X_E$ operations may then be scheduled within a
single carrier-control block.

\subsection{$C_{\Pi_z}D$: Level-Conditioned Displacement}
$C_{\Pi_z}D$ applies a phase-space displacement of a mode conditioned on the qudit occupying a chosen level,
\begin{equation}
C_{\Pi_z}D(\alpha,\ell) = \exp\!\left[\ket{\ell}\!\bra{\ell}\otimes\left(\alpha\,\hat a^\dagger-\alpha^*\hat a\right)\right],
\end{equation}
displacing the mode by $\alpha$ when the qudit is in $\ket{\ell}$ and leaving
it fixed otherwise. Generalized oscillator--qudit conditional displacements
have been studied directly in hybrid qudit--oscillator settings
~\cite{evenhaim2025,chakraborty2026}. In trapped ions, multilevel
M{\o}lmer--S{\o}rensen-type interactions provide state-dependent motional
forces~\cite{molmer1999,sorensen2000,low2020}, the same bichromatic field yields the $\sigma_z$ force used here when its beat note is set to the motional frequency~\cite{bazavan2023}, while
projector-conditioned motional interactions have also been demonstrated
experimentally~\cite{gan2020}. These results establish the physical
ingredients for $C_{\Pi_z}D$. For a specific multilevel encoding, however, obtaining
the exact single-level projector $\Pi_\ell=\ket{\ell}\!\bra{\ell}$ can require
a calibrated combination of level-selective controls rather than a single
bichromatic tone. At the ISA level, $C_{\Pi_z}D$ is the projector-conditioned
operation defined above and is the qudit-level realization of the conditional
displacement $CD$ that anchors Liu's phase-space instruction set
~\cite{liu2026isa}. Applied to a qudit prepared in a superposition of levels,
it entangles the discrete register with the bosonic mode and supplies the
non-Gaussian resource used in the universality construction of
Sec.~\ref{sec:entangle}.

The conditioning can be generalized from a single level to an arbitrary set
$S$ of qudit levels,
\begin{equation}
C_{\Pi_z}D(\alpha,S)
=
\exp\!\left[
\left(
\sum_{\ell\in S}\ket{\ell}\!\bra{\ell}
\right)
\otimes
\left(
\alpha\hat a^\dagger-\alpha^*\hat a
\right)
\right].
\end{equation}
When the same displacement is conditioned on several levels of the same mode,
the corresponding rank-one $C_{\Pi_z}D$ operations can therefore be combined into a
single higher-rank conditional displacement.

\subsection{$C_{\Pi}R$: Level-Conditioned Mode Rotation}
$C_{\Pi}R$ rotates a mode in phase space by an angle proportional to its occupation number, at a rate set by the qudit level,
\begin{equation}
C_{\Pi}R(\theta,\ell) = \exp\!\left(-i\theta\,\ket{\ell}\!\bra{\ell}\otimes\hat n\right),
\qquad \hat n=\hat a^\dagger\hat a,
\end{equation}
acting only when the qudit occupies $\ket{\ell}$. Its physical basis is the
phonon-number-dependent dispersive shift produced by a far-detuned motional
sideband~\cite{wineland1998,katz2023,choi2026}. In a multilevel encoding,
obtaining exact $\Pi_\ell$ selectivity can require a calibrated combination of
level-selective dispersive controls rather than a single detuned tone. At the
ISA level, $C_{\Pi}R$ denotes the projector-conditioned operation defined above. Its
role in this work is to compile the dispersive terms of the vibronic
Hamiltonian directly (Sec.~\ref{sec:ham}).

The conditioning can likewise be generalized from a single level to an
arbitrary set $S$ of qudit levels,
\begin{equation}
C_{\Pi}R(\theta,S)
=
\exp\!\left[
-i\theta
\left(
\sum_{\ell\in S}\ket{\ell}\!\bra{\ell}
\right)
\otimes\hat n
\right].
\end{equation}
When the same mode rotation is conditioned on several qudit levels, the
corresponding rank-one $C_{\Pi}R$ operations can therefore be combined into a single
higher-rank conditional rotation.

\subsection{$C_{\Pi_x}D$: Transition-Conditioned Displacement}
$C_{\Pi_x}D$ couples a qudit transition directly to a bosonic mode. Here
$\Pi_x$ denotes the transition operator $X_{ij}=\ket{i}\!\bra{j}+\ket{j}\!\bra{i}$,
the off-diagonal counterpart of the level projector $\Pi_\ell$, so the mode
is displaced in opposite directions for the two eigenstates of the
transition. The operation reads
\begin{align}
C_{\Pi_x}D(\theta,\ell_i,\ell_j) &= \exp\!\left[-i\theta\left(\ket{\ell_i}\!\bra{\ell_j}+\ket{\ell_j}\!\bra{\ell_i}\right)\otimes\hat x\right], \nonumber \\
\hat x &=\hat a+\hat a^\dagger,
\end{align}
coupling the $\{\ket{\ell_i},\ket{\ell_j}\}$ transition to the mode. Because $\hat x=\hat a+\hat a^\dagger$, $C_{\Pi_x}D$ naturally contains both the Jaynes--Cummings and anti-Jaynes--Cummings interactions, corresponding to the red and blue sidebands, whose combination implements the quantum Rabi interaction~\cite{sutherland2021, gan2020}. It is realized by simultaneous red and blue sideband drives on the chosen transition. Driving only the red sideband, in contrast, yields the excitation-conserving Jaynes--Cummings interaction alone, as required by the sideband instruction set.

For the optical $S_{1/2}\!\leftrightarrow\!D_{5/2}$ transitions of the encoding
below, the spin--motion coupling is supplied directly by the $729$~nm
sideband drives~\cite{meth2025}. At the operating angle of the vibronic
circuit, $C_{\Pi_x}D$ is realized as a single flat bichromatic pulse
(Appendix~\ref{app:pulses}). At large rotation angles the co-driven carrier
becomes significant and additional compensation is required, which we do
not optimize here.

By making sideband interactions native alongside the phase-space primitives, $C_{\Pi_x}D$ unifies the two computational models within a single hardware-native instruction set, enabling the cross-compilation developed in Sec.~\ref{sec:crosscompile}.

The same construction can be extended to multiple qudit transitions coupled
to the same bosonic mode,
\begin{equation}
C_{\Pi_x}D_{E}(\boldsymbol{\theta})
=
\exp\!\left[
-i\sum_{(i,j)\in E}
\theta_{ij}
\left(
\ket{i}\!\bra{j}+\ket{j}\!\bra{i}
\right)\otimes\hat x
\right].
\end{equation}
This multitransition form corresponds to simultaneously addressing the
red- and blue-sideband components of the selected qudit transitions with
independently controlled tones.

\subsection{The Synthesized Beam Splitter}
The phase-space instruction set $\{CD, R_\phi, BS\}$ of
Ref.~\cite{liu2026isa} additionally lists an unconditional beam
splitter coupling two modes. As shown in Sec.~\ref{sec:universality},
the beam splitter is not required for universality of this
architecture, which follows from $X_E$, $Z_d$, and $C_{\Pi_z}D$
alone. We nevertheless provide it as an efficient Gaussian mode-mixing
primitive, so that programs written for the phase-space set run
without modification. It is synthesized from spin--motion
interactions, i.e., two state-dependent-force interactions driven on
modes $j\neq j'$ with equal detuning combine at second order into a
spin-dependent beam-splitter Hamiltonian, following Sutherland and
Srinivas~\cite{sutherland2021, gan2020}. A mediator ancilla ion
sharing the motional bus, prepared in a $\sigma_3$ eigenstate, carries
the conditioning and yields the unconditional beam splitter without
disturbing the computational register (Sec.~\ref{sec:bs}).

\subsection{The Two-Ion Exchange}
When the electronic manifold spans more than one ion, a coupling between
levels on different ions is needed. It moves the excitation from one ion
to the other and is built from the two-qudit M{\o}lmer--S{\o}rensen
interaction demonstrated up to $d=5$ in Ref.~\cite{hrmo2023}, acting on
one level pair of each ion. It is a
native trapped-ion operation and plays no role in the single-ion circuits
of Secs.~\ref{sec:methods} and~\ref{sec:results}. Its cost and its role in
universality are given in Sec.~\ref{sec:multiion}.


\section{Native Trapped-Ion Realization}
\label{sec:amm}
We now connect the instruction set to its trapped-ion hardware realization, starting with the physical hybrid register and the control mechanisms underlying each primitive~\cite{ringbauer2022, low2020, wineland1998, leibfried2003}, then introducing the vibronic Hamiltonian used as the benchmark throughout the rest of the paper~\cite{vu2025}.

\subsection{Register and Encoding}
\label{sec:encoding}
The hybrid register consists of a single $d$-level qudit and $M$ bosonic qumodes. The qudit is carried by the internal electronic levels of one trapped ion, and the qumodes are collective motional modes of the ion crystal~\cite{wineland1998}. Because the addressed internal transitions couple to the shared motional modes, the qudit interacts directly with the modes of the register, eliminating the SWAP operations that mediate interactions on nearest-neighbor architectures.

We encode the qudit in a $^{40}$Ca$^+$ ion on the $S_{1/2}\rightarrow D_{5/2}$
electric-quadrupole transition at $729$~nm, the same level structure used for
universal qudit control and for qudit sideband gates in
Refs.~\cite{ringbauer2022, meth2025}. A static magnetic field of a few Gauss
splits $S_{1/2}$ into two and $D_{5/2}$ into six Zeeman sublevels, giving ten
laser addressable transitions with a metastable lifetime of about $1$~s
\cite{ringbauer2022}. In the vibronic encoding we assign
$\ket{A}=S_{1/2}(m{=}-1/2)$ and place $\ket{B}$, $\ket{C}$, and $\ket{G}$
in three Zeeman sublevels of $D_{5/2}$. The $A\!\leftrightarrow\!B$ and
$A\!\leftrightarrow\!C$ transitions are then direct optical transitions
addressed by the drive frequency. $C_{\Pi_x}D$ acts on these same transitions.
$C_{\Pi_z}D$ and $C_{\Pi}R$ need a projector on a single level, $\Pi_\ell$, and a
sideband force between two code levels would give $\Pi_\ell-\Pi_{\ell'}$
instead. They are therefore driven between the addressed level and an
unused auxiliary sublevel $\ket{a}$ outside the code space, so that
$\sigma_z^{(\ell,a)}$ restricted to the code space is exactly $\Pi_\ell$.
This holds because the tones sit near the $\ell\!\leftrightarrow\!a$
transition, so the other code levels see only off-resonant shifts
suppressed by the ratio of the drive detuning to their spectral
distance.
The auxiliary level is counted in the fit condition $N_c+2\le d_{\rm eff}$.
On this platform our layout holds at most five excited states per ion,
and the sideband gates of Ref.~\cite{meth2025} use the same auxiliary
level construction. Each
transition is resolved by the Zeeman splitting, so
$X_E$ and $Z_d$ act selectively on chosen level pairs. The pulse and timing
parameters used in Appendices~\ref{app:pulses} and~\ref{app:runtime} are
taken from this platform.

In the vibronic application (Sec.~\ref{sec:methods}), each electronic state of the coupled-chromophore complex maps to one qudit level, so the entire electronic manifold sits on a single ion rather than being distributed across qubits.

\subsection{Qudit Platforms Across Ion Species}
\label{sec:platforms}

The instruction set is not tied to a particular ion species or qudit
dimension. We define the experimentally usable qudit dimension
$d_{\rm eff}$ as the number of internal states that can be coherently
initialized, selectively addressed, coupled to the motional modes through the
required primitives, and read out at the fidelity required by the
computation. This distinction is important as the total number of atomic levels
available in an ion does not by itself determine the useful computational
dimension.

Universal qudit control has been demonstrated in ${}^{40}$Ca$^+$
and ${}^{171}$Yb$^+$~\cite{ringbauer2022,aksenov2023}, native two-qudit
entangling gates up to $d=5$~\cite{hrmo2023}, while higher-dimensional
control and readout have been demonstrated in
${}^{137}$Ba$^+$~\cite{low2025qudit,shi2025}. For the single-excitation
encoding considered here, a single ion can contain the full electronic
manifold whenever
\begin{equation}
d_{\rm eff}\ge N_c+2,
\end{equation}
where $N_c$ is the number of excited electronic states and one level is
reserved as the auxiliary sublevel of Sec.~\ref{sec:encoding}. 
The ${}^{40}$Ca$^+$ ion used in our benchmark, with two $S_{1/2}$ and six
$D_{5/2}$ Zeeman sublevels, accommodates the three-chromophore manifold with
room to spare, and ${}^{171}$Yb$^+$ ququarts~\cite{aksenov2023, zalivako2024}
offer an alternative optical encoding with the same instruction set.

Larger $d_{\rm eff}$ extends the largest
discrete manifold that can be represented within one computational ion
without changing the instruction set.

When $N_c+2>d_{\rm eff}$, the discrete manifold is spread over several
ions, and couplings between ions use the two-ion exchange of
Sec.~\ref{sec:isa}. The cost of this split is given in
Sec.~\ref{sec:multiion}.

\subsection{The Vibronic Benchmark}
\label{sec:ham}
To demonstrate the instruction set on a physically relevant problem,
we introduce the vibronic model used as our validation benchmark and
show that each of its terms maps onto an ISA primitive. The system
follows Vu \textit{et al.}~\cite{vu2025}, where three chromophores
$A,B,C$, each a two-level electronic system coupled to a localized
vibrational mode, plus a low-frequency link mode of frequency
$\omega_l$ coupled to chromophore $A$~\cite{macdonell2021}. In the lab
frame,

\begin{align}
\label{eq:H_lab}
H_{\rm lab} \;&=\; \ket{G}\!\bra{G}\otimes(\hat h_A^g+\hat h_B^g+\hat h_C^g) \nonumber \\
             &+ \ket{A}\!\bra{A}\otimes(\hat h_A^e+\hat h_B^g+\hat h_C^g) \nonumber\\
           &+\; \ket{B}\!\bra{B}\otimes(\hat h_A^g+\hat h_B^e+\hat h_C^g) \nonumber \\
             &+ \ket{C}\!\bra{C}\otimes(\hat h_A^g+\hat h_B^g+\hat h_C^e) \nonumber\\
           &+\; J_{AB}\,\bigl(\ket{A}\!\bra{B}+\mathrm{h.c.}\bigr) 
             + J_{AC}\,\bigl(\ket{A}\!\bra{C}+\mathrm{h.c.}\bigr) ,
\end{align}

where $\hat h_R^g$ and $\hat h_R^e$ are the ground- and excited-state vibrational Hamiltonians of chromophore $R$. For chromophore $A$, which couples to both its local mode $a$ and the link mode $l$,
\begin{align}
\hat h_A^g &= \hbar\omega_{g,a}\,\hat a^\dagger\hat a
            + \hbar\omega_l\,\hat l^\dagger\hat l, \\
\hat h_A^e &= \hbar\omega_{e,a}\!\left[\hat a^\dagger\hat a+S_a-\sqrt{S_a}\,(\hat a^\dagger+\hat a)\right] \nonumber\\
           &\quad + \hbar\omega_l\!\left[\hat l^\dagger\hat l+S_l-\sqrt{S_l}\,(\hat l^\dagger+\hat l)\right].
\end{align}
For chromophores $B$ and $C$ (with $R=B,C$ and local mode $r=b,c$),
\begin{align}
\hat h_R^g &= \hbar\omega_{g,r}\,\hat r^\dagger\hat r, \\
\hat h_R^e &= \hbar\omega_{e,r}\!\left[\hat r^\dagger\hat r+S_r-\sqrt{S_r}\,(\hat r^\dagger+\hat r)\right].
\end{align}
Here $S_a,S_l,S_r$ are the Huang--Rhys factors setting the vibronic coupling strengths, and the excitation-transfer couplings are
\begin{equation}
J_{AR}=\hbar J_{AR,0}\!\left[1+\eta_{AR}(\hat l^\dagger+\hat l)\right],\qquad R=B,C,
\end{equation}
where $\eta_{AR}$ is a first-order coupling constant. The modulation of $J_{AR}$ by the low-frequency mode position $(\hat l+\hat l^\dagger)$ reflects the natural influence of vibronic coupling.

Each term of $H_{\rm lab}$ maps onto an ISA primitive. The reorganization
phases become $Z_d$, the free-oscillator mode terms unconditional rotations, the
dispersive shifts $C_{\Pi}R$, the linear Huang--Rhys displacements $C_{\Pi_z}D$, the static
electronic couplings $X_E$, and the link-modulated couplings $C_{\Pi_x}D$. The
correspondence is exact and term-by-term at the ISA level, so the compilation
of the lab-frame model introduces no frame or operator approximation. The
numerical evolution then uses Trotterization and a finite Fock cutoff. The
pulse-level realization of each ISA primitive constitutes a separate
hardware-compilation layer.


\section{Universality}
\label{sec:universality}
Next, we show that the instruction set is universal in the sense of Hamiltonian generation.

The argument rests on one identity. If $e^{-iAt}$ and $e^{-iBt}$ are available, then so is the evolution generated by their commutator,
\begin{equation}
\label{eq:groupcomm}
e^{-iA\epsilon}e^{-iB\epsilon}e^{iA\epsilon}e^{iB\epsilon}
= e^{-[A,B]\epsilon^{2}} + O(\epsilon^{3}),
\end{equation}
so new generators can be synthesized perturbatively by interleaving the evolutions already in hand. Universality reduces to exhibiting native generators whose commutators close on the algebra of the joint system.

Three primitive types suffice. $X_E$ and $Z_d$ generate $\mathfrak{su}(d)$ on the discrete register, and $C_{\Pi_z}D$, applied to each mode, couples the register to the modes and supplies the non-Gaussian resource. Cross-mode terms arise from commutators of $C_{\Pi_z}D$ generators on different modes carrying different qudit factors, so no mode-coupling primitive is required. $C_{\Pi}R$, $C_{\Pi_x}D$, and the synthesized beam splitter play no role in this argument. The two pillars and the multimode extension are established below, with the full induction deferred to Appendix~\ref{app:universality}.

\subsection{Pillar 1: Qudit Control from $X_E$ and $Z_d$}
\label{sec:sud}
The first pillar is complete control of the qudit, namely the ability to generate any element of $\mathrm{SU}(d)$. The $X_E$ primitives supply the off-diagonal generators $X_{ij}=\ket{i}\!\bra{j}+\ket{j}\!\bra{i}$ for each directly driven level pair, and a connected graph of driven transitions together with the phase control of $Z_d$ already generates the remaining pairs by commutation, so the full set of pairs need not be addressed directly, which matches the $^{40}$Ca$^+$ encoding of Sec.~\ref{sec:encoding} where only $S_{1/2}$ to $D_{5/2}$ transitions are laser driven, while $Z_d$ supplies the level projectors $\Pi_\ell=\ket{\ell}\!\bra{\ell}$, whose traceless differences $\Pi_\ell-\Pi_{\ell'}$ span the traceless diagonal sector of $\mathfrak{su}(d)$. The antisymmetric generators are obtained from a single commutator:
\begin{equation}
\label{eq:comm_XPi}
[\,X_{ij},\,\Pi_i\,]
= \ket{j}\!\bra{i}-\ket{i}\!\bra{j}
= i\,Y_{ij},
\qquad
Y_{ij}=i\bigl(\ket{i}\!\bra{j}-\ket{j}\!\bra{i}\bigr).
\end{equation}
The set $\{X_{ij},Y_{ij}\}_{i<j}$ together with $d-1$ independent traceless projector differences spans the generalized Gell-Mann basis of $\mathfrak{su}(d)$. The remaining identity combination of the $Z_d$ projectors contributes only a global phase, so $X_E$ and $Z_d$ generate all of $\mathrm{SU}(d)$ up to that physically irrelevant phase. This is the standard construction underlying universal single-qudit control on trapped-ion platforms, where carrier drives and virtual-$Z$ phase updates are native operations~\cite{ringbauer2022, low2020, low2025qudit}.

For the ququart encoding used throughout this work ($d=4$), the construction gives six $X_{ij}$ generators, six $Y_{ij}$ generators, and three independent traceless diagonal combinations, spanning the $15$-dimensional algebra $\mathfrak{su}(4)$. The additional identity combination of the four $Z_d$ projectors contributes only a global phase.

\subsection{Pillar 2: Universal Single-Mode Control from $C_{\Pi_z}D$}
\label{sec:entangle}
For the second pillar, fix two qudit levels $\ket{a},\ket{b}$ and write $\sigma_1,\sigma_2,\sigma_3$ for the Pauli operators on that two-dimensional subspace, embedded in $\mathbb{C}^d$ and acting as zero on the remaining levels, with Pillar~1 supplying all three natively. Composing two conditional displacements with opposite arguments,
\begin{equation}
\label{eq:cd_sym}
C_{\Pi_z}D(\alpha,\ell_a)\,C_{\Pi_z}D(-\alpha,\ell_b)
= \exp\!\left[\sigma_3\otimes\bigl(\alpha\hat a^\dagger-\alpha^*\hat a\bigr)\right],
\end{equation}
provides the symmetric conditional displacement, the entangling primitive of the phase-space instruction set~\cite{liu2026isa, sutherland2021, eickbusch2022}. Considering the real part of  $\alpha$ supplies the generator $\sigma_3\otimes\hat p$ while the imaginary part gives $\sigma_3\otimes\hat x$.

Conditional displacements alone generate nothing new. Since $\sigma_3$ commutes with itself, the commutator of two of them collapses to a $c$-number,
\begin{equation}
\label{eq:comm_cd}
[\,\sigma_3\otimes\hat x,\;\sigma_3\otimes\hat p\,]
=\Pi_{ab}\otimes[\hat x,\hat p]=2i\,\Pi_{ab}\otimes\Id,
\end{equation}
with $\Pi_{ab}=\ket{a}\!\bra{a}+\ket{b}\!\bra{b}$. New generators arise from commuting the conditional displacements against the qudit rotations of Pillar~1, which change the Pauli factor. Using $[\sigma_3,\sigma_1]=2i\sigma_2$,
\begin{equation}
\label{eq:chain1}
[\,\sigma_3\otimes\hat x,\;\sigma_1\otimes\Id\,]=2i\,\sigma_2\otimes\hat x ,
\end{equation}
and now the two available generators carry \emph{different} Pauli factors, which anticommute. Their commutator is quadratic in the quadratures:
\begin{equation}
\label{eq:chain2}
[\,\sigma_2\otimes\hat x,\;\sigma_3\otimes\hat p\,]
=i\,\sigma_1\otimes\{\hat x,\hat p\} .
\end{equation}
The construction iterates. Each further commutator against a conditional displacement raises the polynomial order in $\hat x$ and $\hat p$ by one, while cycling the qudit factor through the Pauli operators, so the accessible algebra contains operators of arbitrary polynomial degree in the quadratures. It is therefore not confined to the Gaussian sector that displacement and rotation alone would generate. Appendix~\ref{app:universality} carries out the induction, extends the qudit factor from $\mathfrak{su}(2)$ to all of $\mathfrak{su}(d)$, and treats the identity sector.

We also verified the construction numerically. Truncating the mode at
Fock cutoff $N$ and seeding with the six ququart $X_E$ generators
$X_{ij}\otimes\Id$ $(i<j)$, the four $Z_d$ projectors
$\Pi_\ell\otimes\Id$, and the two conditional-displacement generators
$\sigma_3\otimes\hat x$ and $\sigma_3\otimes\hat p$, the generated
algebra grows through the cutoff-independent rank sequence
$12, 29, 105, 375$ over iterations $0$--$3$, identical for
$N=6,8,10$. The measured Fock-basis bandwidth of the generated
operators reaches three. Bandwidth means how many Fock levels apart an
operator can directly connect. Each $\hat a$ or $\hat a^\dagger$ moves
the Fock index by one, so a polynomial of degree $k$ in $\hat x$ and
$\hat p$ can connect levels at most $k$ apart. A bandwidth of three
therefore requires a degree of at least three. Gaussian generators
have a degree of at most two, so these operators lie outside the Gaussian
sector. The finite truncated matrices then reach $\mathfrak{u}(4N)$ at
iteration $4$. Because the $Z_d$ seed set contains traceful level
projectors, the finite-dimensional closure is $\mathfrak{u}(4N)$
rather than $\mathfrak{su}(4N)$, and its dimension $(4N)^2$
necessarily depends on the cutoff. This saturation is a consistency
check on the truncated representation, not the universality claim
itself, which is established analytically in
Appendix~\ref{app:universality}. As a control, replacing the
conditional displacements with unconditional ones, $\Id\otimes\hat x$
and $\Id\otimes\hat p$, keeps the measured bandwidth at one,
consistent with the algebra remaining Gaussian.

The same interaction is what entangles the two registers. Acting on a qudit prepared in a superposition, Eq.~\eqref{eq:cd_sym} correlates the discrete state with the phase space of the mode. Disentangling the qudit afterwards leaves the mode in a superposition of coherent states, so cat states and more general non-classical oscillator states are directly accessible~\cite{gan2020, fluhmann2019}.

\subsection{Multimode Coupling Without a Beam Splitter}
\label{sec:multimode}
Universality extends from a single mode to the full register without any additional primitive. Applying Pillar~2 to each mode separately places every generator $\sigma_k\otimes\hat x_j$ and $\sigma_k\otimes\hat p_j$ in the algebra. Commuting two such generators on different modes with different Pauli factors then produces cross-mode products directly,
\begin{equation}
\label{eq:crossmode}
[\,\sigma_1\otimes\hat x_j,\;\sigma_2\otimes\hat x_{j'}\,]
=2i\,\sigma_3\otimes\hat x_j\hat x_{j'} ,
\end{equation}
since the quadratures of distinct modes commute and the Paulis anticommute. Iterating over quadratures and degrees generates arbitrary joint polynomials (Appendix~\ref{app:universality}), so $X_E$, $Z_d$, and $C_{\Pi_z}D$ alone are universal for the multimode register.

\subsection{The Synthesized Beam Splitter}
\label{sec:bs}
Although not required for universality, an unconditional beam splitter is a useful Gaussian mode-mixing operation, and it completes compatibility with the phase-space instruction set of Ref.~\cite{liu2026isa}. It is synthesized from the same spin--motion interactions that underlie $C_{\Pi_z}D$. Following the trapped-ion construction of Sutherland and Srinivas~\cite{sutherland2021}, two spin-dependent-force interactions applied to modes $j\neq j'$ with equal detuning combine to second order in the Magnus expansion into the spin-dependent beam-splitter interaction
\begin{equation}
\label{eq:bs}
H_{\rm BS} \simeq i\,\Omega\,\sigma_3
\bigl(\hat a_j\hat a_{j'}^\dagger e^{i\phi}-\hat a_j^\dagger\hat a_{j'}e^{-i\phi}\bigr),
\end{equation}
where $\sigma_3$ acts on the chosen two-level subspace of the qudit and both constituent interactions are of the $C_{\Pi_z}D$ type. To obtain the unconditional beam splitter $\exp(\theta[\hat a_j\hat a_{j'}^\dagger e^{i\phi}-\mathrm{h.c.}])$ without disturbing the computational register, the conditioning is carried by a mediator ancilla ion sharing the motional bus. The ancilla is prepared in a $\sigma_3$ eigenstate and the spin-dependent force of Ref.~\cite{sutherland2021} is driven on it, leaving the computational qudit untouched. The mediator is required only when the beam splitter itself is invoked, while the vibronic circuits of Sec.~\ref{sec:methods} use no beam splitter. The construction requires only native spin--motion interactions and has been shown to achieve high-fidelity beam-splitter operations in trapped-ion systems~\cite{sutherland2021, gan2020}.

\subsection{Closure}
Pillar~1 supplies every generator $\Lambda\otimes\Id$ with
$\Lambda\in\mathfrak{su}(d)$. Pillar~2 supplies a generator on each mode separately every $\Lambda\otimes P(\hat x_j,\hat p_j)$ with $P$ an arbitrary polynomial. Their cross-mode commutators, Eq.~\eqref{eq:crossmode}, spread the construction across the full register. Together with repeated commutators, these generators produce Hamiltonians formed from arbitrary qudit operators tensored with polynomials in the quadratures of all oscillator modes, using only the native primitives $X_E$, $Z_d$, and $C_{\Pi_z}D$. The architecture is therefore universal in the sense of Hamiltonian generation, and the proof is given in Appendix~\ref{app:universality}.


\section{Cross-Compilation Between Phase-Space and Sideband Models}
\label{sec:crosscompile}
The architecture supports two of Liu's universal instruction sets~\cite{liu2026isa} using trapped-ion operations, with beam-splitter mode mixing synthesized when required. The phase-space set is anchored on the conditional displacement $C_{\Pi_z}D$, while the sideband set is anchored on the Jaynes--Cummings interaction. On trapped ions both are supported within the same physical control architecture, the latter through $C_{\Pi_x}D$, whose generator $(\ket{\ell_i}\!\bra{\ell_j}+\mathrm{h.c.})\otimes\hat x$ contains the Jaynes--Cummings and anti-Jaynes--Cummings terms~\cite{sutherland2021}. Programs written for either model therefore run without leaving the native gate set.

Because both instruction sets are supported within the same hardware architecture, programs can be translated directly between them. Following the equivalence established by Liu \textit{et al.} (Appendix~D.2)~\cite{liu2026isa}, a conditional displacement can be assembled from sideband interactions interleaved with qudit rotations, and conversely a sideband coupling from conditional displacements and rotations, so a circuit expressed in one instruction set can be rewritten in the other using only native primitives. The choice between the two computational models therefore becomes an algorithmic decision rather than a hardware constraint. Cross-compilation is not an interface between heterogeneous hardware platforms, but an internal compiler transformation within a single native trapped-ion architecture.


\section{HybridLane Implementation}
\label{sec:hybridlane}
We extend HybridLane~\cite{liu_hybridlane2025}, a PennyLane-based framework~\cite{pennylane2018}, from hybrid oscillator--qubit circuits to native hybrid qudit--oscillator computation. The extension introduces the five native primitives as first-class operations parameterized by the addressed qudit level, or level pair, which makes the implementation independent of the chosen qudit dimension. Circuits are written directly in $Z_d$, $X_E$, $C_{\Pi_z}D$, $C_{\Pi}R$, and $C_{\Pi_x}D$ and executed on the framework's bosonic backend.

To validate the approach, we compile the lab-frame Hamiltonian of Sec.~\ref{sec:ham} into the native primitives and evolve it on the extended HybridLane. Accuracy is checked in two ways: (1) against direct numerical integration of $H_{\rm lab}$, which isolates the Trotter error, and (2) against the published reference trajectories of Ref.~\cite{vu2025}, which run the displaced-frame Hamiltonian on a qubit register.


\section{Application: Dissipative Vibronic Dynamics}
\label{sec:methods}
We now validate the proposed architecture on a concrete simulation problem, namely the dissipative vibronic dynamics of a coupled-chromophore complex~\cite{macdonell2021, valahu2023}. When light is absorbed by such a complex, the resulting electronic excitation migrates between chromophores while coupled to vibrational modes and to a dissipative environment~\cite{so2024}. Capturing this transport requires the electronic, vibrational, and dissipative degrees of freedom to be treated together. This makes it both a natural target for a hybrid processor and a demanding benchmark as the dynamics couple electronic excitations to multiple bosonic modes and incorporate open-system dissipation.

The three-chromophore benchmark is illustrated in
Fig.~\ref{fig:vibronic_model}.

\begin{figure}[t]
    \centering
    \includegraphics[width=0.95\columnwidth]{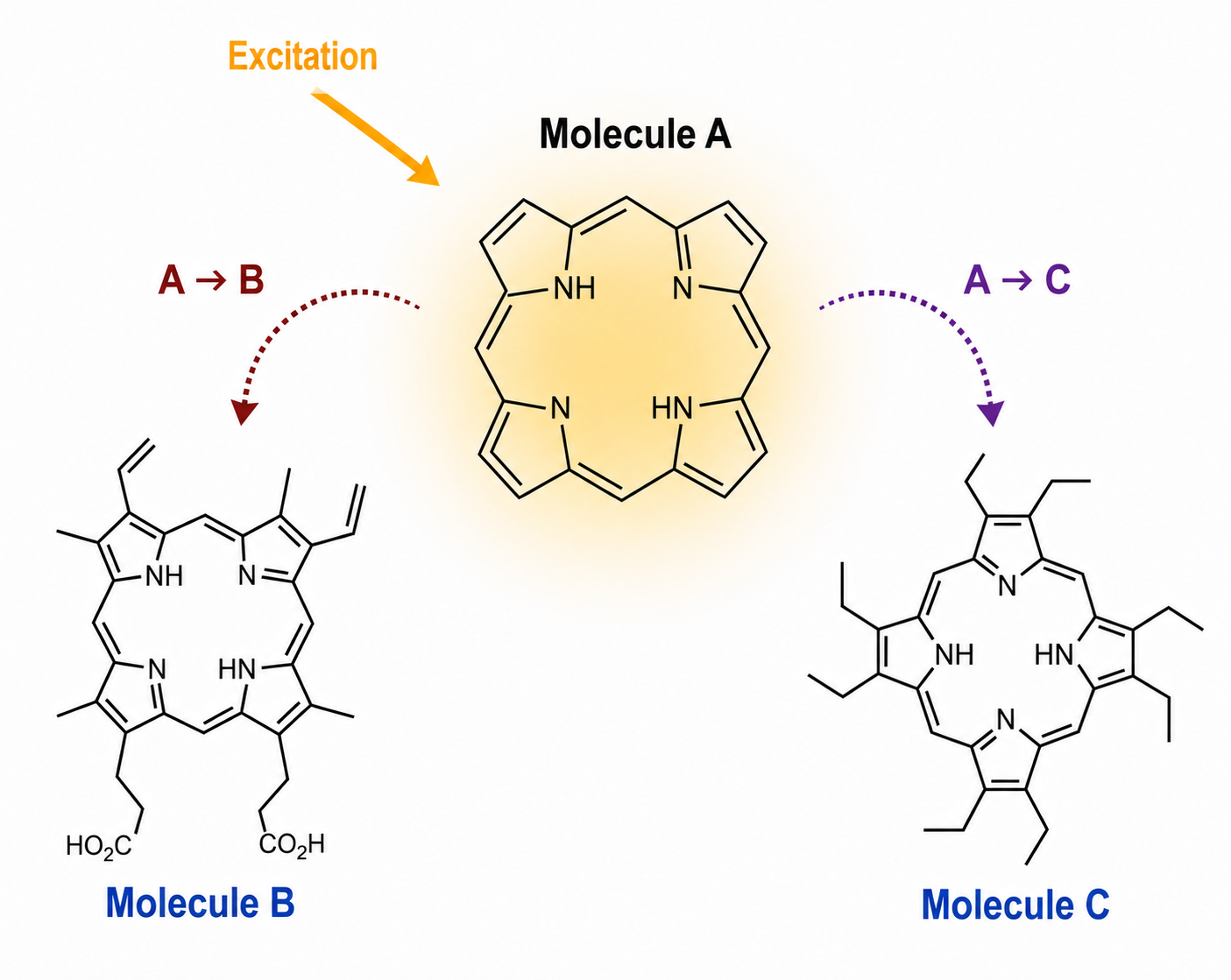}
    \caption{Three-chromophore energy-transfer model used as the vibronic
    benchmark. An excitation initially localized on chromophore $A$ can
    transfer to chromophores $B$ and $C$ while coupled to vibrational modes.}
    \label{fig:vibronic_model}
\end{figure}

We adopt the three-chromophore model of Ref.~\cite{vu2025} as our benchmark for two reasons. First, it exercises the complete instruction set of level-dependent phases and mode rotations, conditional displacements, electronic transitions, and link-mode-modulated sidebands within the same Hamiltonian (see Sec.~\ref{sec:ham}). Second, it is a well-specified model with published parameters, so the comparison rests on a system others can reproduce. The remainder of this section specifies the three evolution methods and the treatment of dissipation. The results follow in Sec.~\ref{sec:results}.

\subsection{Qudit Compilation}
\label{sec:methodA}
The qudit compilation directly compiles the lab-frame Hamiltonian into the native instruction set, without translation through an intermediate qubit circuit. It encodes the electronic state of the complex in the four levels, $\{\ket{G},\ket{A},\ket{B},\ket{C}\}$, of a single qudit, together with four qumodes for the local and link vibrational modes, giving a Hilbert-space dimension $4N^4$ at Fock cutoff $N$. The lab-frame Hamiltonian splits into six pieces, each compiling directly into ISA operations,
\begin{align}
H_{\rm lab} &= H_{\rm free} + H_{C_{\Pi}R} + H_{C_{\Pi_z}D} \nonumber\\
&\quad + H_{Z_d} + H_{X_E} + H_{C_{\Pi_x}D},
\end{align}
where $H_{\rm free}$ collects the unconditional mode rotations (four gates, one per qumode), $H_{C_{\Pi}R}$ the dispersive shifts $\chi_R\Pi_R\hat n_R$ (three level-conditioned rotations), $H_{C_{\Pi_z}D}$ the conditional displacements $-\omega_{eR}\sqrt{S_R}\,\Pi_R\hat x_R$ together with the link-mode term $-\omega_l\sqrt{S_l}\,\Pi_A\hat x_l$ (four conditional displacements), $H_{Z_d}$ the reorganization phases $\omega_{eR}S_R\Pi_R$ and $\omega_l S_l\Pi_A$ (four phase shifts), $H_{X_E}$ the electronic transitions $J^{(0)}_{RR'}(\Pi_{RR'}+\Pi_{R'R})$ ($A\!\leftrightarrow\!B$ and $A\!\leftrightarrow\!C$, two transitions), and $H_{C_{\Pi_x}D}$ the link-modulated sidebands $\eta_{RR'}J^{(0)}_{RR'}(\Pi_{RR'}+\Pi_{R'R})\hat x_l$ (two sideband gates). Each sub-Hamiltonian maps directly onto ISA operations, with the free terms implemented by unconditional mode rotations and the remaining terms by the five native primitives, making the decomposition an architectural property rather than a compiler artifact.

The five half-step sub-Hamiltonians other than $H_{C_{\Pi_x}D}$ contribute $4+3+4+4+2=17$ ISA operations per half-step, and the two $H_{C_{\Pi_x}D}$ terms are applied at the center of the symmetric step,
\begin{equation}
\label{eq:Aorder}
U_A(\tau) = \prod_{k=1}^{5} U_k(\tau/2)\,\cdot\,U_{C_{\Pi_x}D}(\tau)\,\cdot\,\prod_{k=5}^{1} U_k(\tau/2),
\end{equation}
giving $2\times17+2=36$ reference ISA operations per Trotter step.
The two central terms $H_{C_{\Pi_x}D}$ do not commute
with each other, so they are applied as one summed exponential, the
multitransition $C_{\Pi_x}D$ block of Table~\ref{tab:isa}, driven as two sideband
tone pairs on one beam. The step is then a true second order splitting. The simultaneous
two-transition drive is validated under the physical sideband Hamiltonian
at the circuit angles with $1-F=8\times10^{-9}$ (Appendix~\ref{app:pulses}).
Applying the two terms one after another instead adds a commutator error
that raises the Trotter RMSE from $0.007\%$ to $0.008\%$. All qudit
trajectories below use the summed block. A single optimization reduces this further. The two phase shifts on $\Pi_A$, $e^{-i\omega_{eA}S_a\tau\Pi_A}$ and $e^{-i\omega_l S_l\tau\Pi_A}$, commute exactly and fuse into one $Z_d$ operation, removing one operation per half-step for $2\times16+2=34$ reference ISA operations per step. The fused circuit is algebraically equivalent to the unfused decomposition and reproduces its evolution to machine precision. The generalized multilevel controls provide a further exact compilation
reduction. The three $Z_d$ operations acting on different qudit levels in each
half-step are combined into one generalized diagonal $Z_d$ operation, while the
two electronic transitions $A\!\leftrightarrow\!B$ and
$A\!\leftrightarrow\!C$ are combined into one multitransition $X_E$ control
block. The latter is synthesized to reproduce the original sequential $X_E$
unitary. The resulting $27$-control-block compilation reproduces the
$34$-operation reference evolution to machine precision, so the previously reported
trajectories are unchanged.
We use the $34$ reference ISA operations as the bookkeeping count in the
rest of the paper, since its only fusion is an exact identity. The $27$-control-block schedule is reported as the multitone extension, and the count it saves
depends on the control block described in Sec.~\ref{sec:isa}.

The circuit realizes one complete second-order Trotter step entirely within the native instruction set, with no intermediate qubit compilation, CNOT synthesis, or SWAP network (Fig.~\ref{fig:circuit_qudit}).
\begin{figure*}[t]
  \centering
  \includegraphics[width=0.92\textwidth]{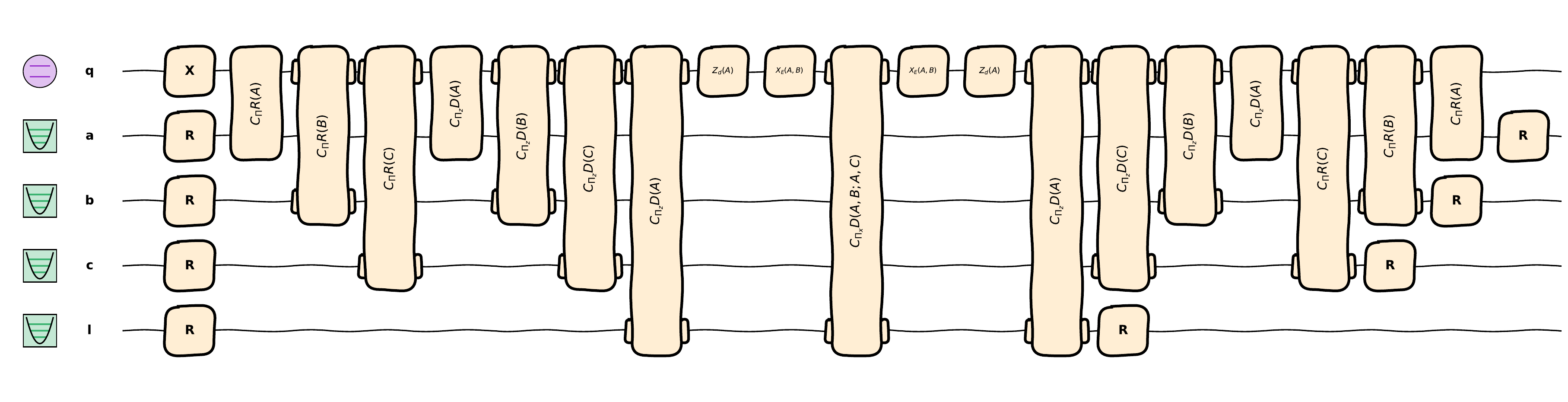}
  \caption{Single second-order Trotter step of the generalized qudit compilation
on the qudit and four-qumode register (Fock cutoff $N=8$). The step uses
$27$ ISA operations and no SWAP or CNOT operations.
This multitone version uses $27$ scheduled control blocks, while the
reference decomposition contains $34$ ISA operations, and neither has SWAP or CNOT operations. R denotes an
unconditional mode rotation, $C_{\Pi}R$, $C_{\Pi_z}D$, $Z_d$, $X_E$, and $C_{\Pi_x}D$ denote the
instruction-set primitives of Sec.~\ref{sec:isa}.}
  \label{fig:circuit_qudit}
\end{figure*}

\subsection{Qubit Reference}
\label{sec:methodB}
The qubit reference is constructed by compiling the displaced-frame
Hamiltonian $H_{\rm cqed}$ of Ref.~\cite{vu2025} into trapped-ion
operations on the same backend, so that the two encodings are compared on
one platform.

As set out in Sec.~\ref{sec:resources}, it works in the displaced frame, where the vibronic coupling becomes $\sigma_z^R(\hat a+\hat a^\dagger)$ and every term is two-body. It compiles that Hamiltonian, $H_{\rm cqed}$, on three qubits, $\ket{G}=\ket{000}$, $\ket{A}=\ket{100}$, $\ket{B}=\ket{010}$, $\ket{C}=\ket{001}$, alongside the same four qumodes.

The dispersive shifts become $\chi_R\sigma_z^R\hat n_R$, the vibronic
couplings $g_{{\rm cd},R}\sigma_z^R\hat x_R$, and the electronic exchanges
$(\sigma_x\sigma_x+\sigma_y\sigma_y)/2$ interactions. At the logical level
one symmetric step uses $32$ ISA operations, of which $20$ are entangling.
These logical operations are not all single native pulses. Under the
M{\o}lmer--S{\o}rensen primitive assumed here, each exchange term is
decomposed into two commuting windows, one for $XX$ and one for
$YY$~\cite{molmer1999, sorensen2000}, and the link-modulated exchange
$J_l$ has no native two-ion primitive and is synthesized from
M{\o}lmer--S{\o}rensen sandwiches around a conditional displacement. The
two $J_l$ terms at the center of the step are applied one after the other,
which adds a splitting error of about $0.004\%$ in mean population RMSE
relative to the exact summed center. Compiling to timed pulses
expands the step to $58$ under the explicit compilation of
Appendix~\ref{app:qubitport}, with each CNOT counted as one pulse. The qudit compilation of Sec.~\ref{sec:methodA} does not expand,
since every ISA operation is a single native interaction. It compiles
to $27$ timed pulses, fewer than its $34$ operations because $Z_d$ costs
no pulse.

The qubit reference compiles the displaced-frame Hamiltonian $H_{\rm cqed}$, not $H_{\rm lab}$. It represents the qubit formulation used in Ref.~\cite{vu2025} and serves as the established result against which the qudit compilation is validated.

\subsection{Simulation}
\label{sec:protocol}
All compilations start from $\ket{A}$ with every mode in vacuum and use second-order Trotter decompositions with the same step size and Fock cutoff, and run on the extended HybridLane of Sec.~\ref{sec:hybridlane}.

The population trajectories reported below were generated using the
$34$-operation reference qudit circuit with the summed $C_{\Pi_x}D$ center, and the
$27$-control-block multitone schedule reproduces this evolution to machine
precision. The plotted closed system qubit reference uses the summed $J_l$ center,
and the sequential circuit of Sec.~\ref{sec:methodB} differs from it by
$0.004\%$ mean population RMSE. The qudit compilation uses the four-level register. The qubit reference, by contrast, uses a one-qubit-per-state register on the same backend, so observed differences originate from the encoding and the Hamiltonian it compiles, not from the solver. Populations $P_R(t)=\braket{\Pi_R}$ are extracted at each step, checked against direct integration of $H_{\rm lab}$ with QuTiP~\cite{qutip2012, qutip2013} to isolate the Trotter error, and compared against the published reference trajectories of Ref.~\cite{vu2025}.

For the closed system (Sec.~\ref{sec:results_closed}) we integrate the Schr{\"o}dinger equation with QuTiP~\cite{qutip2012, qutip2013}. For the dissipative dynamics (Sec.~\ref{sec:results_dissipative}) we use a continuous Lindblad master equation, cross-checked against a discretized Stinespring channel~\cite{lindblad1976, stinespring1955}.

The amplitude-damping and pure-dephasing channels for chromophore $R$ are
\begin{equation}
\label{eq:Lqd}
L_{\rm amp,R} = \sqrt{\gamma_{\rm amp,R}}\,\ket{G}\!\bra{R}, \ \ 
L_{\rm dep,R} = \sqrt{\gamma_{\rm dep,R}}\,\bigl(I - 2\Pi_R\bigr).
\end{equation}
The dephasing operator is the qudit form of the qubit $\sigma_z^R$
within the single-excitation manifold, where $\sigma_z^R$ has eigenvalue $+1$ on all states with chromophore $R$ unexcited and $-1$ on $\ket{R}$, which is exactly $I-2\Pi_R$. The two encodings therefore implement the same physical dissipation, even though they compile different Hamiltonians.


\section{Results}
\label{sec:results}
We validate the native qudit compilation of $H_{\rm lab}$ against the published reference of Ref.~\cite{vu2025}, first in the closed system, then under dissipation, and finally with respect to the two discretization parameters. Agreement is quantified by the root-mean-square error (RMSE) between chromophore population trajectories.

\subsection{Closed-System Dynamics}
\label{sec:results_closed}
Figure~\ref{fig:closed} compares three trajectories over the full
$2$~ps window considering direct integration of $H_{\rm lab}$, the Trotterized qudit compilation of $H_{\rm lab}$, and the qubit compilation of the displaced-frame $H_{\rm cqed}$. The Trotterized qudit dynamics match the exact lab-frame evolution to an RMSE of order $10^{-4}$, so the Trotter error is negligible at the step size used. The qubit compilation reproduces the published reference trajectory to $0.031\%$, confirming that our reimplementation of the reference is faithful.
The two compiled models differ from each other by $1.6\%$ at these
parameters. The residual arises between two different finite-dimensional
representations, $H_{\rm lab}$ and the displaced-frame $H_{\rm cqed}$, with
the Trotter contribution being much smaller.

\begin{figure*}[t]
  \centering
  \includegraphics[width=0.9\textwidth]{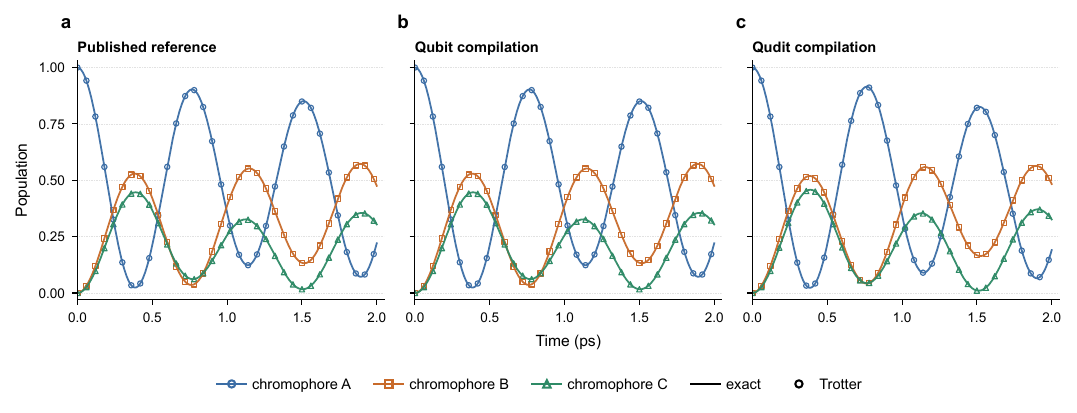}
  \caption{Closed-system population dynamics over $2$~ps, corresponding to Fig.~6 of Ref.~\cite{vu2025}. (a) The published reference and (b) the qubit compilation compile the displaced-frame $H_{\rm cqed}$; (c) the qudit compilation compiles $H_{\rm lab}$. Solid lines are exact integration, markers the Trotterized circuits. The qudit compilation matches the exact evolution to an RMSE of order $10^{-4}$. The qubit compilation reproduces the published reference to $0.031\%$}.
  \label{fig:closed}
\end{figure*}

\subsection{Dissipative Dynamics}
\label{sec:results_dissipative}
Figures~\ref{fig:amp}--\ref{fig:combined} give the dissipative dynamics under amplitude damping, pure dephasing, and both together, across the eight panels of Ref.~\cite{vu2025}. Figures~\ref{fig:amp}--\ref{fig:combined} show the qudit compilation against the published trajectories. Table~\ref{tab:rmse} additionally reports the displaced-frame comparison.

At Fock cutoff $N=4$, the displaced-frame $H_{\rm cqed}$ calculation,
implemented with the level-encoded Kraus builder, reproduces the reference
to $1.03\%$ RMSE, while the qudit calculation gives $1.09\%$, i.e., the two
agree to $0.30\%$. At $N=8$, the displaced-frame calculation is implemented
as the explicit three-qubit circuit, giving corresponding values of
$1.02\%$, $1.09\%$, and $0.45\%$.

The $N=4$ and $N=8$ rows were produced with different shot backends, a
Kraus channel and a Stinespring dilation, so the small change between
rows reflects both the cutoff and the solver.  That the qudit
compilation of $H_{\rm lab}$ results are close to a reference computed
in the displaced frame is the very point of the comparison. It
indicates that the frame transformation is faithful at these
parameters, and the native compilation recovers the same physics
without performing it.

The strongly dephased panel of Fig.~\ref{fig:dep} departs from the published
curve by $\sim6\%$, while the lab-frame and displaced-frame compilations
agree with each other to $0.7\%$ there. The residual is between our
calculations and the published reference, not between the two encodings.

\begin{figure*}[t]
  \centering
  \includegraphics[width=0.9\textwidth]{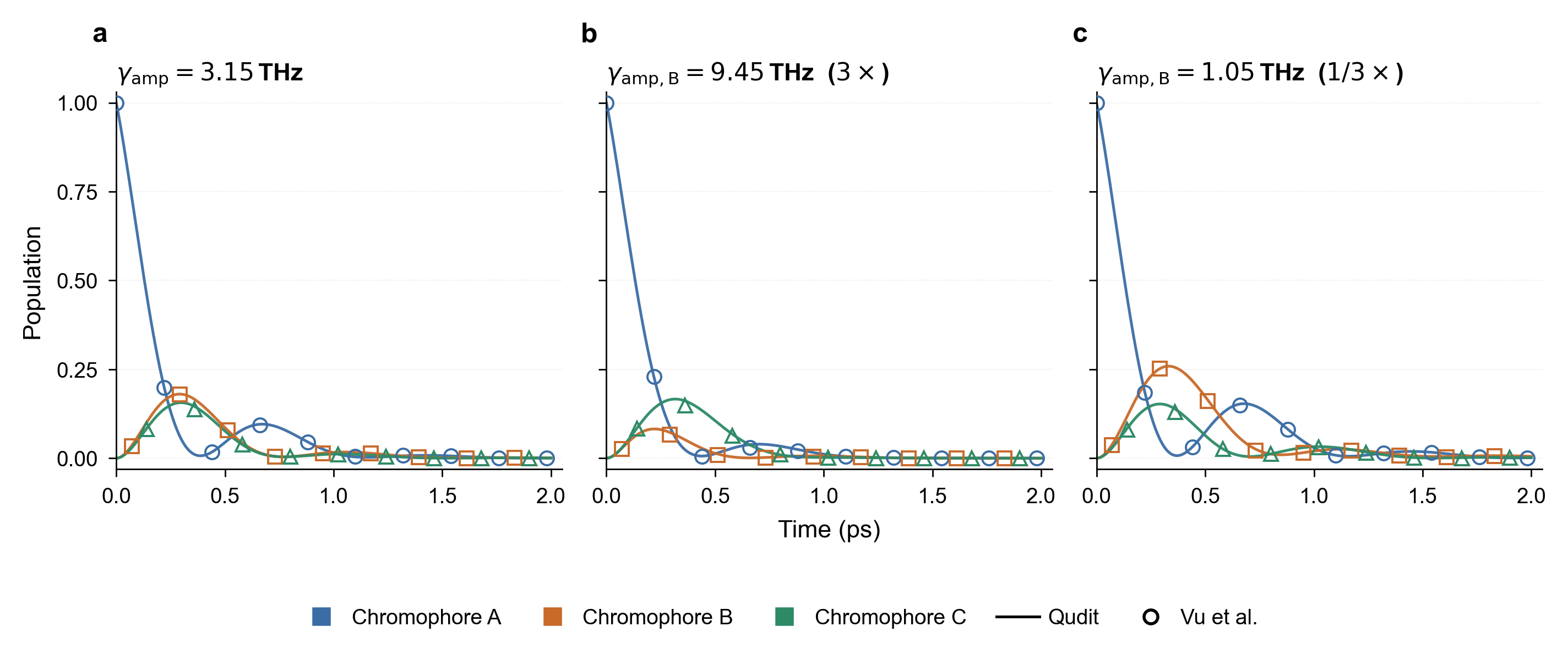}
  \caption{Dissipative population dynamics under amplitude damping, corresponding to Fig.~8 of Ref.~\cite{vu2025}. The qudit compilation tracks the published reference across all chromophores.}
  \label{fig:amp}
\end{figure*}

\begin{figure*}[t]
  \centering
  \includegraphics[width=0.9\textwidth]{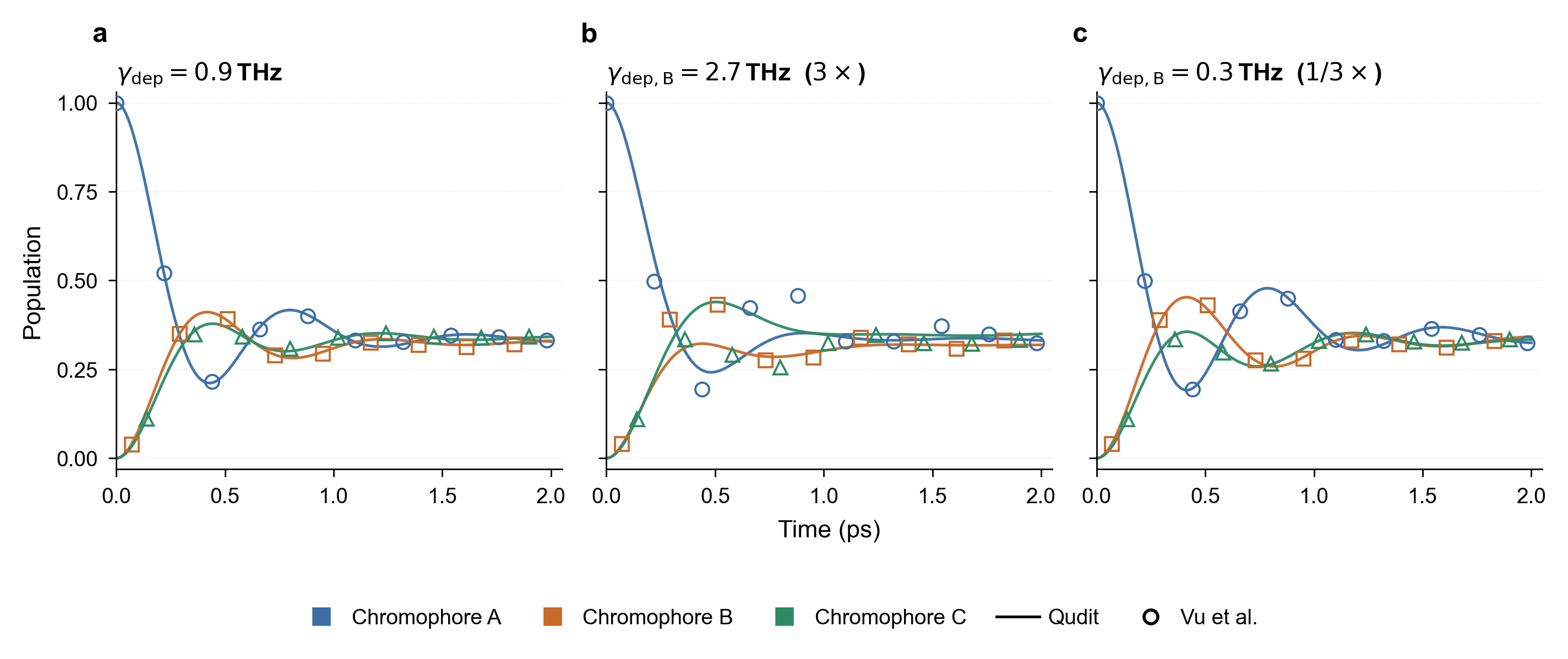}
  \caption{Dissipative population dynamics under pure dephasing, corresponding to Fig.~9 of Ref.~\cite{vu2025}. Agreement of the qudit compilation with the published reference is better than $0.8\%$ on all panels except the strongly dephased one, where both compilations depart from the published reference by $\sim6\%$ while agreeing with each other to $0.7\%$.}
  \label{fig:dep}
\end{figure*}

\begin{figure*}[t]
  \centering
  \includegraphics[width=0.7\textwidth]{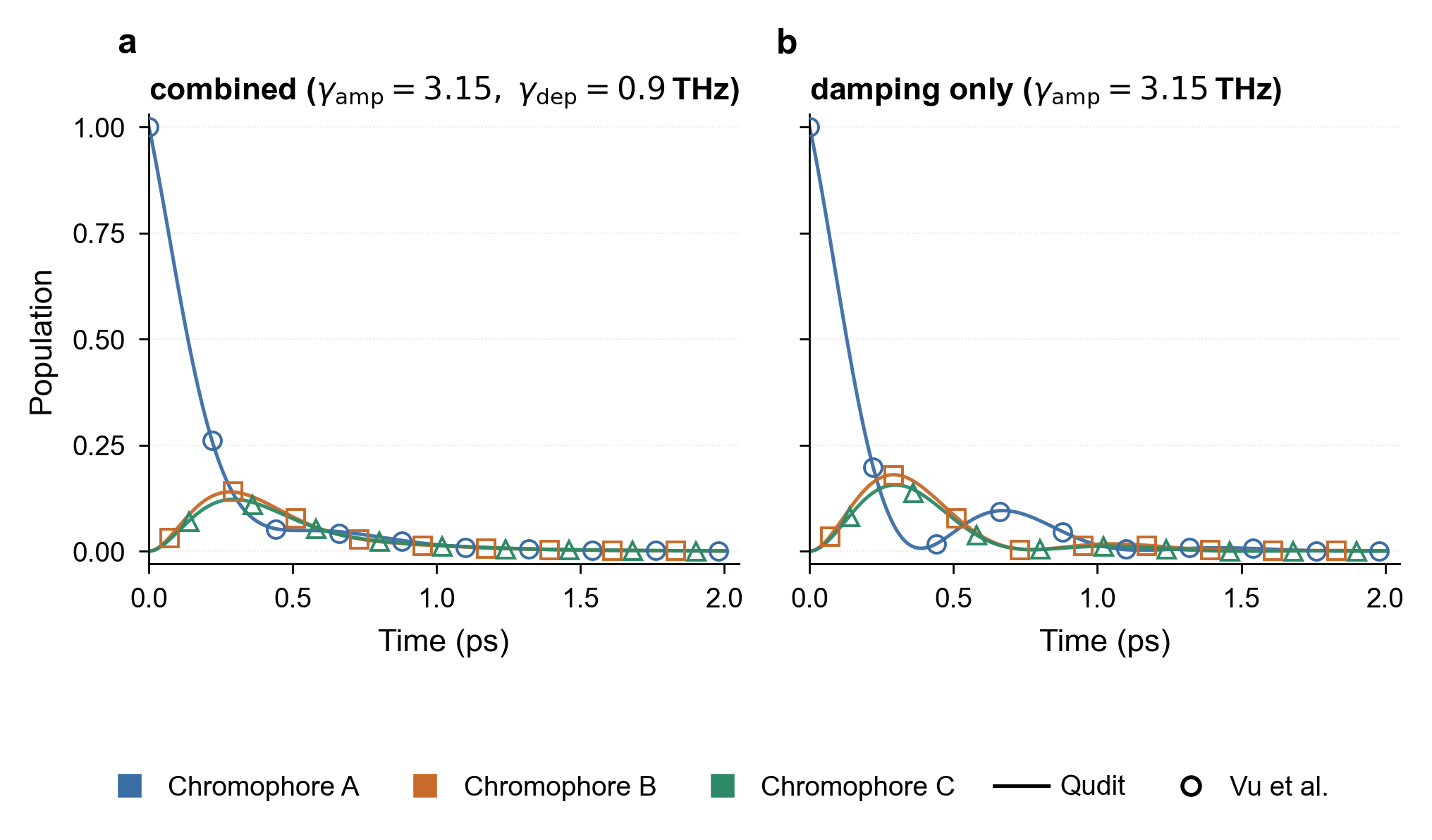}
  \caption{Dissipative population dynamics under combined amplitude damping and dephasing, corresponding to Fig.~10 of Ref.~\cite{vu2025}.}
  \label{fig:combined}
\end{figure*}

\begin{table}[h]
  \centering
  \begin{tabular}{lccc}
    \toprule
    Cutoff & Displaced--ref. & Qudit--ref. & Displaced--qudit \\
    \midrule
    $N=4$ & $1.03\%$ & $1.09\%$ & $0.30\%$ \\
    $N=8$ & $1.02\%$ & $1.09\%$ & $0.45\%$ \\
    \bottomrule
  \end{tabular}
  \caption{Population-trajectory RMSE for the dissipative dynamics, averaged
over the eight panels of Figs.~\ref{fig:amp}--\ref{fig:combined}. At $N=4$,
the displaced-frame $H_{\rm cqed}$ result is produced by the level-encoded
Kraus builder, whereas at $N=8$ it is produced by the explicit three-qubit
Stinespring circuit. The qudit column evaluates $H_{\rm lab}$. The two rows
therefore represent independently validated implementations rather than a
controlled Fock-cutoff sweep. The reported averages include the strongly dephased panel discussed in the
text.}
  \label{tab:rmse}
\end{table}

\subsection{Convergence}
\label{sec:results_convergence}
Figure~\ref{fig:conv_trotter} shows the closed-system population RMSE
against the unitary Trotter step size on a log--log scale. After tightening
the numerical integration tolerances so that the solver error remains below
the Trotter error, both compilations exhibit the expected second-order
Suzuki scaling. Over the fitted range $\tau\ge5$~fs, the fitted slopes are
$1.981$ for the qudit compilation and $2.027$ for the qubit reference
($R^2=1.0000$ for both). At the production step size $\tau=5$~fs, the mean
population RMSE is $0.0062\%$ for the qudit compilation and $0.0044\%$ for
the qubit reference. These results support the choice $\tau=5$~fs. We use
$N=8$ as the working Fock cutoff for the circuit calculations.

\begin{figure}[t]
  \centering
    \includegraphics[width=\columnwidth]{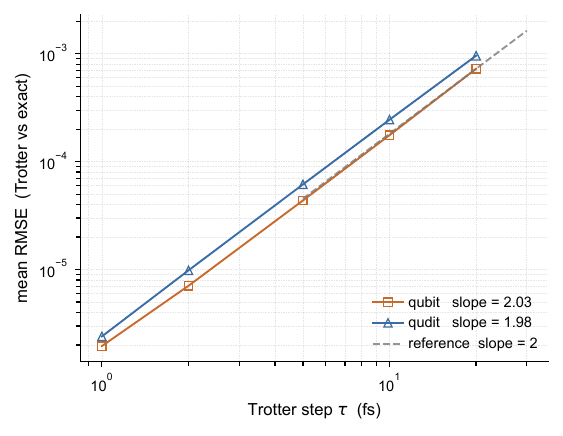}
  \caption{Closed-system population RMSE versus unitary Trotter step size
$\tau$ (log--log). Over the fitted range $\tau\ge5$~fs, the slopes are
$1.981$ for the qudit compilation and $2.027$ for the qubit reference,
consistent with second-order Suzuki scaling. The operating point
$\tau=5$~fs lies in the converged regime.}
  \label{fig:conv_trotter}
\end{figure}


\section{Resource Scaling and Architectural Tradeoffs}
\label{sec:resources}
This section reports programmatically counted resource metrics for one
second-order Trotter step of the vibronic Hamiltonian of
Sec.~\ref{sec:ham}, generalized to $N_c$ chromophores and anchored at the
$N_c=3$ case of Sec.~\ref{sec:methods}. Reference counts are ISA
operations as defined in Sec.~\ref{sec:isa}, while the multitone count is
reported in scheduled control blocks. The underlying pulse-level
implementation is platform-dependent, and for large $N_c$ not every level
pair needs to correspond to a single directly driven physical transition.

The two trapped-ion benchmarks intentionally follow different Hamiltonian
representations. The qudit compilation applies $H_{\rm lab}$ directly,
whereas the qubit reference follows Ref.~\cite{vu2025} and applies the
displaced-frame Hamiltonian $H_{\rm cqed}$. This is a comparison against the
published qubit formulation. A one-hot qubit register can also represent
$H_{\rm lab}$, since within the single-excitation code space,
$\Pi_R=(I-\sigma_z^R)/2$, so the lab-frame projector is also a single-qubit
operator. The supported architectural advantages of the qudit encoding are
instead computational-register compression and interaction locality. The full
electronic manifold occupies one ion, its projectors are native level
projectors, and electronic exchanges become intra-ion $X_E$ operations.

The platform itself accounts for the largest share of the reduction. On the
superconducting hardware layout of Ref.~\cite{vu2025}, nearest-neighbor
connectivity forces SWAP networks, and one Trotter step of the three-site
model decomposes into $84$ CNOT, $9$ conditional-displacement, and $3$ SNAP
gates. 

That estimate counts each Hamiltonian term once, includes only two-qubit
and qubit--mode gates, and prices the qubit part for a nearest-neighbor
line. It is a resource estimate for that layout and is not in the same
counting convention as the trapped-ion circuits below, which counts every
operation of a full symmetric step. We therefore use it as context only and
compare the two trapped-ion encodings against each other.
The corresponding layout for a general one-dimensional $N_c$-chromophore
array uses $2N_c-2$ transmons and $2N_c-2$ cavities. Neither trapped-ion compilation
contains explicit CNOT or SWAP operations at the logical level. Instead, the corresponding interactions are
implemented directly by the native operations, with the shared motional bus
coupling the discrete register to the motional modes. The trapped-ion
register must in turn supply $2N_c-2$ spectrally resolvable motional modes
from the ion crystal, so the compression concerns the discrete computational
register rather than the oscillator resources, which both platforms provide
in comparable number. The qudit register then
adds the compression and locality advantages of the following subsections on
top of this platform-level reduction.

\subsection{Computational-Register Footprint}

The first resource difference is the size of the discrete computational
register. In the one-qubit-per-state trapped-ion encoding, with one
computational qubit per ion, $N_c$ excited electronic states are represented
by $N_c$ qubits, with the reference state represented by the all-zero state.
Thus,
\begin{equation}
N_{\rm comp}^{(\rm qubit)}=N_c .
\end{equation}

When the full manifold fits within one experimentally usable qudit,
$d_{\rm eff}\ge N_c+2$, the corresponding qudit encoding requires only
\begin{equation}
N_{\rm comp}^{(\rm qudit)}=1 .
\end{equation}
The computational-register compression factor in this regime is therefore
\begin{equation}
\mathcal{C}_{\rm ion}
=
\frac{N_{\rm comp}^{(\rm qubit)}}
     {N_{\rm comp}^{(\rm qudit)}}
=
N_c .
\end{equation}

For the three-chromophore benchmark, the qubit representation therefore uses
three computational qubits, whereas the qudit representation stores the
entire four-state manifold
$\{\ket{G},\ket{A},\ket{B},\ket{C}\}$ in one ququart ion.

These numbers refer specifically to ions carrying the discrete computational
register, not to the total number of ions in the trap. Additional ions may be
required to provide the desired collective motional-mode structure or to
mediate synthesized multimode operations.

\subsection{Interaction Locality}

Register compression also changes the locality of the electronic
interactions. Consider an electronic coupling
\begin{equation}
H_{RR'}
=
J_{RR'}
\left(
\ket{R}\!\bra{R'}+\ket{R'}\!\bra{R}
\right).
\end{equation}
When $\ket{R}$ and $\ket{R'}$ are levels of the same qudit, this interaction
is an intra-ion transition and is implemented directly by the $X_E$ primitive.
In the one-qubit-per-state representation, the corresponding excitation
exchange acts on two different qubits and is implemented through a two-ion
interaction.

Thus, whenever the relevant electronic manifold fits within one qudit,
edges of the electronic coupling graph are converted from inter-ion
interactions into local transitions within a single ion. In the
three-chromophore benchmark, both $A\leftrightarrow B$ and
$A\leftrightarrow C$ couplings are therefore local $X_E$ operations.

This locality advantage should not be confused with a reduction in the total
number of qudit--oscillator entangling operations. Here,
the generalized qudit compilation additionally reduces the
scheduled control-block count, and the qudit compilation also uses fewer qudit--oscillator entangling
operations, $16$ against $20$ at $N_c=3$, because the electronic exchange
terms become local $X_E$ operations.

\subsection{ISA Operation Counts}
Table~\ref{tab:gatecounts} shows the number of operations required.
The rank-one compilation follows programmatically generated counts that
grow linearly in $N_c$. With the generalized multilevel controls,
compatible $Z_d$ and $X_E$ operations can be combined, reducing the
number of control blocks per step, as verified programmatically across
the system sizes considered.

The
one-qubit-per-state encoding uses a count that grows at the same rate,
see Appendix~\ref{app:runtime}.

\begin{table}[h]
  \centering
  \small
    \begin{tabular}{lccc}
    \toprule
  & Count & Timed pulses & Entangling \\
\midrule
Qubit reference & 32 ISA ops. & 58 & 20 \\
Qudit, reference & 34 ISA ops. & 27 & 16 \\
Qudit, multitone & 27 ctrl. blocks & 25 & 15 \\
\bottomrule
    \end{tabular}
  \caption{Operation counts per second-order Trotter step at $N_c=3$.
  The Count column reports ISA operations for the reference circuits and
scheduled control blocks for the multitone circuit. The Timed pulses column counts
  pulses after compiling to trapped-ion primitives, with the dispersive
  gates $C_{\Pi}R$ and CR each a single pulse and
    the $J_l$ term of the qubit reference decomposed with each CNOT counted as
  one pulse (Appendix~\ref{app:qubitport}). The two central
  $C_{\Pi_x}D$ terms count as two logical operations executed in one pulse.
  The Entangling column reports entangling ISA operations for the reference
circuits and entangling control blocks for the multitone circuit. The qudit circuits compile $H_{\rm lab}$, the qubit reference
  compiles $H_{\rm cqed}$.}
  \label{tab:gatecounts}
\end{table}

At the logical level the counts are close, $34$ against $32$ at $N_c=3$,
with both growing linearly in $N_c$ and the multitone qudit version
growing more slowly, where
the reference count reflects the number of Hamiltonian terms while the multitone count
summarizes the number of control blocks, with $Z_d$, $X_E$, and the central $C_{\Pi_x}D$ each fused
into one block. After compilation to timed pulses the reference qudit circuit
uses $27$ and the qubit reference $58$, because each exchange needs two
M{\o}lmer--S{\o}rensen windows and $J_l$ has no native gate. The
entangling count also favors the qudit, $16$ against $20$. The larger
difference is in time, since the qubit cost sits in a few long
M{\o}lmer--S{\o}rensen windows rather than in many short pulses, as
Appendix~\ref{app:runtime} shows.
This reduction is specific to the benchmark circuit and its pricing model.

Its resource benefit
instead lies in computational-register compression, interaction locality, and
the direct representation of the physical discrete manifold.

A compact binary encoding would store the manifold in
$\lceil\log_2(N_c+1)\rceil$ qubits, but the state-conditioned terms then
become multi-controlled operators whose decomposition cost grows sharply with
$N_c$. It is not the natural choice for a single-excitation model, and
Ref.~\cite{vu2025} uses the one-qubit-per-state encoding, which we adopt for
the comparison.

\subsection{Architectural Tradeoff}
The qudit architecture exchanges spatial register complexity for local
multilevel-control complexity. When the discrete manifold fits within one
qudit, it reduces the number of computational ions, converts electronic
couplings into intra-ion transitions, and represents state-conditioned
operators such as $\Pi_R=\ket{R}\!\bra{R}$ directly. 
The qudit representation permits the lab-frame vibronic Hamiltonian to
compile term by term in native level-selective operations. The qubit
reference considered here instead follows the displaced-frame formulation
of Ref.~\cite{vu2025}, although the one-hot encoding can also represent
$H_{\rm lab}$ within its code space.
A mediator ancilla ion is needed only when the synthesized beam splitter
itself is invoked, while universality and the vibronic circuits require none.

The corresponding cost is spectral-control complexity. The number of internal
levels and spectrally resolved controls grows with $N_c$, and these must be
resolved, calibrated, and controlled while suppressing leakage and
off-resonant excitation. The useful operating point is therefore determined
by a hardware-dependent tradeoff between computational-ion compression and
multilevel-control complexity, rather than by qudit dimension alone.

\subsection{Projected Runtime}
Under the timing model of Appendix~\ref{app:runtime}, one Trotter step at
$N_c=3$ is projected at $210~\mu$s on the qudit against
$2208~\mu$s for the qubit reference run sequentially, a ratio of $10.5$
under the sequential timing model.
The projection rests on three
assumptions. First, the dispersive detuning behind $C_{\Pi}R$ is an
operating point chosen in this work. We
bound it from below by the sideband Rabi frequency and from above by the
motional mode spacing, and report the projected ratio across that window
in Appendix~\ref{app:pulses}. Second, every gate is priced from its own compiled angle by a stated rule
(Appendix~\ref{app:runtime}), and the $C_{\Pi}R$ duration is fixed by a
physical duration scan rather than by the formula alone. The M{\o}lmer--S{\o}rensen time is taken as
angle independent, which is the measured practice of
Ref.~\cite{meth2025} and sets most of the qubit time. Third, several
primitives are independently replayed under the physical drive
Hamiltonian, but not all, and the full Trotter circuit is not simulated
pulse by pulse. The projected advantage is therefore a compiled-circuit estimate
under the stated pulse parameters.

\subsection{Scaling Beyond One Ion}
\label{sec:multiion}
One ion holds at most $K=d_{\rm eff}-2$ excited states, since one level
is the reference and one is the auxiliary. Our $^{40}$Ca$^+$ layout of
Sec.~\ref{sec:encoding} reserves one more sublevel and holds $K=5$. When
$N_c$ exceeds $K$,
the manifold is spread over
\begin{equation}
N_{\rm ion}=\left\lceil\frac{N_c}{K}\right\rceil
\end{equation}
computational ions, and the register compression becomes a constant factor
of $K$ per ion. Each ion keeps its own reference level, and the global
$\ket{G}$ is every ion in its reference level. Gates inside one ion keep
their local form. An electronic coupling whose two levels sit on different
ions moves the excitation from one ion to the other, and on trapped ions
it is built from M{\o}lmer--S{\o}rensen interactions between qudit
levels~\cite{hrmo2023}. We price it as the qubit reference prices its
exchange, two windows for the exchange and one window for the
link-modulated term. Universality is unchanged. Local control on each ion
from $X_E$ and $Z_d$ together with the two-ion exchange on a connected set
of ion pairs gives full control of the multi-ion register, and
$C_{\Pi_z}D$ on each ion couples it to the modes as in
Appendix~\ref{app:universality}.

The projected speedup has a closed form. Following the one-dimensional
chain extension of Ref.~\cite{vu2025}, each of the $N_c-1$ bonds carries
two exchange terms and one or two mode-modulated terms, one on an end
bond and two on an interior bond. Under the timing model of
Appendix~\ref{app:runtime}, a bond costs $c_{\rm out}$ when its two
levels sit on different ions, $1,000~\mu$s on an end bond and
$1,200~\mu$s on an interior bond, and
$c_{\rm in}$ when they sit on the same ion, about
$1.3$ or $2.0~\mu$s from the class-mean times of $X_E$ and $C_{\Pi_x}D$.
The qubit reference pays $c_{\rm out}$ on every bond. Packing $K$
consecutive chromophores per ion gives $M=\lceil N_c/K\rceil$ ions, with
$M-1$ bonds crossing ions and $N_c-M$ inside. Let $T_0$ be the time of the remaining gates, which under the timing
model is the same on both sides at every $N_c$. Then
\begin{equation}
\label{eq:cluster}
S=\frac{T_0+\sum_{\rm bonds} c_{\rm out}}
{T_0+\sum_{\rm inside} c_{\rm in}+\sum_{\rm crossing} c_{\rm out}} .
\end{equation}
Eq.~\eqref{eq:cluster} reproduces every entry of Table~\ref{tab:runtime}.
The class-mean gate times and the one-window cost of each mode-modulated
term are assumptions of the timing model.

Three properties follow. First, when the whole manifold fits in one ion,
$N_c\le K$, no bond crosses and $S$ grows with $N_c$. The first crossing
bond at $N_c=K+1$ adds $c_{\rm out}$ to the qudit time, and every row of
Table~\ref{tab:runtime} and of the split cases below has $S$ smaller
there than at $N_c=K$. Second, for long chains with $N_c\gg K$ both times grow linearly
in $N_c$, so $S$ is set by the time each side adds per extra
chromophore. Adding one chromophore adds one bond. The qubit reference
pays $c_{\rm out}$ for it. The qudit pays $c_{\rm in}$ for $K-1$ bonds
out of every $K$ and $c_{\rm out}$ for the one that crosses ions.
Writing $a$ for the growth of $T_0$ per chromophore and taking
$c_{\rm in}$ and $c_{\rm out}$ at their interior bond values,
$S$ approaches
\begin{equation}
S_\infty=\frac{a+c_{\rm out}}{a+c_{\rm in}\,(1-1/K)+c_{\rm out}/K},
\end{equation}
which lies below $K$ and above one. Third, $S\ge1$ for all $N_c$, since
$T_0$ is equal on both sides, the crossing bonds cost the same, and the
inside bonds cost less on the qudit.

For the three-chromophore benchmark at $d_{\rm eff}=4$ with one
auxiliary per ion, $K=2$, two ions are needed and one bond crosses. The
projected step time rises from $210~\mu$s to $1209~\mu$s and $S$ drops
from $10.5$ to $1.8$. On $^{40}$Ca$^+$ with $K=5$, a five-chromophore
chain still fits in one ion at $S=12.8$. Six, eight, and twelve
chromophores need two, two, and three ions at $S=4.2$, $4.7$, and
$4.1$, and $S_\infty=4.1$. The gain of the qudit register is therefore
the number of bonds that stay inside one ion. Among the cases evaluated
it is largest when the whole manifold fits, and ions with many usable
levels, such as $^{137}$Ba$^+$~\cite{low2025qudit, shi2025}, keep a
larger constant gain for long chains. The same argument extends beyond
a chain. For an interaction graph that is clustered, where most
couplings sit inside groups no larger than $K$ and only a few cross
between groups, the fraction of couplings that stay inside one ion can
be higher than in a chain, and the projected gain under the stated
timing model can be larger.


\section{Discussion}
\label{sec:discussion}

The principal benefit of the proposed architecture remains structural,
although generalized multilevel controls also reduce the number of scheduled
control blocks. 

When the relevant discrete manifold fits within
a single qudit, the computational register is compressed into one multilevel
ion, electronic couplings become intra-ion transitions, and state-conditioned
operators are represented directly. For the vibronic benchmark this enables
the lab-frame Hamiltonian to be compiled without the displaced-frame
transformation used in the qubit formulation.

These advantages do not imply a demonstrated improvement in fidelity. For
the three-chromophore benchmark the reference qudit circuit compiles to $27$ timed
pulses, where the qubit reference needs $58$, and under the timing model
of Appendix~\ref{app:runtime} one Trotter step has a projected speedup of $10.5$
on the qudit under the sequential timing model. That projection is an estimate under the
stated pulse parameters.

The practical
operating regime therefore depends on the hardware cost and fidelity of
multilevel control.

The relevant scaling parameter is the experimentally usable qudit dimension
$d_{\rm eff}$. Increasing $d_{\rm eff}$ allows larger discrete manifolds to be
encoded within one ion, but also increases spectral-control requirements.
In particular, the fidelity scaling of the native qudit--oscillator
interactions at large $d$ remains to be experimentally established.

Although vibronic dynamics provides the validation considered here, the
instruction set and universality result are not specific to molecular
simulation and apply more broadly to hybrid problems involving multilevel
discrete systems coupled to bosonic modes. The same control structure may
also provide a basis for qudit-assisted bosonic state preparation,
error-correction protocols, and bosonic-code constructions in which a
multilevel ancilla controls oscillator operations.


\section{Conclusion}
\label{sec:conclusion}
Scaling trapped-ion processors by adding ions while preserving performance is
challenging. Exploiting the multilevel structure within each ion is a
complementary scaling direction, and this work develops it for hybrid
computation. We introduced a hardware-native instruction set for hybrid
qudit--oscillator computation on trapped ions, consisting of five native primitives for
qudit control and qudit--motion interactions. Three of them, $X_E$, $Z_d$, and $C_{\Pi_z}D$, already generate
Hamiltonians spanning arbitrary qudit operators tensored with polynomials in
the quadratures of all modes. An ancilla-mediated beam splitter is provided as an
efficient Gaussian mode-mixing primitive. The phase-space and sideband
computational models coexist natively within one architecture.

The principal advantage of the qudit register remains architectural, while
generalized multilevel controls additionally reduce the scheduled control-block count.  

Whenever a discrete manifold fits within the usable
dimension $d_{\rm eff}$, $N_c$ computational qubit ions can be represented by
one multilevel ion, couplings between encoded states become intra-ion
transitions, and state projectors $\Pi_R$ are local single-qudit operators, so the
lab-frame Hamiltonian maps directly onto the ISA without a frame
transformation.
The number of internal levels and spectrally resolved controls grows with
$N_c$,
and a larger $d$ raises the spectral addressing and calibration burden,
while the resulting fidelity scaling remains hardware dependent
(Sec.~\ref{sec:resources}).

A qubit encoding does not incur the same local
multilevel-control scaling, because each computational unit remains
two-level, but pays instead with $N_c$ computational ions and inter-ion
entangling operations, along with its own system-size effects in addressing
and mode structure. For the vibronic model, ISA operation counts grow linearly in $N_c$ for both
compilations and multitone control blocks grow more slowly,
while the qudit encoding reduces the computational register from
$N_c$ ions to one whenever $N_c+2\leq d_{\rm eff}$. Past that point the
manifold is spread over $\lceil N_c/K\rceil$ ions with $K$ excited
states per ion. The compression is then a constant factor per ion, and
the projected runtime ratio settles at a constant above one that grows
with $K$, $4.1$ on $^{40}$Ca$^+$ under the timing model, so ions with many usable levels keep a larger gain for long
chains (Sec.~\ref{sec:multiion}).

We demonstrated this structure on the dissipative vibronic dynamics of a
three-chromophore complex previously simulated on superconducting hybrid
hardware, where each Trotter step requires $84$ CNOT gates
(Sec.~\ref{sec:resources}). 
That estimate counts each term once for a nearest-neighbor layout. On
trapped ions the shared motional bus removes the routing overhead in both
encodings.

Its
four-state electronic manifold is encoded in one ququart, the vibrations
remain native motional modes, and the lab-frame Hamiltonian maps term by term
onto the instruction set. One ququart replaces
three computational qubits
at $34$ reference ISA operations against $32$, compiling to $27$ timed
pulses against $58$, and one Trotter step is projected to run $10.5$
times faster on the qudit under the stated sequential timing model, and
the dynamics reproduce the published dissipative reference
trajectories to approximately $1.1\%$ average RMSE.


\begin{acknowledgments}
We thank Ken Brown and Jungsang Kim (Duke University) for insightful discussions.
This work was supported in part by National Science Foundation grants
OSI-2531350, PHY-2325080 and OMA-2120757.
Yuan Liu and Frank Mueller were also supported in part by the
U.S. Department of Energy, Office of Science, Advanced Scientific
Computing Research, under contract number DE-SC0025384.
\end{acknowledgments}


\section*{Data Availability}
The code and data supporting the findings of this study are openly available at \url{https://github.com/fmuelle4711/hybrid-isa}.

\bibliographystyle{apsrev4-2}
\bibliography{references}

\begin{thebibliography}{47}%
\makeatletter
\providecommand \@ifxundefined [1]{%
 \@ifx{#1\undefined}
}%
\providecommand \@ifnum [1]{%
 \ifnum #1\expandafter \@firstoftwo
 \else \expandafter \@secondoftwo
 \fi
}%
\providecommand \@ifx [1]{%
 \ifx #1\expandafter \@firstoftwo
 \else \expandafter \@secondoftwo
 \fi
}%
\providecommand \natexlab [1]{#1}%
\providecommand \enquote  [1]{``#1''}%
\providecommand \bibnamefont  [1]{#1}%
\providecommand \bibfnamefont [1]{#1}%
\providecommand \citenamefont [1]{#1}%
\providecommand \href@noop [0]{\@secondoftwo}%
\providecommand \href [0]{\begingroup \@sanitize@url \@href}%
\providecommand \@href[1]{\@@startlink{#1}\@@href}%
\providecommand \@@href[1]{\endgroup#1\@@endlink}%
\providecommand \@sanitize@url [0]{\catcode `\\12\catcode `\$12\catcode `\&12\catcode `\#12\catcode `\^12\catcode `\_12\catcode `\%12\relax}%
\providecommand \@@startlink[1]{}%
\providecommand \@@endlink[0]{}%
\providecommand \url  [0]{\begingroup\@sanitize@url \@url }%
\providecommand \@url [1]{\endgroup\@href {#1}{\urlprefix }}%
\providecommand \urlprefix  [0]{URL }%
\providecommand \Eprint [0]{\href }%
\providecommand \doibase [0]{https://doi.org/}%
\providecommand \selectlanguage [0]{\@gobble}%
\providecommand \bibinfo  [0]{\@secondoftwo}%
\providecommand \bibfield  [0]{\@secondoftwo}%
\providecommand \translation [1]{[#1]}%
\providecommand \BibitemOpen [0]{}%
\providecommand \bibitemStop [0]{}%
\providecommand \bibitemNoStop [0]{.\EOS\space}%
\providecommand \EOS [0]{\spacefactor3000\relax}%
\providecommand \BibitemShut  [1]{\csname bibitem#1\endcsname}%
\let\auto@bib@innerbib\@empty
\bibitem [{\citenamefont {Ransford}\ \emph {et~al.}(2026)\citenamefont {Ransford} \emph {et~al.}}]{ransford2026}%
  \BibitemOpen
  \bibfield  {author} {\bibinfo {author} {\bibfnamefont {A.}~\bibnamefont {Ransford}} \emph {et~al.},\ }\href {https://doi.org/10.1038/s41586-026-10676-4} {\bibfield  {journal} {\bibinfo  {journal} {Nature}\ }\textbf {\bibinfo {volume} {655}},\ \bibinfo {pages} {81} (\bibinfo {year} {2026})}\BibitemShut {NoStop}%
\bibitem [{\citenamefont {Pino}\ \emph {et~al.}(2021)\citenamefont {Pino} \emph {et~al.}}]{pino2021}%
  \BibitemOpen
  \bibfield  {author} {\bibinfo {author} {\bibfnamefont {J.~M.}\ \bibnamefont {Pino}} \emph {et~al.},\ }\href {https://doi.org/10.1038/s41586-021-03318-4} {\bibfield  {journal} {\bibinfo  {journal} {Nature}\ }\textbf {\bibinfo {volume} {592}},\ \bibinfo {pages} {209} (\bibinfo {year} {2021})}\BibitemShut {NoStop}%
\bibitem [{\citenamefont {Low}\ \emph {et~al.}(2025)\citenamefont {Low} \emph {et~al.}}]{low2025qudit}%
  \BibitemOpen
  \bibfield  {author} {\bibinfo {author} {\bibfnamefont {P.~J.}\ \bibnamefont {Low}} \emph {et~al.},\ }\href {https://doi.org/10.1038/s41534-025-01031-y} {\bibfield  {journal} {\bibinfo  {journal} {npj Quantum Information}\ }\textbf {\bibinfo {volume} {11}},\ \bibinfo {pages} {85} (\bibinfo {year} {2025})}\BibitemShut {NoStop}%
\bibitem [{\citenamefont {Ringbauer}\ \emph {et~al.}(2022)\citenamefont {Ringbauer} \emph {et~al.}}]{ringbauer2022}%
  \BibitemOpen
  \bibfield  {author} {\bibinfo {author} {\bibfnamefont {M.}~\bibnamefont {Ringbauer}} \emph {et~al.},\ }\href {https://doi.org/10.1038/s41567-022-01658-0} {\bibfield  {journal} {\bibinfo  {journal} {Nature Physics}\ }\textbf {\bibinfo {volume} {18}},\ \bibinfo {pages} {1053} (\bibinfo {year} {2022})}\BibitemShut {NoStop}%
\bibitem [{\citenamefont {MacDonell}\ \emph {et~al.}(2021)\citenamefont {MacDonell}, \citenamefont {Dickerson}, \citenamefont {Birch}, \citenamefont {Kumar}, \citenamefont {Edmunds}, \citenamefont {Biercuk}, \citenamefont {Hempel},\ and\ \citenamefont {Kassal}}]{macdonell2021}%
  \BibitemOpen
  \bibfield  {author} {\bibinfo {author} {\bibfnamefont {R.~J.}\ \bibnamefont {MacDonell}}, \bibinfo {author} {\bibfnamefont {C.~E.}\ \bibnamefont {Dickerson}}, \bibinfo {author} {\bibfnamefont {C.~J.~T.}\ \bibnamefont {Birch}}, \bibinfo {author} {\bibfnamefont {A.}~\bibnamefont {Kumar}}, \bibinfo {author} {\bibfnamefont {C.~L.}\ \bibnamefont {Edmunds}}, \bibinfo {author} {\bibfnamefont {M.~J.}\ \bibnamefont {Biercuk}}, \bibinfo {author} {\bibfnamefont {C.}~\bibnamefont {Hempel}},\ and\ \bibinfo {author} {\bibfnamefont {I.}~\bibnamefont {Kassal}},\ }\href {https://doi.org/10.1039/D1SC02142G} {\bibfield  {journal} {\bibinfo  {journal} {Chemical Science}\ }\textbf {\bibinfo {volume} {12}},\ \bibinfo {pages} {9794} (\bibinfo {year} {2021})}\BibitemShut {NoStop}%
\bibitem [{\citenamefont {Valahu}\ \emph {et~al.}(2023)\citenamefont {Valahu}, \citenamefont {Olaya-Agudelo}, \citenamefont {MacDonell}, \citenamefont {Navickas}, \citenamefont {Rao}, \citenamefont {Millican}, \citenamefont {P\'erez-S\'anchez}, \citenamefont {Yuen-Zhou}, \citenamefont {Biercuk}, \citenamefont {Hempel}, \citenamefont {Tan},\ and\ \citenamefont {Kassal}}]{valahu2023}%
  \BibitemOpen
  \bibfield  {author} {\bibinfo {author} {\bibfnamefont {C.~H.}\ \bibnamefont {Valahu}}, \bibinfo {author} {\bibfnamefont {V.~C.}\ \bibnamefont {Olaya-Agudelo}}, \bibinfo {author} {\bibfnamefont {R.~J.}\ \bibnamefont {MacDonell}}, \bibinfo {author} {\bibfnamefont {T.}~\bibnamefont {Navickas}}, \bibinfo {author} {\bibfnamefont {A.~D.}\ \bibnamefont {Rao}}, \bibinfo {author} {\bibfnamefont {M.~J.}\ \bibnamefont {Millican}}, \bibinfo {author} {\bibfnamefont {J.~B.}\ \bibnamefont {P\'erez-S\'anchez}}, \bibinfo {author} {\bibfnamefont {J.}~\bibnamefont {Yuen-Zhou}}, \bibinfo {author} {\bibfnamefont {M.~J.}\ \bibnamefont {Biercuk}}, \bibinfo {author} {\bibfnamefont {C.}~\bibnamefont {Hempel}}, \bibinfo {author} {\bibfnamefont {T.~R.}\ \bibnamefont {Tan}},\ and\ \bibinfo {author} {\bibfnamefont {I.}~\bibnamefont {Kassal}},\ }\href {https://doi.org/10.1038/s41557-023-01300-3} {\bibfield  {journal} {\bibinfo  {journal} {Nature Chemistry}\ }\textbf {\bibinfo {volume} {15}},\ \bibinfo {pages} {1503} (\bibinfo {year}
  {2023})}\BibitemShut {NoStop}%
\bibitem [{\citenamefont {Davoudi}\ \emph {et~al.}(2021)\citenamefont {Davoudi}, \citenamefont {Linke},\ and\ \citenamefont {Pagano}}]{davoudi2021}%
  \BibitemOpen
  \bibfield  {author} {\bibinfo {author} {\bibfnamefont {Z.}~\bibnamefont {Davoudi}}, \bibinfo {author} {\bibfnamefont {N.~M.}\ \bibnamefont {Linke}},\ and\ \bibinfo {author} {\bibfnamefont {G.}~\bibnamefont {Pagano}},\ }\href {https://doi.org/10.1103/PhysRevResearch.3.043072} {\bibfield  {journal} {\bibinfo  {journal} {Phys. Rev. Research}\ }\textbf {\bibinfo {volume} {3}},\ \bibinfo {pages} {043072} (\bibinfo {year} {2021})}\BibitemShut {NoStop}%
\bibitem [{\citenamefont {Pagano}\ \emph {et~al.}(2025)\citenamefont {Pagano}, \citenamefont {Adamczyk},\ and\ \citenamefont {So}}]{pagano2025}%
  \BibitemOpen
  \bibfield  {author} {\bibinfo {author} {\bibfnamefont {G.}~\bibnamefont {Pagano}}, \bibinfo {author} {\bibfnamefont {W.}~\bibnamefont {Adamczyk}},\ and\ \bibinfo {author} {\bibfnamefont {V.}~\bibnamefont {So}},\ }\bibfield  {journal} {\bibinfo  {journal} {arXiv preprint arXiv:2505.20412}\ }\href {https://doi.org/10.48550/arXiv.2505.20412} {10.48550/arXiv.2505.20412} (\bibinfo {year} {2025}),\ \Eprint {https://arxiv.org/abs/2505.20412} {arXiv:2505.20412 [quant-ph]} \BibitemShut {NoStop}%
\bibitem [{\citenamefont {So}\ \emph {et~al.}(2024)\citenamefont {So}, \citenamefont {Suganthi}, \citenamefont {Menon}, \citenamefont {Zhu}, \citenamefont {Zhuravel}, \citenamefont {Pu}, \citenamefont {Wolynes}, \citenamefont {Onuchic},\ and\ \citenamefont {Pagano}}]{so2024}%
  \BibitemOpen
  \bibfield  {author} {\bibinfo {author} {\bibfnamefont {V.}~\bibnamefont {So}}, \bibinfo {author} {\bibfnamefont {M.~D.}\ \bibnamefont {Suganthi}}, \bibinfo {author} {\bibfnamefont {A.}~\bibnamefont {Menon}}, \bibinfo {author} {\bibfnamefont {M.}~\bibnamefont {Zhu}}, \bibinfo {author} {\bibfnamefont {R.}~\bibnamefont {Zhuravel}}, \bibinfo {author} {\bibfnamefont {H.}~\bibnamefont {Pu}}, \bibinfo {author} {\bibfnamefont {P.~G.}\ \bibnamefont {Wolynes}}, \bibinfo {author} {\bibfnamefont {J.~N.}\ \bibnamefont {Onuchic}},\ and\ \bibinfo {author} {\bibfnamefont {G.}~\bibnamefont {Pagano}},\ }\href {https://doi.org/10.1126/sciadv.ads8011} {\bibfield  {journal} {\bibinfo  {journal} {Science Advances}\ }\textbf {\bibinfo {volume} {10}},\ \bibinfo {pages} {eads8011} (\bibinfo {year} {2024})}\BibitemShut {NoStop}%
\bibitem [{\citenamefont {Vu}\ \emph {et~al.}(2025)\citenamefont {Vu} \emph {et~al.}}]{vu2025}%
  \BibitemOpen
  \bibfield  {author} {\bibinfo {author} {\bibfnamefont {N.~P.}\ \bibnamefont {Vu}} \emph {et~al.},\ }\href {https://doi.org/10.1021/acs.jctc.5c00315} {\bibfield  {journal} {\bibinfo  {journal} {J. Chem. Theory Comput.}\ }\textbf {\bibinfo {volume} {21}},\ \bibinfo {pages} {6258} (\bibinfo {year} {2025})}\BibitemShut {NoStop}%
\bibitem [{\citenamefont {Liu}\ \emph {et~al.}(2026)\citenamefont {Liu} \emph {et~al.}}]{liu2026isa}%
  \BibitemOpen
  \bibfield  {author} {\bibinfo {author} {\bibfnamefont {Y.}~\bibnamefont {Liu}} \emph {et~al.},\ }\href {https://doi.org/10.1103/4rf7-9tfx} {\bibfield  {journal} {\bibinfo  {journal} {PRX Quantum}\ }\textbf {\bibinfo {volume} {7}},\ \bibinfo {pages} {010201} (\bibinfo {year} {2026})}\BibitemShut {NoStop}%
\bibitem [{\citenamefont {Crane}\ \emph {et~al.}(2024)\citenamefont {Crane}, \citenamefont {Smith}, \citenamefont {Tomesh}, \citenamefont {Eickbusch}, \citenamefont {Martyn}, \citenamefont {K{\"u}hn}, \citenamefont {Funcke}, \citenamefont {DeMarco}, \citenamefont {Chuang}, \citenamefont {Wiebe}, \citenamefont {Schuckert},\ and\ \citenamefont {Girvin}}]{crane2024}%
  \BibitemOpen
  \bibfield  {author} {\bibinfo {author} {\bibfnamefont {E.}~\bibnamefont {Crane}}, \bibinfo {author} {\bibfnamefont {K.~C.}\ \bibnamefont {Smith}}, \bibinfo {author} {\bibfnamefont {T.}~\bibnamefont {Tomesh}}, \bibinfo {author} {\bibfnamefont {A.}~\bibnamefont {Eickbusch}}, \bibinfo {author} {\bibfnamefont {J.~M.}\ \bibnamefont {Martyn}}, \bibinfo {author} {\bibfnamefont {S.}~\bibnamefont {K{\"u}hn}}, \bibinfo {author} {\bibfnamefont {L.}~\bibnamefont {Funcke}}, \bibinfo {author} {\bibfnamefont {M.~A.}\ \bibnamefont {DeMarco}}, \bibinfo {author} {\bibfnamefont {I.~L.}\ \bibnamefont {Chuang}}, \bibinfo {author} {\bibfnamefont {N.}~\bibnamefont {Wiebe}}, \bibinfo {author} {\bibfnamefont {A.}~\bibnamefont {Schuckert}},\ and\ \bibinfo {author} {\bibfnamefont {S.~M.}\ \bibnamefont {Girvin}},\ }\bibfield  {journal} {\bibinfo  {journal} {arXiv preprint arXiv:2409.03747}\ }\href {https://doi.org/10.48550/arXiv.2409.03747} {10.48550/arXiv.2409.03747} (\bibinfo {year} {2024}),\ \Eprint
  {https://arxiv.org/abs/2409.03747} {arXiv:2409.03747 [quant-ph]} \BibitemShut {NoStop}%
\bibitem [{\citenamefont {Kemper}\ \emph {et~al.}(2025)\citenamefont {Kemper}, \citenamefont {Alvertis}, \citenamefont {Asaduzzaman}, \citenamefont {Bakalov}, \citenamefont {Baron}, \citenamefont {Bierman} \emph {et~al.}}]{kemper2025}%
  \BibitemOpen
  \bibfield  {author} {\bibinfo {author} {\bibfnamefont {A.~F.}\ \bibnamefont {Kemper}}, \bibinfo {author} {\bibfnamefont {A.}~\bibnamefont {Alvertis}}, \bibinfo {author} {\bibfnamefont {M.}~\bibnamefont {Asaduzzaman}}, \bibinfo {author} {\bibfnamefont {B.~N.}\ \bibnamefont {Bakalov}}, \bibinfo {author} {\bibfnamefont {D.}~\bibnamefont {Baron}}, \bibinfo {author} {\bibfnamefont {J.}~\bibnamefont {Bierman}}, \emph {et~al.},\ }\bibfield  {journal} {\bibinfo  {journal} {arXiv preprint arXiv:2511.13882}\ }\href {https://doi.org/10.48550/arXiv.2511.13882} {10.48550/arXiv.2511.13882} (\bibinfo {year} {2025}),\ \Eprint {https://arxiv.org/abs/2511.13882} {arXiv:2511.13882 [quant-ph]} \BibitemShut {NoStop}%
\bibitem [{\citenamefont {Low}\ \emph {et~al.}(2020)\citenamefont {Low}, \citenamefont {White}, \citenamefont {Cox}, \citenamefont {Day},\ and\ \citenamefont {Senko}}]{low2020}%
  \BibitemOpen
  \bibfield  {author} {\bibinfo {author} {\bibfnamefont {P.~J.}\ \bibnamefont {Low}}, \bibinfo {author} {\bibfnamefont {B.~M.}\ \bibnamefont {White}}, \bibinfo {author} {\bibfnamefont {A.~A.}\ \bibnamefont {Cox}}, \bibinfo {author} {\bibfnamefont {M.~L.}\ \bibnamefont {Day}},\ and\ \bibinfo {author} {\bibfnamefont {C.}~\bibnamefont {Senko}},\ }\href {https://doi.org/10.1103/PhysRevResearch.2.033128} {\bibfield  {journal} {\bibinfo  {journal} {Phys. Rev. Research}\ }\textbf {\bibinfo {volume} {2}},\ \bibinfo {pages} {033128} (\bibinfo {year} {2020})}\BibitemShut {NoStop}%
\bibitem [{\citenamefont {Aksenov}\ \emph {et~al.}(2023)\citenamefont {Aksenov}, \citenamefont {Zalivako}, \citenamefont {Semerikov}, \citenamefont {Borisenko}, \citenamefont {Semenin}, \citenamefont {Sidorov}, \citenamefont {Fedorov}, \citenamefont {Khabarova},\ and\ \citenamefont {Kolachevsky}}]{aksenov2023}%
  \BibitemOpen
  \bibfield  {author} {\bibinfo {author} {\bibfnamefont {M.~A.}\ \bibnamefont {Aksenov}}, \bibinfo {author} {\bibfnamefont {I.~V.}\ \bibnamefont {Zalivako}}, \bibinfo {author} {\bibfnamefont {I.~A.}\ \bibnamefont {Semerikov}}, \bibinfo {author} {\bibfnamefont {A.~S.}\ \bibnamefont {Borisenko}}, \bibinfo {author} {\bibfnamefont {N.~V.}\ \bibnamefont {Semenin}}, \bibinfo {author} {\bibfnamefont {P.~L.}\ \bibnamefont {Sidorov}}, \bibinfo {author} {\bibfnamefont {A.~K.}\ \bibnamefont {Fedorov}}, \bibinfo {author} {\bibfnamefont {K.~Y.}\ \bibnamefont {Khabarova}},\ and\ \bibinfo {author} {\bibfnamefont {N.~N.}\ \bibnamefont {Kolachevsky}},\ }\href {https://doi.org/10.1103/PhysRevA.107.052612} {\bibfield  {journal} {\bibinfo  {journal} {Phys. Rev. A}\ }\textbf {\bibinfo {volume} {107}},\ \bibinfo {pages} {052612} (\bibinfo {year} {2023})}\BibitemShut {NoStop}%
\bibitem [{\citenamefont {Wineland}\ \emph {et~al.}(1998)\citenamefont {Wineland}, \citenamefont {Monroe}, \citenamefont {Itano}, \citenamefont {Leibfried}, \citenamefont {King},\ and\ \citenamefont {Meekhof}}]{wineland1998}%
  \BibitemOpen
  \bibfield  {author} {\bibinfo {author} {\bibfnamefont {D.~J.}\ \bibnamefont {Wineland}}, \bibinfo {author} {\bibfnamefont {C.}~\bibnamefont {Monroe}}, \bibinfo {author} {\bibfnamefont {W.~M.}\ \bibnamefont {Itano}}, \bibinfo {author} {\bibfnamefont {D.}~\bibnamefont {Leibfried}}, \bibinfo {author} {\bibfnamefont {B.~E.}\ \bibnamefont {King}},\ and\ \bibinfo {author} {\bibfnamefont {D.~M.}\ \bibnamefont {Meekhof}},\ }\href {https://doi.org/10.6028/jres.103.019} {\bibfield  {journal} {\bibinfo  {journal} {J. Res. Natl. Inst. Stand. Technol.}\ }\textbf {\bibinfo {volume} {103}},\ \bibinfo {pages} {259} (\bibinfo {year} {1998})}\BibitemShut {NoStop}%
\bibitem [{\citenamefont {Leibfried}\ \emph {et~al.}(2003)\citenamefont {Leibfried}, \citenamefont {Blatt}, \citenamefont {Monroe},\ and\ \citenamefont {Wineland}}]{leibfried2003}%
  \BibitemOpen
  \bibfield  {author} {\bibinfo {author} {\bibfnamefont {D.}~\bibnamefont {Leibfried}}, \bibinfo {author} {\bibfnamefont {R.}~\bibnamefont {Blatt}}, \bibinfo {author} {\bibfnamefont {C.}~\bibnamefont {Monroe}},\ and\ \bibinfo {author} {\bibfnamefont {D.}~\bibnamefont {Wineland}},\ }\href {https://doi.org/10.1103/RevModPhys.75.281} {\bibfield  {journal} {\bibinfo  {journal} {Rev. Mod. Phys.}\ }\textbf {\bibinfo {volume} {75}},\ \bibinfo {pages} {281} (\bibinfo {year} {2003})}\BibitemShut {NoStop}%
\bibitem [{\citenamefont {Sutherland}\ and\ \citenamefont {Srinivas}(2021)}]{sutherland2021}%
  \BibitemOpen
  \bibfield  {author} {\bibinfo {author} {\bibfnamefont {R.~T.}\ \bibnamefont {Sutherland}}\ and\ \bibinfo {author} {\bibfnamefont {R.}~\bibnamefont {Srinivas}},\ }\href {https://doi.org/10.1103/PhysRevA.104.032609} {\bibfield  {journal} {\bibinfo  {journal} {Phys. Rev. A}\ }\textbf {\bibinfo {volume} {104}},\ \bibinfo {pages} {032609} (\bibinfo {year} {2021})}\BibitemShut {NoStop}%
\bibitem [{\citenamefont {M{\o}lmer}\ and\ \citenamefont {S{\o}rensen}(1999)}]{molmer1999}%
  \BibitemOpen
  \bibfield  {author} {\bibinfo {author} {\bibfnamefont {K.}~\bibnamefont {M{\o}lmer}}\ and\ \bibinfo {author} {\bibfnamefont {A.}~\bibnamefont {S{\o}rensen}},\ }\href {https://doi.org/10.1103/PhysRevLett.82.1835} {\bibfield  {journal} {\bibinfo  {journal} {Phys. Rev. Lett.}\ }\textbf {\bibinfo {volume} {82}},\ \bibinfo {pages} {1835} (\bibinfo {year} {1999})}\BibitemShut {NoStop}%
\bibitem [{\citenamefont {S{\o}rensen}\ and\ \citenamefont {M{\o}lmer}(2000)}]{sorensen2000}%
  \BibitemOpen
  \bibfield  {author} {\bibinfo {author} {\bibfnamefont {A.}~\bibnamefont {S{\o}rensen}}\ and\ \bibinfo {author} {\bibfnamefont {K.}~\bibnamefont {M{\o}lmer}},\ }\href {https://doi.org/10.1103/PhysRevA.62.022311} {\bibfield  {journal} {\bibinfo  {journal} {Phys. Rev. A}\ }\textbf {\bibinfo {volume} {62}},\ \bibinfo {pages} {022311} (\bibinfo {year} {2000})}\BibitemShut {NoStop}%
\bibitem [{\citenamefont {Lloyd}\ and\ \citenamefont {Braunstein}(1999)}]{lloyd1999}%
  \BibitemOpen
  \bibfield  {author} {\bibinfo {author} {\bibfnamefont {S.}~\bibnamefont {Lloyd}}\ and\ \bibinfo {author} {\bibfnamefont {S.~L.}\ \bibnamefont {Braunstein}},\ }\href {https://doi.org/10.1103/PhysRevLett.82.1784} {\bibfield  {journal} {\bibinfo  {journal} {Phys. Rev. Lett.}\ }\textbf {\bibinfo {volume} {82}},\ \bibinfo {pages} {1784} (\bibinfo {year} {1999})}\BibitemShut {NoStop}%
\bibitem [{\citenamefont {Brif}\ \emph {et~al.}(2010)\citenamefont {Brif}, \citenamefont {Chakrabarti},\ and\ \citenamefont {Rabitz}}]{brif2010}%
  \BibitemOpen
  \bibfield  {author} {\bibinfo {author} {\bibfnamefont {C.}~\bibnamefont {Brif}}, \bibinfo {author} {\bibfnamefont {R.}~\bibnamefont {Chakrabarti}},\ and\ \bibinfo {author} {\bibfnamefont {H.}~\bibnamefont {Rabitz}},\ }\href {https://doi.org/10.1088/1367-2630/12/7/075008} {\bibfield  {journal} {\bibinfo  {journal} {New J. Phys.}\ }\textbf {\bibinfo {volume} {12}},\ \bibinfo {pages} {075008} (\bibinfo {year} {2010})}\BibitemShut {NoStop}%
\bibitem [{\citenamefont {Krastanov}\ \emph {et~al.}(2015)\citenamefont {Krastanov} \emph {et~al.}}]{krastanov2015}%
  \BibitemOpen
  \bibfield  {author} {\bibinfo {author} {\bibfnamefont {S.}~\bibnamefont {Krastanov}} \emph {et~al.},\ }\href {https://doi.org/10.1103/PhysRevA.92.040303} {\bibfield  {journal} {\bibinfo  {journal} {Phys. Rev. A}\ }\textbf {\bibinfo {volume} {92}},\ \bibinfo {pages} {040303} (\bibinfo {year} {2015})}\BibitemShut {NoStop}%
\bibitem [{\citenamefont {Heeres}\ \emph {et~al.}(2015)\citenamefont {Heeres}, \citenamefont {Vlastakis}, \citenamefont {Holland}, \citenamefont {Krastanov}, \citenamefont {Albert}, \citenamefont {Frunzio}, \citenamefont {Jiang},\ and\ \citenamefont {Schoelkopf}}]{heeres2015}%
  \BibitemOpen
  \bibfield  {author} {\bibinfo {author} {\bibfnamefont {R.~W.}\ \bibnamefont {Heeres}}, \bibinfo {author} {\bibfnamefont {B.}~\bibnamefont {Vlastakis}}, \bibinfo {author} {\bibfnamefont {E.}~\bibnamefont {Holland}}, \bibinfo {author} {\bibfnamefont {S.}~\bibnamefont {Krastanov}}, \bibinfo {author} {\bibfnamefont {V.~V.}\ \bibnamefont {Albert}}, \bibinfo {author} {\bibfnamefont {L.}~\bibnamefont {Frunzio}}, \bibinfo {author} {\bibfnamefont {L.}~\bibnamefont {Jiang}},\ and\ \bibinfo {author} {\bibfnamefont {R.~J.}\ \bibnamefont {Schoelkopf}},\ }\href {https://doi.org/10.1103/PhysRevLett.115.137002} {\bibfield  {journal} {\bibinfo  {journal} {Phys. Rev. Lett.}\ }\textbf {\bibinfo {volume} {115}},\ \bibinfo {pages} {137002} (\bibinfo {year} {2015})}\BibitemShut {NoStop}%
\bibitem [{\citenamefont {Eickbusch}\ \emph {et~al.}(2022)\citenamefont {Eickbusch}, \citenamefont {Sivak}, \citenamefont {Ding}, \citenamefont {Elder}, \citenamefont {Jha}, \citenamefont {Venkatraman}, \citenamefont {Royer}, \citenamefont {Girvin}, \citenamefont {Schoelkopf},\ and\ \citenamefont {Devoret}}]{eickbusch2022}%
  \BibitemOpen
  \bibfield  {author} {\bibinfo {author} {\bibfnamefont {A.}~\bibnamefont {Eickbusch}}, \bibinfo {author} {\bibfnamefont {V.}~\bibnamefont {Sivak}}, \bibinfo {author} {\bibfnamefont {A.~Z.}\ \bibnamefont {Ding}}, \bibinfo {author} {\bibfnamefont {S.~S.}\ \bibnamefont {Elder}}, \bibinfo {author} {\bibfnamefont {S.~R.}\ \bibnamefont {Jha}}, \bibinfo {author} {\bibfnamefont {J.}~\bibnamefont {Venkatraman}}, \bibinfo {author} {\bibfnamefont {B.}~\bibnamefont {Royer}}, \bibinfo {author} {\bibfnamefont {S.~M.}\ \bibnamefont {Girvin}}, \bibinfo {author} {\bibfnamefont {R.~J.}\ \bibnamefont {Schoelkopf}},\ and\ \bibinfo {author} {\bibfnamefont {M.~H.}\ \bibnamefont {Devoret}},\ }\href {https://doi.org/10.1038/s41567-022-01776-9} {\bibfield  {journal} {\bibinfo  {journal} {Nat. Phys.}\ }\textbf {\bibinfo {volume} {18}},\ \bibinfo {pages} {1464} (\bibinfo {year} {2022})}\BibitemShut {NoStop}%
\bibitem [{\citenamefont {Katz}\ and\ \citenamefont {Monroe}(2023)}]{katz2023}%
  \BibitemOpen
  \bibfield  {author} {\bibinfo {author} {\bibfnamefont {O.}~\bibnamefont {Katz}}\ and\ \bibinfo {author} {\bibfnamefont {C.}~\bibnamefont {Monroe}},\ }\href {https://doi.org/10.1103/PhysRevLett.131.033604} {\bibfield  {journal} {\bibinfo  {journal} {Phys. Rev. Lett.}\ }\textbf {\bibinfo {volume} {131}},\ \bibinfo {pages} {033604} (\bibinfo {year} {2023})}\BibitemShut {NoStop}%
\bibitem [{\citenamefont {Gan}\ \emph {et~al.}(2020)\citenamefont {Gan}, \citenamefont {Maslennikov}, \citenamefont {Tseng}, \citenamefont {Nguyen},\ and\ \citenamefont {Matsukevich}}]{gan2020}%
  \BibitemOpen
  \bibfield  {author} {\bibinfo {author} {\bibfnamefont {H.~C.~J.}\ \bibnamefont {Gan}}, \bibinfo {author} {\bibfnamefont {G.}~\bibnamefont {Maslennikov}}, \bibinfo {author} {\bibfnamefont {K.-W.}\ \bibnamefont {Tseng}}, \bibinfo {author} {\bibfnamefont {C.}~\bibnamefont {Nguyen}},\ and\ \bibinfo {author} {\bibfnamefont {D.}~\bibnamefont {Matsukevich}},\ }\href {https://doi.org/10.1103/PhysRevLett.124.170502} {\bibfield  {journal} {\bibinfo  {journal} {Phys. Rev. Lett.}\ }\textbf {\bibinfo {volume} {124}},\ \bibinfo {pages} {170502} (\bibinfo {year} {2020})}\BibitemShut {NoStop}%
\bibitem [{\citenamefont {Fl{\"u}hmann}\ \emph {et~al.}(2019)\citenamefont {Fl{\"u}hmann}, \citenamefont {Nguyen}, \citenamefont {Marinelli}, \citenamefont {Negnevitsky}, \citenamefont {Mehta},\ and\ \citenamefont {Home}}]{fluhmann2019}%
  \BibitemOpen
  \bibfield  {author} {\bibinfo {author} {\bibfnamefont {C.}~\bibnamefont {Fl{\"u}hmann}}, \bibinfo {author} {\bibfnamefont {T.~L.}\ \bibnamefont {Nguyen}}, \bibinfo {author} {\bibfnamefont {M.}~\bibnamefont {Marinelli}}, \bibinfo {author} {\bibfnamefont {V.}~\bibnamefont {Negnevitsky}}, \bibinfo {author} {\bibfnamefont {K.}~\bibnamefont {Mehta}},\ and\ \bibinfo {author} {\bibfnamefont {J.~P.}\ \bibnamefont {Home}},\ }\href {https://doi.org/10.1038/s41586-019-0960-6} {\bibfield  {journal} {\bibinfo  {journal} {Nature}\ }\textbf {\bibinfo {volume} {566}},\ \bibinfo {pages} {513} (\bibinfo {year} {2019})}\BibitemShut {NoStop}%
\bibitem [{\citenamefont {Zalivako}\ \emph {et~al.}(2025)\citenamefont {Zalivako} \emph {et~al.}}]{zalivako2024}%
  \BibitemOpen
  \bibfield  {author} {\bibinfo {author} {\bibfnamefont {I.~V.}\ \bibnamefont {Zalivako}} \emph {et~al.},\ }\href {https://doi.org/10.3390/quantum7020019} {\bibfield  {journal} {\bibinfo  {journal} {Quantum Reports}\ }\textbf {\bibinfo {volume} {7}},\ \bibinfo {pages} {19} (\bibinfo {year} {2025})}\BibitemShut {NoStop}%
\bibitem [{\citenamefont {Hrmo}\ \emph {et~al.}(2023)\citenamefont {Hrmo}, \citenamefont {Wilhelm}, \citenamefont {Gerster}, \citenamefont {van Mourik}, \citenamefont {Huber}, \citenamefont {Blatt}, \citenamefont {Schindler}, \citenamefont {Monz},\ and\ \citenamefont {Ringbauer}}]{hrmo2023}%
  \BibitemOpen
  \bibfield  {author} {\bibinfo {author} {\bibfnamefont {P.}~\bibnamefont {Hrmo}}, \bibinfo {author} {\bibfnamefont {B.}~\bibnamefont {Wilhelm}}, \bibinfo {author} {\bibfnamefont {L.}~\bibnamefont {Gerster}}, \bibinfo {author} {\bibfnamefont {M.~W.}\ \bibnamefont {van Mourik}}, \bibinfo {author} {\bibfnamefont {M.}~\bibnamefont {Huber}}, \bibinfo {author} {\bibfnamefont {R.}~\bibnamefont {Blatt}}, \bibinfo {author} {\bibfnamefont {P.}~\bibnamefont {Schindler}}, \bibinfo {author} {\bibfnamefont {T.}~\bibnamefont {Monz}},\ and\ \bibinfo {author} {\bibfnamefont {M.}~\bibnamefont {Ringbauer}},\ }\href {https://doi.org/10.1038/s41467-023-37375-2} {\bibfield  {journal} {\bibinfo  {journal} {Nat. Commun.}\ }\textbf {\bibinfo {volume} {14}},\ \bibinfo {pages} {2242} (\bibinfo {year} {2023})}\BibitemShut {NoStop}%
\bibitem [{\citenamefont {Shi}\ \emph {et~al.}(2026)\citenamefont {Shi}, \citenamefont {Sinanan-Singh}, \citenamefont {Burke}, \citenamefont {Chiaverini},\ and\ \citenamefont {Chuang}}]{shi2025}%
  \BibitemOpen
  \bibfield  {author} {\bibinfo {author} {\bibfnamefont {X.}~\bibnamefont {Shi}}, \bibinfo {author} {\bibfnamefont {J.}~\bibnamefont {Sinanan-Singh}}, \bibinfo {author} {\bibfnamefont {T.~J.}\ \bibnamefont {Burke}}, \bibinfo {author} {\bibfnamefont {J.}~\bibnamefont {Chiaverini}},\ and\ \bibinfo {author} {\bibfnamefont {I.~L.}\ \bibnamefont {Chuang}},\ }\href {https://doi.org/10.1038/s41467-026-68746-0} {\bibfield  {journal} {\bibinfo  {journal} {Nature Communications}\ }\textbf {\bibinfo {volume} {17}},\ \bibinfo {pages} {1911} (\bibinfo {year} {2026})}\BibitemShut {NoStop}%
\bibitem [{\citenamefont {Kiktenko}\ \emph {et~al.}(2025)\citenamefont {Kiktenko}, \citenamefont {Nikolaeva},\ and\ \citenamefont {Fedorov}}]{kiktenko2025}%
  \BibitemOpen
  \bibfield  {author} {\bibinfo {author} {\bibfnamefont {E.~O.}\ \bibnamefont {Kiktenko}}, \bibinfo {author} {\bibfnamefont {A.~S.}\ \bibnamefont {Nikolaeva}},\ and\ \bibinfo {author} {\bibfnamefont {A.~K.}\ \bibnamefont {Fedorov}},\ }\href {https://doi.org/10.1103/RevModPhys.97.021003} {\bibfield  {journal} {\bibinfo  {journal} {Rev. Mod. Phys.}\ }\textbf {\bibinfo {volume} {97}},\ \bibinfo {pages} {021003} (\bibinfo {year} {2025})}\BibitemShut {NoStop}%
\bibitem [{\citenamefont {Chakraborty}\ and\ \citenamefont {Albert}(2026)}]{chakraborty2026}%
  \BibitemOpen
  \bibfield  {author} {\bibinfo {author} {\bibfnamefont {S.}~\bibnamefont {Chakraborty}}\ and\ \bibinfo {author} {\bibfnamefont {V.~V.}\ \bibnamefont {Albert}},\ }\href {https://doi.org/10.1103/yr64-ypxt} {\bibfield  {journal} {\bibinfo  {journal} {PRX Quantum}\ }\textbf {\bibinfo {volume} {7}},\ \bibinfo {pages} {020320} (\bibinfo {year} {2026})}\BibitemShut {NoStop}%
\bibitem [{\citenamefont {Furches}\ \emph {et~al.}(2026)\citenamefont {Furches}, \citenamefont {Stavenger},\ and\ \citenamefont {Ortiz~Marrero}}]{liu_hybridlane2025}%
  \BibitemOpen
  \bibfield  {author} {\bibinfo {author} {\bibfnamefont {J.}~\bibnamefont {Furches}}, \bibinfo {author} {\bibfnamefont {T.~J.}\ \bibnamefont {Stavenger}},\ and\ \bibinfo {author} {\bibfnamefont {C.}~\bibnamefont {Ortiz~Marrero}},\ }\href {https://doi.org/10.48550/arXiv.2603.10919} {\bibinfo {title} {Hybridlane: A software development kit for hybrid continuous-discrete variable quantum computing}} (\bibinfo {year} {2026}),\ \Eprint {https://arxiv.org/abs/2603.10919} {arXiv:2603.10919 [quant-ph]} \BibitemShut {NoStop}%
\bibitem [{\citenamefont {Bergholm}\ \emph {et~al.}(2018)\citenamefont {Bergholm}, \citenamefont {Izaac}, \citenamefont {Schuld}, \citenamefont {Gogolin}, \citenamefont {Ahmed}, \citenamefont {Ajith}, \citenamefont {Alam} \emph {et~al.}}]{pennylane2018}%
  \BibitemOpen
  \bibfield  {author} {\bibinfo {author} {\bibfnamefont {V.}~\bibnamefont {Bergholm}}, \bibinfo {author} {\bibfnamefont {J.}~\bibnamefont {Izaac}}, \bibinfo {author} {\bibfnamefont {M.}~\bibnamefont {Schuld}}, \bibinfo {author} {\bibfnamefont {C.}~\bibnamefont {Gogolin}}, \bibinfo {author} {\bibfnamefont {S.}~\bibnamefont {Ahmed}}, \bibinfo {author} {\bibfnamefont {V.}~\bibnamefont {Ajith}}, \bibinfo {author} {\bibfnamefont {M.~S.}\ \bibnamefont {Alam}}, \emph {et~al.},\ }\bibfield  {journal} {\bibinfo  {journal} {arXiv preprint arXiv:1811.04968}\ }\href {https://doi.org/10.48550/arXiv.1811.04968} {10.48550/arXiv.1811.04968} (\bibinfo {year} {2018}),\ \Eprint {https://arxiv.org/abs/1811.04968} {arXiv:1811.04968 [quant-ph]} \BibitemShut {NoStop}%
\bibitem [{\citenamefont {B\u{a}z\u{a}van}\ \emph {et~al.}(2023)\citenamefont {B\u{a}z\u{a}van}, \citenamefont {Saner}, \citenamefont {Minder}, \citenamefont {Hughes}, \citenamefont {Sutherland}, \citenamefont {Lucas}, \citenamefont {Srinivas},\ and\ \citenamefont {Ballance}}]{bazavan2023}%
  \BibitemOpen
  \bibfield  {author} {\bibinfo {author} {\bibfnamefont {O.}~\bibnamefont {B\u{a}z\u{a}van}}, \bibinfo {author} {\bibfnamefont {S.}~\bibnamefont {Saner}}, \bibinfo {author} {\bibfnamefont {M.}~\bibnamefont {Minder}}, \bibinfo {author} {\bibfnamefont {A.~C.}\ \bibnamefont {Hughes}}, \bibinfo {author} {\bibfnamefont {R.~T.}\ \bibnamefont {Sutherland}}, \bibinfo {author} {\bibfnamefont {D.~M.}\ \bibnamefont {Lucas}}, \bibinfo {author} {\bibfnamefont {R.}~\bibnamefont {Srinivas}},\ and\ \bibinfo {author} {\bibfnamefont {C.~J.}\ \bibnamefont {Ballance}},\ }\href {https://doi.org/10.1103/PhysRevA.107.022617} {\bibfield  {journal} {\bibinfo  {journal} {Physical Review A}\ }\textbf {\bibinfo {volume} {107}},\ \bibinfo {pages} {022617} (\bibinfo {year} {2023})}\BibitemShut {NoStop}%
\bibitem [{\citenamefont {Choi}\ \emph {et~al.}(2026)\citenamefont {Choi}, \citenamefont {Kang}, \citenamefont {Kim}, \citenamefont {You}, \citenamefont {Lee},\ and\ \citenamefont {Kim}}]{choi2026}%
  \BibitemOpen
  \bibfield  {author} {\bibinfo {author} {\bibfnamefont {W.}~\bibnamefont {Choi}}, \bibinfo {author} {\bibfnamefont {J.}~\bibnamefont {Kang}}, \bibinfo {author} {\bibfnamefont {K.}~\bibnamefont {Kim}}, \bibinfo {author} {\bibfnamefont {J.}~\bibnamefont {You}}, \bibinfo {author} {\bibfnamefont {K.}~\bibnamefont {Lee}},\ and\ \bibinfo {author} {\bibfnamefont {T.}~\bibnamefont {Kim}},\ }\bibfield  {journal} {\bibinfo  {journal} {arXiv preprint arXiv:2601.04591}\ }\href {https://doi.org/10.48550/arXiv.2601.04591} {10.48550/arXiv.2601.04591} (\bibinfo {year} {2026}),\ \Eprint {https://arxiv.org/abs/2601.04591} {arXiv:2601.04591 [quant-ph]} \BibitemShut {NoStop}%
\bibitem [{\citenamefont {McKay}\ \emph {et~al.}(2017)\citenamefont {McKay}, \citenamefont {Wood}, \citenamefont {Sheldon}, \citenamefont {Chow},\ and\ \citenamefont {Gambetta}}]{mckay2017}%
  \BibitemOpen
  \bibfield  {author} {\bibinfo {author} {\bibfnamefont {D.~C.}\ \bibnamefont {McKay}}, \bibinfo {author} {\bibfnamefont {C.~J.}\ \bibnamefont {Wood}}, \bibinfo {author} {\bibfnamefont {S.}~\bibnamefont {Sheldon}}, \bibinfo {author} {\bibfnamefont {J.~M.}\ \bibnamefont {Chow}},\ and\ \bibinfo {author} {\bibfnamefont {J.~M.}\ \bibnamefont {Gambetta}},\ }\href {https://doi.org/10.1103/PhysRevA.96.022330} {\bibfield  {journal} {\bibinfo  {journal} {Phys. Rev. A}\ }\textbf {\bibinfo {volume} {96}},\ \bibinfo {pages} {022330} (\bibinfo {year} {2017})}\BibitemShut {NoStop}%
\bibitem [{\citenamefont {Even-Haim}\ \emph {et~al.}(2025)\citenamefont {Even-Haim}, \citenamefont {Diringer}, \citenamefont {Ruimy}, \citenamefont {Baranes}, \citenamefont {Gorlach}, \citenamefont {Hacohen-Gourgy},\ and\ \citenamefont {Kaminer}}]{evenhaim2025}%
  \BibitemOpen
  \bibfield  {author} {\bibinfo {author} {\bibfnamefont {S.}~\bibnamefont {Even-Haim}}, \bibinfo {author} {\bibfnamefont {A.~A.}\ \bibnamefont {Diringer}}, \bibinfo {author} {\bibfnamefont {R.}~\bibnamefont {Ruimy}}, \bibinfo {author} {\bibfnamefont {G.}~\bibnamefont {Baranes}}, \bibinfo {author} {\bibfnamefont {A.}~\bibnamefont {Gorlach}}, \bibinfo {author} {\bibfnamefont {S.}~\bibnamefont {Hacohen-Gourgy}},\ and\ \bibinfo {author} {\bibfnamefont {I.}~\bibnamefont {Kaminer}},\ }\href {https://doi.org/10.22331/q-2025-12-10-1934} {\bibfield  {journal} {\bibinfo  {journal} {Quantum}\ }\textbf {\bibinfo {volume} {9}},\ \bibinfo {pages} {1934} (\bibinfo {year} {2025})}\BibitemShut {NoStop}%
\bibitem [{\citenamefont {Meth}\ \emph {et~al.}(2025)\citenamefont {Meth}, \citenamefont {Zhang}, \citenamefont {Haase}, \citenamefont {Edmunds}, \citenamefont {Postler}, \citenamefont {Jena}, \citenamefont {Steiner}, \citenamefont {Dellantonio}, \citenamefont {Blatt}, \citenamefont {Zoller}, \citenamefont {Monz}, \citenamefont {Schindler}, \citenamefont {Muschik},\ and\ \citenamefont {Ringbauer}}]{meth2025}%
  \BibitemOpen
  \bibfield  {author} {\bibinfo {author} {\bibfnamefont {M.}~\bibnamefont {Meth}}, \bibinfo {author} {\bibfnamefont {J.}~\bibnamefont {Zhang}}, \bibinfo {author} {\bibfnamefont {J.~F.}\ \bibnamefont {Haase}}, \bibinfo {author} {\bibfnamefont {C.}~\bibnamefont {Edmunds}}, \bibinfo {author} {\bibfnamefont {L.}~\bibnamefont {Postler}}, \bibinfo {author} {\bibfnamefont {A.~J.}\ \bibnamefont {Jena}}, \bibinfo {author} {\bibfnamefont {A.}~\bibnamefont {Steiner}}, \bibinfo {author} {\bibfnamefont {L.}~\bibnamefont {Dellantonio}}, \bibinfo {author} {\bibfnamefont {R.}~\bibnamefont {Blatt}}, \bibinfo {author} {\bibfnamefont {P.}~\bibnamefont {Zoller}}, \bibinfo {author} {\bibfnamefont {T.}~\bibnamefont {Monz}}, \bibinfo {author} {\bibfnamefont {P.}~\bibnamefont {Schindler}}, \bibinfo {author} {\bibfnamefont {C.}~\bibnamefont {Muschik}},\ and\ \bibinfo {author} {\bibfnamefont {M.}~\bibnamefont {Ringbauer}},\ }\href {https://doi.org/10.1038/s41567-025-02797-w} {\bibfield  {journal} {\bibinfo  {journal} {Nature Physics}\
  }\textbf {\bibinfo {volume} {21}},\ \bibinfo {pages} {570} (\bibinfo {year} {2025})}\BibitemShut {NoStop}%
\bibitem [{\citenamefont {Johansson}\ \emph {et~al.}(2012)\citenamefont {Johansson}, \citenamefont {Nation},\ and\ \citenamefont {Nori}}]{qutip2012}%
  \BibitemOpen
  \bibfield  {author} {\bibinfo {author} {\bibfnamefont {J.~R.}\ \bibnamefont {Johansson}}, \bibinfo {author} {\bibfnamefont {P.~D.}\ \bibnamefont {Nation}},\ and\ \bibinfo {author} {\bibfnamefont {F.}~\bibnamefont {Nori}},\ }\href {https://doi.org/10.1016/j.cpc.2012.02.021} {\bibfield  {journal} {\bibinfo  {journal} {Comput. Phys. Commun.}\ }\textbf {\bibinfo {volume} {183}},\ \bibinfo {pages} {1760} (\bibinfo {year} {2012})}\BibitemShut {NoStop}%
\bibitem [{\citenamefont {Johansson}\ \emph {et~al.}(2013)\citenamefont {Johansson}, \citenamefont {Nation},\ and\ \citenamefont {Nori}}]{qutip2013}%
  \BibitemOpen
  \bibfield  {author} {\bibinfo {author} {\bibfnamefont {J.~R.}\ \bibnamefont {Johansson}}, \bibinfo {author} {\bibfnamefont {P.~D.}\ \bibnamefont {Nation}},\ and\ \bibinfo {author} {\bibfnamefont {F.}~\bibnamefont {Nori}},\ }\href {https://doi.org/10.1016/j.cpc.2012.11.019} {\bibfield  {journal} {\bibinfo  {journal} {Comput. Phys. Commun.}\ }\textbf {\bibinfo {volume} {184}},\ \bibinfo {pages} {1234} (\bibinfo {year} {2013})}\BibitemShut {NoStop}%
\bibitem [{\citenamefont {Lindblad}(1976)}]{lindblad1976}%
  \BibitemOpen
  \bibfield  {author} {\bibinfo {author} {\bibfnamefont {G.}~\bibnamefont {Lindblad}},\ }\href {https://doi.org/10.1007/BF01608499} {\bibfield  {journal} {\bibinfo  {journal} {Commun. Math. Phys.}\ }\textbf {\bibinfo {volume} {48}},\ \bibinfo {pages} {119} (\bibinfo {year} {1976})}\BibitemShut {NoStop}%
\bibitem [{\citenamefont {Stinespring}(1955)}]{stinespring1955}%
  \BibitemOpen
  \bibfield  {author} {\bibinfo {author} {\bibfnamefont {W.~F.}\ \bibnamefont {Stinespring}},\ }\href {https://doi.org/10.1090/S0002-9939-1955-0069403-4} {\bibfield  {journal} {\bibinfo  {journal} {Proc. Amer. Math. Soc.}\ }\textbf {\bibinfo {volume} {6}},\ \bibinfo {pages} {211} (\bibinfo {year} {1955})}\BibitemShut {NoStop}%
\bibitem [{\citenamefont {Hall}(2015)}]{hall2015}%
  \BibitemOpen
  \bibfield  {author} {\bibinfo {author} {\bibfnamefont {B.~C.}\ \bibnamefont {Hall}},\ }\href {https://doi.org/10.1007/978-3-319-13467-3} {\emph {\bibinfo {title} {Lie Groups, Lie Algebras, and Representations: An Elementary Introduction}}},\ \bibinfo {edition} {2nd}\ ed.,\ \bibinfo {series} {Graduate Texts in Mathematics}, Vol.\ \bibinfo {volume} {222}\ (\bibinfo  {publisher} {Springer},\ \bibinfo {address} {Cham},\ \bibinfo {year} {2015})\BibitemShut {NoStop}%
\bibitem [{\citenamefont {Hakimi~Heris}\ \emph {et~al.}(2026)\citenamefont {Hakimi~Heris}, \citenamefont {Liu},\ and\ \citenamefont {Mueller}}]{hakimi2026hypulse}%
  \BibitemOpen
  \bibfield  {author} {\bibinfo {author} {\bibfnamefont {M.}~\bibnamefont {Hakimi~Heris}}, \bibinfo {author} {\bibfnamefont {Y.}~\bibnamefont {Liu}},\ and\ \bibinfo {author} {\bibfnamefont {F.}~\bibnamefont {Mueller}},\ }\bibfield  {journal} {\bibinfo  {journal} {arXiv preprint arXiv:2604.26804}\ }\href {https://doi.org/10.48550/arXiv.2604.26804} {10.48550/arXiv.2604.26804} (\bibinfo {year} {2026}),\ \Eprint {https://arxiv.org/abs/2604.26804} {arXiv:2604.26804 [quant-ph]} \BibitemShut {NoStop}%
\bibitem [{\citenamefont {Schindler}\ \emph {et~al.}(2013)\citenamefont {Schindler}, \citenamefont {Nigg}, \citenamefont {Monz}, \citenamefont {Barreiro}, \citenamefont {Martinez}, \citenamefont {Wang}, \citenamefont {Quint}, \citenamefont {Brandl}, \citenamefont {Nebendahl}, \citenamefont {Roos}, \citenamefont {Chwalla}, \citenamefont {Hennrich},\ and\ \citenamefont {Blatt}}]{schindler2013}%
  \BibitemOpen
  \bibfield  {author} {\bibinfo {author} {\bibfnamefont {P.}~\bibnamefont {Schindler}}, \bibinfo {author} {\bibfnamefont {D.}~\bibnamefont {Nigg}}, \bibinfo {author} {\bibfnamefont {T.}~\bibnamefont {Monz}}, \bibinfo {author} {\bibfnamefont {J.~T.}\ \bibnamefont {Barreiro}}, \bibinfo {author} {\bibfnamefont {E.}~\bibnamefont {Martinez}}, \bibinfo {author} {\bibfnamefont {S.~X.}\ \bibnamefont {Wang}}, \bibinfo {author} {\bibfnamefont {S.}~\bibnamefont {Quint}}, \bibinfo {author} {\bibfnamefont {M.~F.}\ \bibnamefont {Brandl}}, \bibinfo {author} {\bibfnamefont {V.}~\bibnamefont {Nebendahl}}, \bibinfo {author} {\bibfnamefont {C.~F.}\ \bibnamefont {Roos}}, \bibinfo {author} {\bibfnamefont {M.}~\bibnamefont {Chwalla}}, \bibinfo {author} {\bibfnamefont {M.}~\bibnamefont {Hennrich}},\ and\ \bibinfo {author} {\bibfnamefont {R.}~\bibnamefont {Blatt}},\ }\href {https://doi.org/10.1088/1367-2630/15/12/123012} {\bibfield  {journal} {\bibinfo  {journal} {New J. Phys.}\ }\textbf {\bibinfo {volume} {15}},\ \bibinfo {pages}
  {123012} (\bibinfo {year} {2013})}\BibitemShut {NoStop}%
\end{thebibliography}%


\appendix
\section{Proof of Universality}
\label{app:universality}

We prove the claim stated in Sec.~\ref{sec:universality}, namely the generators of $X_E$, $Z_d$, and $C_{\Pi_z}D$ alone, together with their repeated commutators, generate Hamiltonians formed from arbitrary qudit operators tensored with polynomials in the quadratures of all oscillator modes. No mode-coupling primitive is required.

\subsection{Setup and Notation}
Fix two qudit levels $\ket{a},\ket{b}$ and let $\sigma_1,\sigma_2,\sigma_3$ denote the Pauli operators on that two-dimensional subspace, embedded in $\mathbb{C}^d$ and acting as zero on the remaining levels. Pillar~1 (Sec.~\ref{sec:sud}) gives all of $\mathfrak{su}(d)$ on the discrete register, and in particular all three $\sigma_k$. From $C_{\Pi_z}D$, Eq.~\eqref{eq:cd_sym}, we have the two seed generators
\begin{equation}
\label{eq:seeds}
G_x=\sigma_3\otimes\hat x, \qquad G_p=\sigma_3\otimes\hat p ,
\end{equation}
obtained by taking the displacement amplitude $\alpha$, real and imaginary, respectively.

Throughout we use the identity, valid for any operators $A,B$ on the qudit and $U,V$ on the mode,
\begin{equation}
\label{eq:tensorcomm}
[\,A\otimes U,\;B\otimes V\,]
=\tfrac12\,[A,B]\otimes\{U,V\}
+\tfrac12\,\{A,B\}\otimes[U,V],
\end{equation}
which follows from expanding both sides. For Pauli operators
$\{\sigma_j,\sigma_k\}=2\delta_{jk}$ and
$[\sigma_j,\sigma_k]=2i\epsilon_{jkl}\sigma_l$, so
Eq.~\eqref{eq:tensorcomm} splits cleanly, where Distinct Paulis
anticommute and only the first term survives, identical Paulis commute
and only the second survives.

Let $\mathcal{L}$ denote the Lie algebra generated by $\{\Lambda\otimes\Id\}_{\Lambda\in\mathfrak{su}(d)}\cup\{G_x,G_p\}$. We show $\mathcal{L}$ contains $\Lambda\otimes P(\hat x,\hat p)$ for every $\Lambda\in\mathfrak{su}(d)$ and every polynomial $P$, and also $\Id\otimes P$.

\subsection{Step 1: All Paulis at Degree One}
Applying Eq.~\eqref{eq:tensorcomm} with $A=\sigma_3$, $B=\sigma_1$, $U=\hat x$, $V=\Id$ gives $\{U,V\}=2\hat x$ and $[U,V]=0$, so
\begin{equation}
\label{eq:step1}
[\,\sigma_3\otimes\hat x,\;\sigma_1\otimes\Id\,]
=\tfrac12(2i\epsilon_{31l}\sigma_l)\otimes 2\hat x
=2i\,\sigma_2\otimes\hat x .
\end{equation}
Repeating with $\sigma_2\otimes\Id$ in place of $\sigma_1\otimes\Id$, and with $\hat p$ in place of $\hat x$, yields
\begin{equation}
\label{eq:degree1}
\sigma_k\otimes\hat x\in\mathcal{L},\qquad
\sigma_k\otimes\hat p\in\mathcal{L},\qquad k=1,2,3 .
\end{equation}

\subsection{Step 2: Degree Two, and Why Conditional Displacements Alone Do Not Suffice}
Two conditional displacements with the same qudit factor generate nothing new. Taking $A=B=\sigma_3$ in Eq.~\eqref{eq:tensorcomm} gives $[A,B]=0$ and $\{A,B\}=2\,\Pi_{ab}$, where $\Pi_{ab}=\ket{a}\!\bra{a}+\ket{b}\!\bra{b}$, so
\begin{equation}
\label{eq:cnumber}
[\,\sigma_3\otimes\hat x,\;\sigma_3\otimes\hat p\,]
=\Pi_{ab}\otimes[\hat x,\hat p]
=2i\,\Pi_{ab}\otimes\Id,
\end{equation}
a multiple of the identity on the addressed subspace. This is the
precise sense in which the conditional displacement is not, on its
own, a source of new generators. Rather, it is the complete control of
the discrete register that unlocks the construction.

Commuting two \emph{different} Paulis does produce a new generator. With $A=\sigma_2$, $B=\sigma_3$, $U=\hat x$, $V=\hat p$, the Paulis anticommute, so only the first term of Eq.~\eqref{eq:tensorcomm} survives:
\begin{equation}
\label{eq:degree2a}
[\,\sigma_2\otimes\hat x,\;\sigma_3\otimes\hat p\,]
=\tfrac12(2i\sigma_1)\otimes\{\hat x,\hat p\}
=i\,\sigma_1\otimes\{\hat x,\hat p\} .
\end{equation}
Likewise, with $A=\sigma_1$, $B=\sigma_2$, $U=V=\hat x$,
\begin{equation}
\label{eq:degree2b}
[\,\sigma_1\otimes\hat x,\;\sigma_2\otimes\hat x\,]
=\tfrac12(2i\sigma_3)\otimes 2\hat x^2
=2i\,\sigma_3\otimes\hat x^2 ,
\end{equation}
and similarly for $\hat p^2$. Every quadratic monomial in $\hat x$ and $\hat p$ is therefore in $\mathcal{L}$, tensored with a Pauli, and by Step~1's argument the Pauli factor can then be rotated to any of the three.

\subsection{Step 3: Induction on Polynomial Degree}
\begin{lemma}
\label{lem:induction}
For every $n\ge 1$, every polynomial $Q$ of degree $\le n$ in $\hat x$ and $\hat p$, and every $k\in\{1,2,3\}$, the operator $\sigma_k\otimes Q$ lies in $\mathcal{L}$.
\end{lemma}

\begin{proof}
Steps~1 and~2 establish the cases $n=1$ and $n=2$. Suppose the claim holds up to degree $n$, and let $Q$ be a monomial of degree $n$. Choose $j\ne k$, so that $\sigma_j$ and $\sigma_k$ anticommute and Eq.~\eqref{eq:tensorcomm} reduces to its first term:
\begin{equation}
\label{eq:inductstep}
[\,\sigma_j\otimes\hat x,\;\sigma_k\otimes Q\,]
=\tfrac12\,(2i\epsilon_{jkl}\sigma_l)\otimes\{\hat x,Q\}
=i\,\epsilon_{jkl}\,\sigma_l\otimes\{\hat x,Q\} .
\end{equation}
Expanding the anticommutator,
\begin{equation}
\label{eq:anticomm_expand}
\{\hat x,Q\} = 2\,\hat x\,Q + [\,Q,\hat x\,] ,
\end{equation}
where the first term is a monomial of degree $n+1$ and the second, being a commutator of a degree-$n$ polynomial with $\hat x$, has degree at most $n-1$ by the canonical relation $[\hat x,\hat p]=2i$. The correction term is therefore of degree $\le n-1$, and by the inductive hypothesis $\sigma_l\otimes[Q,\hat x]$ already lies in $\mathcal{L}$. Subtracting it from Eq.~\eqref{eq:inductstep} leaves
\begin{equation}
2i\,\epsilon_{jkl}\,\sigma_l\otimes \hat x\,Q \;\in\;\mathcal{L}.
\end{equation}
The same argument with $\hat p$ in place of $\hat x$ gives $\sigma_l\otimes\hat p\,Q\in\mathcal{L}$. Since every monomial of degree $n+1$ is of the form $\hat x Q$ or $\hat p Q$ with $Q$ of degree $n$, and since the Pauli index $l$ can be rotated to any value by the argument of Step~1, the claim holds at degree $n+1$. By induction it holds for all $n$.
\end{proof}

\subsection{Step 4: From $\mathfrak{su}(2)$ to $\mathfrak{su}(d)$}
Lemma~\ref{lem:induction} gives $\sigma_k\otimes P$ for arbitrary $P$, but the qudit factor is so far confined to the $\{\ket{a},\ket{b}\}$ subspace. Pillar~1 supplies $\Lambda\otimes\Id$ for every $\Lambda\in\mathfrak{su}(d)$, and Eq.~\eqref{eq:tensorcomm} with $V=\Id$ gives
\begin{equation}
\label{eq:sud_lift}
[\,\Lambda\otimes\Id,\;\sigma_k\otimes P\,]
=[\Lambda,\sigma_k]\otimes P .
\end{equation}
The map $\Lambda\mapsto[\Lambda,\sigma_k]$ is the adjoint action of $\mathfrak{su}(d)$ on itself. As a real Lie algebra $\mathfrak{su}(d)$ is simple, its adjoint representation is irreducible and the only nonzero invariant subspace containing $\sigma_k$ is all of $\mathfrak{su}(d)$~\cite{hall2015}. Repeated application of Eq.~\eqref{eq:sud_lift} therefore generates $\Lambda\otimes P$ for every $\Lambda\in\mathfrak{su}(d)$, with $P$ fixed. Since $P$ was arbitrary, $\mathcal{L}$ contains $\Lambda\otimes P$ for all $\Lambda\in\mathfrak{su}(d)$ and all polynomials $P$.

\subsection{Step 5: Several Modes from $C_{\Pi_z}D$ Alone}
Steps~1--4 were carried out on a single mode $j$ and apply verbatim to
each mode of the register, since $C_{\Pi_z}D$ acts on any chosen
mode. In particular, Step~1 places $\sigma_r\otimes\hat x_j$ and
$\sigma_r\otimes\hat p_j$ in $\mathcal{L}$ for every mode $j$ and
every Pauli index $r$. Cross-mode products then follow directly from
Eq.~\eqref{eq:tensorcomm}. In particular, for two different modes
$j\neq j'$ and two different Pauli indices $r\neq s$, the quadratures
commute, so only the Pauli commutator survives,
\begin{equation}
\label{eq:crossmode_app}
[\,\sigma_r\otimes\hat x_j,\;\sigma_s\otimes\hat x_{j'}\,]
=\tfrac12(2i\epsilon_{rst}\sigma_t)\otimes 2\hat x_j\hat x_{j'}
=2i\,\epsilon_{rst}\,\sigma_t\otimes\hat x_j\hat x_{j'} ,
\end{equation}
and likewise for any pair of quadratures. The induction of
Lemma~\ref{lem:induction} extends to total degree across all modes. As
such, commuting a degree-$n$ multimode monomial (with a fixed Pauli
factor) against a single-mode generator $\sigma_r\otimes\hat x_{j''}$
with a different Pauli index raises the degree on mode $j''$ by one,
with the lower-degree correction from $[\hat x_{j''},\hat p_{j''}]=2i$
already in $\mathcal{L}$ by hypothesis, exactly as in the single-mode
proof. Since every multimode monomial is reachable this way,
$\mathcal{L}$ contains $\sigma_r\otimes P$ for every polynomial $P$ in
the quadratures of all modes, and by Step~4 also $\Lambda\otimes P$
for every $\Lambda\in\mathfrak{su}(d)$. No mode-coupling primitive is
used.

\subsection{Step 6: The Identity Sector}
What remains is to obtain $\Id\otimes P$. Take $A=B=\sigma_3$ in
Eq.~\eqref{eq:tensorcomm}, so that $[A,B]=0$ and $\{A,B\}=2\Pi_{ab}$:
\begin{equation}
\label{eq:identity_sector}
[\,\sigma_3\otimes U,\;\sigma_3\otimes V\,]
=\Pi_{ab}\otimes[U,V] .
\end{equation}
By Step~5, $U$ and $V$ may be arbitrary multimode polynomials. Under commutation the quadratures generate the full polynomial
algebra~\cite{lloyd1999}. Explicitly, for any $m,n\ge0$,
\begin{equation}
\label{eq:ladder}
[\,\hat x_j^{\,m+1}\hat p_j^{\,n},\;\hat p_j\,]
= 2i\,(m{+}1)\,\hat x_j^{\,m}\hat p_j^{\,n},
\end{equation}
which follows from $[\hat x_j,\hat p_j]=2i$ and the Leibniz rule, and
likewise with $\hat x_j$ and $\hat p_j$ exchanged. Since Step~5 already
supplies $\sigma_3\otimes U$ for arbitrary multimode polynomials $U$,
Eq.~\eqref{eq:ladder} shows that
$\Pi_{ab}\otimes\hat x_j^{\,m}\hat p_j^{\,n}$ is generated for every
$m,n\ge0$. Factors involving other modes commute through unchanged, so the
same construction gives every normally ordered multimode monomial. Hence,
$\Pi_{ab}\otimes P\in\mathcal{L}$ for every multimode polynomial $P$ and
every level pair $\{a,b\}$. Summing over all unordered pairs gives
\begin{equation}
\sum_{a<b}\Pi_{ab}=(d-1)\Id,
\end{equation}
so $\Id\otimes P\in\mathcal{L}$ for arbitrary $d\ge2$.
Together with the $\mathfrak{su}(d)$ sector of Step~4, this spans the full
space of qudit operators tensored with $P$.

The synthesized beam splitter of Sec.~\ref{sec:bs} supplies the Gaussian
generator $B_{jj'}=\tfrac12(\hat x_j\hat p_{j'}-\hat p_j\hat x_{j'})$
directly; it is contained in the algebra above and is provided as a primitive
for efficiency rather than necessity.

\subsection{Conclusion}
$\mathcal{L}$ contains $\Lambda\otimes P(\hat x_1,\hat p_1,\dots,\hat x_M,\hat p_M)$ for every $\Lambda$ in the space of qudit operators and every polynomial $P$. These span the polynomial Hamiltonian algebra of the joint qudit--oscillator system, establishing universality in the sense of Hamiltonian generation~\cite{lloyd1999, brif2010}.\qed

The single-mode construction (Steps~1--4) was checked numerically for the ququart. The calculation was seeded with the six $X_E$ generators $X_{ij}\otimes\Id$ for $i<j$, the four traceful $Z_d$ projectors $\Pi_\ell\otimes\Id$, and the two conditional-displacement generators $\sigma_3\otimes\hat x$ and $\sigma_3\otimes\hat p$, giving twelve initial matrices in total. On a mode truncated at Fock cutoff $N$, the generated algebra follows the rank sequence $12, 29, 105, 375$ over iterations $0$--$3$, identically for $N=6,8,10$. The measured Fock-basis bandwidth of the generated operators through iteration $3$ is three. As explained in Sec.~\ref{sec:entangle}, bandwidth is a lower bound on polynomial degree, so this certifies non-Gaussian operators of at least degree three but gives no upper bound, and the iteration-$3$ span of dimension $375$ exceeds the $160$-dimensional space of degree at most three, so higher-degree operators are also present. At iteration $4$, the finite matrices saturate on $\mathfrak{u}(4N)$, of dimension $(4N)^2$; and because the $Z_d$ seeds are traceful level projectors, the closure is $\mathfrak{u}(4N)$ rather than $\mathfrak{su}(4N)$. This saturation is necessarily cutoff-dependent, serving as a consistency check on the truncated representation rather than as evidence for the infinite-dimensional statement, which Steps~1--6 establish analytically. The control with unconditional displacements follows the sequence $12, 19, 21, 24$ over the same iterations at $N=6$ and never exceeds bandwidth one at $N=6,8,10$, so the numerics separate the Gaussian and non-Gaussian seed sets. This is a consistency check and does not by itself rule out truncation effects.

\section{Qubit Reference Compilation}
\label{app:qubitport}

This appendix records how the trapped-ion qubit reference of
Sec.~\ref{sec:methodB} is built and how its $32$ logical operations
expand to $58$ timed pulses, with each CNOT counted as one pulse. Its purpose is to make the Timed pulses column of
Table~\ref{tab:gatecounts} reproducible.

\subsection{From the published circuit to $32$ logical operations}

The reference compiles the displaced-frame Hamiltonian $H_{\rm cqed}$ of
Ref.~\cite{vu2025} on three qubits and four modes. In the form given
there, $\tilde H/\hbar=\tilde H_0/\hbar+\tilde H_1/\hbar+\tilde H_{2,XX}/\hbar+\tilde H_{2,YY}/\hbar$ with
\begin{align}
\tilde H_0/\hbar &= \omega_a a^\dagger a+\omega_b b^\dagger b+\omega_c c^\dagger c+\omega_l l^\dagger l \nonumber\\
  &\quad -\Delta_{ab}\frac{\sigma_b^z}{2}-\Delta_{ac}\frac{\sigma_c^z}{2}, \nonumber\\
\tilde H_1/\hbar &= -\frac{\chi_a}{2}a^\dagger a\,\sigma_a^z-\frac{\chi_b}{2}b^\dagger b\,\sigma_b^z-\frac{\chi_c}{2}c^\dagger c\,\sigma_c^z \nonumber\\
  &\quad + g_{cd,a}(a+a^\dagger)\frac{\sigma_a^z}{2}+g_{cd,b}(b+b^\dagger)\frac{\sigma_b^z}{2} \nonumber\\
  &\quad + g_{cd,c}(c+c^\dagger)\frac{\sigma_c^z}{2}, \nonumber\\
\tilde H_{2,XX}/\hbar &= g_{cd,l}(l+l^\dagger)\frac{\sigma_a^z}{4}+\frac{g_{ab}}{2}\sigma_a^x\sigma_b^x+\frac{g_{ac}}{2}\sigma_a^x\sigma_c^x \nonumber\\
  &\quad +\frac{g_{abl}}{2}\sigma_a^x\sigma_b^x(l+l^\dagger) \nonumber\\
  &\quad +\frac{g_{acl}}{2}\sigma_a^x\sigma_c^x(l+l^\dagger), \nonumber\\
\tilde H_{2,YY}/\hbar &= g_{cd,l}(l+l^\dagger)\frac{\sigma_a^z}{4}+\frac{g_{ab}}{2}\sigma_a^y\sigma_b^y+\frac{g_{ac}}{2}\sigma_a^y\sigma_c^y \nonumber\\
  &\quad +\frac{g_{abl}}{2}\sigma_a^y\sigma_b^y(l+l^\dagger) \nonumber\\
  &\quad +\frac{g_{acl}}{2}\sigma_a^y\sigma_c^y(l+l^\dagger).
\end{align}
On trapped ions $\tilde H_0$ compiles to mode rotations and virtual
$r_z$, $\tilde H_1$ to CR and CD, the $g_{ab},g_{ac}$ terms to the
exchange $J$, the two $g_{cd,l}$ halves to one link-mode CD, and the
$g_{abl},g_{acl}$ terms to the link-modulated exchange $J_l$, which our
step applies once at the center. Its released circuit
runs a full second-order step, every block once forward and once
backward, and uses an ancilla qubit with SWAP and CNOT routing to reach
the link mode. On trapped ions the shared motional bus makes the ancilla
and the routing unnecessary, so we drop them and keep the $22$ Hamiltonian
terms. Grouping the terms into the trapped-ion operations of
Table~\ref{tab:isa}, one half step uses $15$ operations, and the two
link-modulated exchange terms sit once at the center, giving
$2\times15+2=32$. The published estimate of $84$ CNOT, $9$ CD, and $3$
SNAP counts each term once for a nearest-neighbor superconducting line
and is not this circuit. Counting the released circuit directly gives
$132$ primitive gates per full step, of which $72$ act on more than one
register. That recount reproduces the published CD and SNAP numbers up
to the factor of two between one application and the full symmetric
step, $18$ against $9$ and $6$ against $3$. The CNOT number comes to $88$ against the published $84$ once SWAP
routing is expanded, a small difference in routing convention that we
do not consider in our gate accounting.

\begin{figure*}[t]
  \centering
  \includegraphics[width=0.92\textwidth]{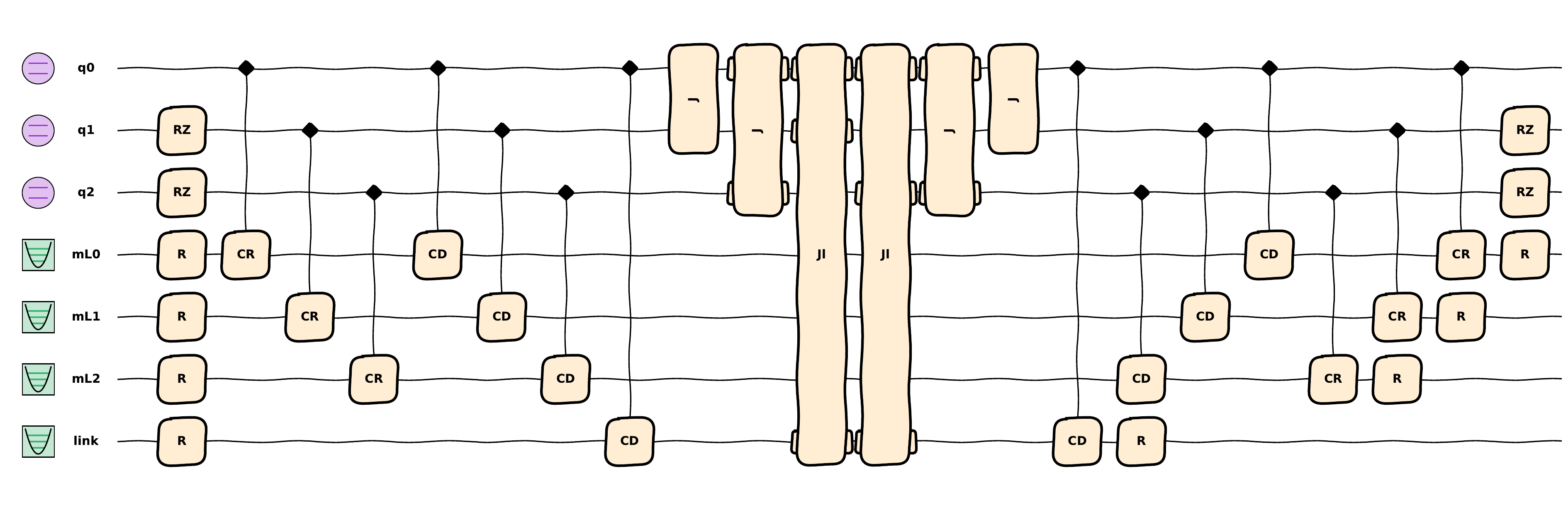}
  \caption{Single second-order Trotter step of the qubit reference on
  three qubits and four qumodes at the ISA level, $32$ operations.
  R denotes an unconditional mode rotation, RZ the virtual detuning
  phase, CR and CD the dispersive and conditional displacement gates,
  J the exchange and Jl the link-modulated exchange at the center.
  Table~\ref{tab:qubit32} gives the native realization of each class.}
  \label{fig:circuit_qubit}
\end{figure*}

\begin{figure*}[t]
  \centering
  \includegraphics[width=\textwidth]{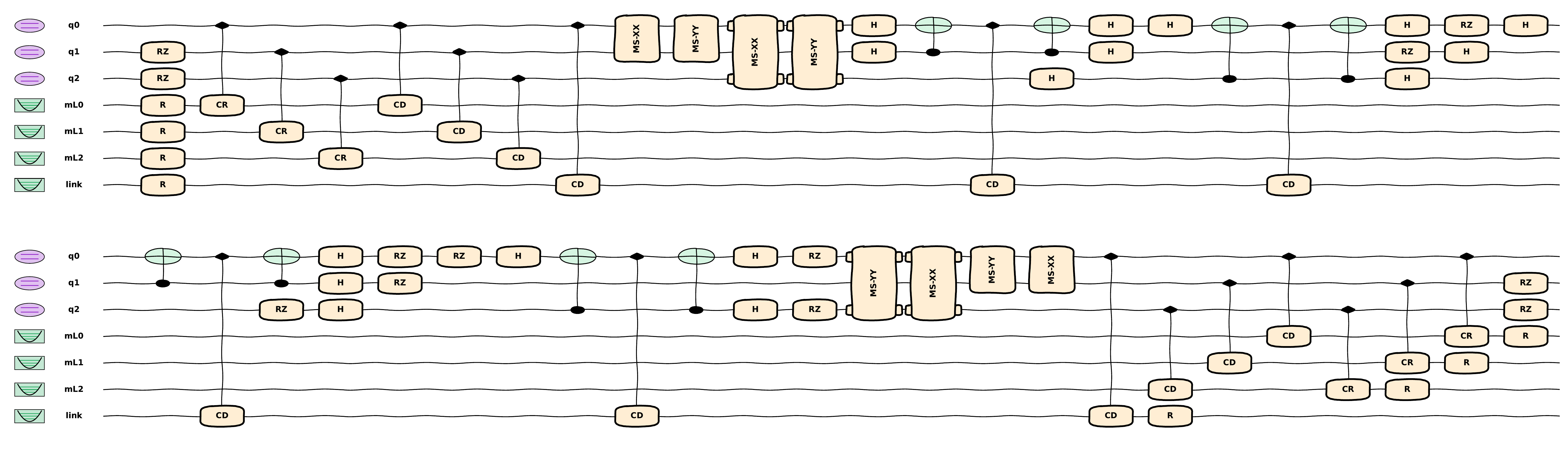}
  \caption{The same qubit step compiled to $58$ timed pulses and $12$
  virtual phases, split into two rows at the center of the symmetric
  step. Each exchange, J, becomes an MS-XX and an MS-YY window, and each
  link-modulated exchange, Jl, becomes a sandwich of Hadamard pulses,
  CNOTs, and conditional displacements on the link mode. The CD, CR, and
  R pulses map one to one onto the operations of
  Fig.~\ref{fig:circuit_qubit}.}
  \label{fig:circuit_qubit_58}
\end{figure*}

\begin{table}[h]
  \centering
  \begin{tabular}{lcc}
    \toprule
    Operation & Per step & Native realization \\
    \midrule
    Mode rotation $R$ & 8 & one pulse \\
    Dispersive $\chi\,\sigma_z\hat n$ & 6 & one pulse \\
    Conditional displacement & 8 & one pulse \\
    Detuning $r_z$ & 4 & virtual, $T=0$ \\
    Exchange $J$ & 4 & two MS windows \\
    Link exchange $J_l$ & 2 & synthesized, $14$ pulses \\
    \midrule
    Total logical & 32 & \\
    \bottomrule
  \end{tabular}
  \caption{Logical histogram of the qubit reference for one second-order
  Trotter step at $N_c=3$ and the native realization of each class.}
  \label{tab:qubit32}
\end{table}

\subsection{From $32$ logical to $58$ timed pulses}

Table~\ref{tab:qubit32} shows that $22$ of the $32$ logical operations are single-pulse operations, four are
virtual, and the remaining six exchange operations require further decomposition. The whole expansion
comes from the two exchange classes.

The exchange $J(\sigma_x\sigma_x+\sigma_y\sigma_y)/2$ is not one
M{\o}lmer--S{\o}rensen gate. A bichromatic M{\o}lmer--S{\o}rensen
interaction produces $XX$ or $YY$ for one phase setting, not their sum.
Since $[XX,YY]=0$ the sum splits exactly into two windows, one $XX$ window
and one $YY$ window, with no Trotter error. Four exchange operations
therefore become eight M{\o}lmer--S{\o}rensen windows.

The link-modulated exchange $J_l$ has no native two-ion primitive. It is
the same $XX+YY$ structure multiplied by the link-mode position, and we
build it from two sandwiches, one for $XX$ and one for $YY$, each wrapping
a conditional displacement on the link mode between basis changes on the
two qubits. Per $J_l$ this costs $8$ Hadamard pulses, $4$ CNOT, and $2$
conditional displacements, plus $4$ virtual $Z$ updates. The two $J_l$
operations therefore add $28$ timed pulses and $8$ virtual updates.

The sum is $8+6+8+8+28=58$ timed pulses and $12$ virtual updates, which
is the count in Table~\ref{tab:gatecounts} and is drawn in
Fig.~\ref{fig:circuit_qubit_58}. The qudit circuit under
the same rule gives $27$, with the two
central $C_{\Pi_x}D$ terms in one pulse. The
CNOT is priced as one primitive. On trapped ions it is itself one
M{\o}lmer--S{\o}rensen gate plus local rotations, so $58$ is a lower
bound. The runtime model of Appendix~\ref{app:runtime} prices each $J_l$
at one M{\o}lmer--S{\o}rensen window, the cost of a single native
entangling gate. The saved waveforms for CD and CR, and for every
qudit primitive, are in the repository. The M{\o}lmer--S{\o}rensen
window has no simulated waveform, its duration is the measured value of
Ref.~\cite{meth2025}.

\subsection{Executed circuit}

Both encodings were executed with finite shots on the extended HybridLane
backend at Fock cutoff $8$. The qudit circuit was run with $10^4$ shots at $41$ time points and
agrees with exact integration to $0.40$ to $0.42\%$ RMSE, with the
largest single deviation at $2.7$ standard deviations, and a $10^3$ shot
run gives three times the RMSE, as expected for shot noise.
The qubit reference was run with $10^4$ shots at $41$ time points and
agrees with exact integration to $0.35$ to $0.43\%$ RMSE, with the
largest single deviation at $3.0$ standard deviations. These runs validate
the circuit logic of both encodings. They do not validate pulse cost,
which is the subject of Appendices~\ref{app:pulses}
and~\ref{app:runtime}.

\section{Pulse-Level Parameters}
\label{app:pulses}

This appendix lists how each ISA primitive of Sec.~\ref{sec:isa} is
realized at the pulse level, the hardware parameters used, and the
fidelity reached. Most simulations use a Fock cutoff of $8$ and $200$ time
slots per pulse. The $C_{\Pi}R$ physical replay uses $400$ slots, the
two-transition $C_{\Pi_x}D$ run $20{,}000$, and the $X_E$ run is a
five-level model with no mode. Two kinds of fidelity appear. A self-consistency value
replays the final waveform through the control model the optimizer was
given, in a separate integration with $40$ substeps per slot, and only
asks whether the pulse reaches its own target. A physical replay value
propagates the same waveform through the real drive Hamiltonian with no
imposed generator, and is the number that matters. Two of the five gates, $C_{\Pi_z}D$ and its qubit twin CD,
were optimized against imposed generators, so their self-consistency
values are small by construction. The other three are flat pulses.
Physical replay values exist for $X_E$, $C_{\Pi_z}D$,
$C_{\Pi}R$, and $C_{\Pi_x}D$. For CD and CR no physical replay exists yet, and the
paper does not claim one. Optimizer self-reported fidelities are never
quoted, since they were found to run optimistic at long durations.

\subsection{Hardware parameters}

Table~\ref{tab:params} collects the parameters. Each is tagged as
measured, derived, or assumed. Measured values are taken from the
$^{40}$Ca$^+$ qudit processor of Ref.~\cite{meth2025}. Derived values
follow from those by a stated relation. Assumed values are operating
points chosen in this work and are the ones a referee should read as
choices rather than facts.

\begin{table}[h]
  \centering
  \begin{tabular}{llll}
    \toprule
    Parameter & Value & Tag & Source \\
    \midrule
    Lamb--Dicke factor $\eta$ & $0.041$ & measured & \cite{meth2025} \\
    Axial mode $\nu/2\pi$ & $0.77$~MHz & measured & \cite{meth2025} \\
    Sideband Rabi $\Omega_{\rm sb}/2\pi$ & $4$~kHz & measured & \cite{meth2025} \\
    Carrier Rabi $\Omega_c/2\pi$ & $97.56$~kHz & derived & $\Omega_{\rm sb}/\eta$ \\
    MS gate time $T_{\rm MS}$ & $200~\mu$s & measured & \cite{meth2025} \\
    Motional coherence $\tau_{\rm mot}$ & $27.4$~ms & measured & \cite{meth2025} \\
    Heating rate & $2.7$~phonons/s & measured & \cite{meth2025} \\
    Laser coherence $T_2$ & $92$~ms & measured & \cite{meth2025} \\
    $C_{\Pi}R$ detuning $\delta/2\pi$ & $40$~kHz & assumed & this work \\
    \bottomrule
  \end{tabular}
  \caption{Hardware parameters used in the pulse simulations and in the
  runtime model of Appendix~\ref{app:runtime}. The carrier Rabi frequency
  is not assumed at a round value but derived from the measured sideband
  Rabi frequency and Lamb--Dicke factor.}
  \label{tab:params}
\end{table}

\subsection{Method per gate}

Table~\ref{tab:pulsemethods} gives the pulse method, duration, and
replay fidelity for each primitive at the operating point of the
vibronic circuit. The five primitives are not all realized the same way.
$C_{\Pi_z}D$ is shaped by HyPulse~\cite{hakimi2026hypulse} at the test
angle, where the optimizer returns a flat pulse, one quadrature at the
bound. At the circuit angles all five are flat pulses by construction,
checked directly against the physical drive, and the flat pulse there
scores at least as well as an optimized one.

\begin{table*}[t]
  \centering
  \small
  \begin{tabular}{llllll}
    \toprule
    Gate & Physical drive & Pulse method & Duration & Self-consist. $1-F$ & Physical $1-F$ \\
    \midrule
    $Z_d$ & virtual phase & none & $0$ & $0$ & exact \\
    $X_E$ & carrier, two tones & flat, aggregate cap & $0.470~\mu$s & none & $7.8\times10^{-7}$ \\
    $C_{\Pi_z}D$ & $\sigma_z$ force & HyPulse & $160.4~\mu$s & $1.1\times10^{-5}$ & $1$ to $2\times10^{-3}$ \\
    CD (qubit ref.) & $\sigma_z$ force & HyPulse & $39.79~\mu$s & $4.0\times10^{-6}$ & none \\
    CR (qubit ref.) & detuned sideband & GRAPE & $24.19~\mu$s (priced as $C_{\Pi}R$) & $1.1\times10^{-5}$ & none \\
    $C_{\Pi}R$ & detuned sideband & flat & $25.000~\mu$s & $5\times10^{-14}$ & $2.4\times10^{-3}$ \\
    $C_{\Pi_x}D$ & bichromatic sideband & flat & $0.649~\mu$s & real drive & $3.5\times10^{-9}$ \\
    \bottomrule
  \end{tabular}
  \caption{Pulse realization of each primitive at the circuit operating
    point. Durations are those of the pulse runs at their test angles, the
  circuit durations used in Appendix~\ref{app:runtime} follow from the
  scaling rules given there. The
  $C_{\Pi}R$ duration is the physical dispersive time and is what the runtime
  model uses.
  Self-consistency values replay the waveform through the optimizer's own
  control model. Physical values use the real drive Hamiltonian with no
  imposed generator. $X_E$ is a flat two-tone block at the circuit
  angles. The $C_{\Pi_z}D$
  physical value uses the full Bessel drive with $5~\mu$s $\sin^2$ ramps.
    $C_{\Pi_x}D$ is optimized under the real drive, so it has no separate
  self-consistency value. The dispersive gate of the qubit reference, CR,
  is priced at the $C_{\Pi}R$ duration in Appendix~\ref{app:runtime}, its own
  optimized duration is shorter. All values use the same replay integrator with
  virtual $Z$ phases absorbed. The saved waveforms are in the repository.}
  \label{tab:pulsemethods}
\end{table*}

\paragraph{$Z_d$.} Virtual, implemented as a phase bookkeeping update on the
drives that address the level. It costs no time.

\paragraph{$X_E$.} A resonant carrier drive on the addressed optical
transition. A single transition is a flat carrier pulse with $\theta=\Omega t$,
the standard rule~\cite{schindler2013}. The multitone version that fuses
$A\!\leftrightarrow\!B$ and $A\!\leftrightarrow\!C$ is also a flat block,
two carrier tones driven at once, one resonant on each transition, with
amplitudes in the ratio $J_{AB}$ to $J_{AC}$ and the total field capped by
$\sqrt{\sum_k\Omega_k^2}\le\Omega_c$, since the two tones share one beam.
Because the drive is flat, its Hamiltonian is time independent and equal
to the summed generator, so the multitransition block is reached exactly
with no shaping. The residual error is off-resonant cross-talk, each tone
also driving the other transitions out of $A$. That cross-talk is set by
the Zeeman splittings of the $D_{5/2}$ sublevels, which
Ref.~\cite{meth2025} does not quote, so we assume a $2$ to $5$~MHz window
and take the amplitude that works across all of it. At $7\%$ of the cap,
$T=0.470~\mu$s, the block reaches $1-F=7.8\times10^{-7}$ at $2$~MHz and
$1.6\times10^{-6}$ at $5$~MHz, with auxiliary leakage below
$1.2\times10^{-6}$. This is the duration the runtime model of
Appendix~\ref{app:runtime} prices.

\paragraph{$C_{\Pi_z}D$ and CD.} Both use a $\sigma_z$ type state-dependent force
from a bichromatic drive detuned by $\pm\nu/2$ from the carrier, following
Ref.~\cite{bazavan2023}. The controls are the two quadrature amplitudes,
shaped by HyPulse~\cite{hakimi2026hypulse}. The $C_{\Pi_z}D$ waveform is then replayed against the full Bessel
expansion of the real drive, which is the physical value in
Table~\ref{tab:pulsemethods}. Unramped at $160.4~\mu$s it reaches $0.998$
to $0.999$, and with $5~\mu$s $\sin^2$ ramps at $165.4~\mu$s it reaches
$0.999$. Scaled to the circuit displacement $\alpha=0.008$ at $2.57~\mu$s
and replayed under the same real drive, the $C_{\Pi_z}D$ pulse reaches
$1-F=5\times10^{-6}$ with leakage $2\times10^{-4}$ after the virtual
$Z$ and mode frame correction, so the linear scaling used in
Appendix~\ref{app:runtime} is supported by a physical replay, and the
gate is cleaner at the circuit angle than at the test angle. The CD
waveform of the qubit reference has not been replayed
against the real drive, so only its self-consistency value is given. A residual $\sigma_y$ term
and the Stark shift on spectator levels are both absorbed into per-level
virtual $Z$ phases, the latter to $10^{-5}$.

\paragraph{$C_{\Pi}R$.} A single blue sideband tone detuned by $\delta$ from the
sideband. In the dispersive limit the drive produces the generator
$\Pi_\ell\otimes\hat n$ with rate $\chi=\Omega_{\rm sb}^2/4\delta$. We do
not impose this generator for the physical value. The pulse is a flat drive
on the real detuned sideband, and the generator is fitted from the
resulting propagator. The self-consistency value in
Table~\ref{tab:pulsemethods} comes from an imposed-generator run and only
confirms the compiler, it is not a physical fidelity. Over a
long window the fitted rate agrees with the Schrieffer--Wolff prediction to
$1.7\%$, $\chi_{\rm fit}/\chi_{\rm pred}=0.983$. The gate duration is one detuning period, $T=2\pi/\delta=25.000~\mu$s.
A population driven out of the addressed level mostly returns on that
period, so the leakage has a minimum there, and the compiled angle is set by the
sideband amplitude rather than by the duration. The three compiled
angles need $\Omega_{\rm sb}/2\pi=3.77$, $4.12$, and $4.38$~kHz, within
$10\%$ of the measured $4$~kHz, and each delivers its angle to $1.000$
with leakage $5.4\times10^{-4}$, $9.0\times10^{-4}$, and
$1.3\times10^{-3}$, and replay infidelity $1.7$, $2.4$, and
$3.0\times10^{-3}$. All spectator
level pairs stay exactly dark throughout. Mid-gate the off-diagonal
excursion into the auxiliary level reaches $0.26$ and closes again, so
endpoint and maximum over time are reported separately. The carrier light shift on the driven pair is
level-diagonal and independent of the mode occupation, so it is absorbed
as a virtual $Z_d$, and $C_{\Pi}R$ is priced as a
single pulse. The sign convention is $+\theta$ with the phase
of the drive as defined in the repository.

\paragraph{$C_{\Pi_x}D$.} Simultaneous red and blue sideband tones on the
addressed optical transition with symmetric bichromatic phase, giving the
generator $(\ket{\ell_i}\!\bra{\ell_j}+\mathrm{h.c.})\otimes\hat x$. The
simulation uses the real drive with time-dependent control operators,
$G(t)=\cos(\nu t)\,\ket{\ell_j}\!\bra{\ell_i}\otimes D(i\eta e^{i\nu t})$,
no drift term, and no rotating wave approximation. At the circuit angle
$\theta=0.0015$ a flat pulse of duration half a motional period,
$0.649~\mu$s, at $9.2\%$ of $\Omega_c$ reaches $1-F_{\rm indep}=3.5\times10^{-9}$
with leakage $6.5\times10^{-9}$. GRAPE started from this pulse changes it
by less than $10^{-11}$ in fidelity, so no shaping is needed at this angle.
The naive duration $\theta/\Omega_{\rm sb}=0.06~\mu$s does not work,
because below about half a motional period the counter-rotating terms have
not averaged and the $\hat x$ quadrature selection does not yet exist. At
large angles the co-driven carrier becomes significant. At $\theta=0.5$ a
flat pulse fails, and GRAPE under the hardware amplitude cap and a
bandwidth limit reaches $2.4\times10^{-3}$ with $2\times10^{-2}$ leakage.
The vibronic circuit uses only small angles, so this regime is not needed
here. Driving the $A\!\leftrightarrow\!B$ and $A\!\leftrightarrow\!C$
sideband pairs at the same time, as the summed central block of
Sec.~\ref{sec:methodA} requires, reaches $1-F=8\times10^{-9}$ against the
exact summed exponential with the spectator level exactly dark, at an
aggregate amplitude of $15\%$ of $\Omega_c$.

\subsection{Choice of the $C_{\Pi}R$ detuning}
\label{app:delta}

The $C_{\Pi}R$ detuning $\delta$ is an operating point chosen in this work. It is
bounded from below by the dispersive condition
$\Omega_{\rm sb}\ll\delta$ and from above by the requirement to stay on one
sideband, $\delta\ll\nu$ and below the motional mode spacing. Our repository
does not anchor the mode spectrum for the full crystal, so we use $\nu$ as
the stand-in upper bound. The value $\delta/2\pi=40$~kHz gives
$\Omega_{\rm sb}/\delta=0.1$ and $\delta/\nu=0.05$.
Table~\ref{tab:deltasweep} shows the gate across the window under the
amplitude calibration of Appendix~\ref{app:runtime}. With the duration
fixed at one detuning period, the leakage is the same at every detuning
and the choice of $\delta$ trades gate time against the sideband
amplitude the drive must supply. The operating point $40$~kHz sits
within $3\%$ of the measured amplitude.

\begin{table}[h]
  \centering
  \small
  \setlength{\tabcolsep}{4pt}
  \begin{tabular}{rrrrlc}
    \toprule
    $\delta/2\pi$ & $T$ & $\Omega_{\rm sb}/2\pi$ & $\Omega_{\rm sb}/\Omega_{\rm sb}^{\rm meas}$ & Leakage & Ratio \\
    (kHz) & ($\mu$s) & (kHz) & & & (seq.) \\
    \midrule
    20 & 50.00 & 2.06 & 0.51$\times$ & $1.3\times10^{-3}$ & 6.5 \\
    30 & 33.33 & 3.09 & 0.77$\times$ & $1.3\times10^{-3}$ & 8.7 \\
    40 & 25.00 & 4.12 & 1.03$\times$ & $1.3\times10^{-3}$ & 10.5 \\
    60 & 16.67 & 6.17 & 1.54$\times$ & $1.3\times10^{-3}$ & 13.4 \\
    80 & 12.50 & 8.23 & 2.06$\times$ & $1.3\times10^{-3}$ & 15.7 \\
    \bottomrule
  \end{tabular}
  \caption{$C_{\Pi}R$ across the detuning window under amplitude
  calibration. Every dispersive gate runs for one detuning period,
  $T=2\pi/\delta$, and the compiled angle is set by the sideband
  amplitude. In the dispersive approximation, $\theta=\chi T$ with
  $\chi=\Omega_{\rm sb}^2/4\delta$ gives
  $\Omega_{\rm sb}=\delta\sqrt{2\theta/\pi}$, which fixes the scaling.
  The amplitudes tabulated here are solved numerically under the full
  drive and run about $2\%$ above it.
  In the simulated drive model the endpoint leakage and the ratio
  $\Omega_{\rm sb}/\delta\approx0.10$ come out the same at every
  detuning, as expected when neither $\delta T$ nor
  $\Omega_{\rm sb}/\delta$ depends on $\delta$. Fixed hardware
  frequencies not in the model can break this. The
  detuning trades duration against laser power. The last column is the
  projected sequential runtime ratio at $N_c=3$ from
  Appendix~\ref{app:runtime} with all other gates unchanged. It rises
  with $\delta$. The ends of the window ask for $2.06$ and $8.23$~kHz
  against a measured $4.0$~kHz.}
  \label{tab:deltasweep}
\end{table}

\section{Projected Runtime}
\label{app:runtime}

This appendix gives the timing model behind the projected step times
quoted in the main text. Every number is read from the same table object
that produces the counts, so the two cannot drift apart.

\subsection{Rules}

Each gate instance is priced from its own compiled angle, and the step
time is the sum. The timing model uses a sequential schedule. The
electronic couplings form a star centered on ion $A$, so the two-ion
exchange gates share an ion and cannot be parallelized under a
disjoint-pair scheduling assumption. Unconditional mode rotations can in
principle be absorbed into oscillator frame tracking at no cost. We
nevertheless charge them, and removing that charge would raise the
sequential runtime ratio to $12.7$. Every exchange $J$ is two
M{\o}lmer--S{\o}rensen windows, since $XX$ and $YY$ do not fuse, and
each $J_l$ is priced at one window, the cost of a single native
entangling gate, with its synthesized form given in
Appendix~\ref{app:qubitport}.
The rules follow.
\begin{description}\setlength{\itemsep}{0pt}
  \item[\textbf{R1}] Unconditional mode rotation, $5.0~\mu$s, a scheduling allowance chosen in this work and charged on both sides. At $N_c=3$ it is $40~\mu$s of the qudit total and $40~\mu$s of the qubit total.
  \item[\textbf{R2}] $T=2\pi/\delta=25.000~\mu$s, one detuning period,
    a single flat pulse, for every dispersive gate on both sides. A population driven out of the addressed level mostly returns on that period, so the leakage has a minimum there. The compiled angle is delivered by calibrating the sideband amplitude at this fixed duration, $\Omega_{\rm sb}/2\pi=3.77$, $4.12$, and $4.38$~kHz for the three angles, each verified under the physical drive to deliver its angle to $1.000$ (Appendix~\ref{app:pulses}). The angle column is as compiled.
  \item[\textbf{R3}] $T=2|\alpha|/u$ with $u=\eta\Omega_c(J_1{+}J_3)(4\Omega_c/\nu)=6.23\times10^{-3}$~rad/$\mu$s. The CD of the qubit reference is the same physical gate and uses the same bound.
  \item[\textbf{R4}] Flat two-tone carrier block, $T=2|\theta|/(0.07\,\Omega_c)$ (Appendix~\ref{app:pulses}).
  \item[\textbf{R5}] Not angle based, half a motional period, $\pi/\nu=0.649~\mu$s, below which the quadrature selection does not exist.
  \item[\textbf{R6}] Virtual frame update, $T=0$.
  \item[\textbf{R7}] M{\o}lmer--S{\o}rensen window, $200~\mu$s, taken as angle independent following Ref.~\cite{meth2025}.
  \item[\textbf{R8}] Inside a $J_l$ sandwich. $J_l$ is priced at one M{\o}lmer--S{\o}rensen window, so these carry no separate charge.
\end{description}

\subsection{Results}

Both tables price the reference circuits, the $34$-operation qudit ISA circuit and the $32$-operation qubit ISA circuit, compiled to $36$ and $70$ tape operations. Neither table prices the multitone qudit schedule, that total is given below them.

\begin{table*}[t]
  \centering
  \footnotesize
  \setlength{\tabcolsep}{4pt}
  \begin{tabular}{r l r l r @{\hspace{1.4em}} r l r l r}
    \toprule
    \# & gate & angle & rule & $T$ ($\mu$s) & \# & gate & angle & rule & $T$ ($\mu$s) \\
    \midrule
    1 & $R$ & $+0.1237$ & R1 & 5.000 & 19 & $C_{\Pi_x}D$ & $+0.0020$ & R5 & 0.649 \\
    2 & $R$ & $+0.1245$ & R1 & 5.000 & 20 & $X_E$ & $+0.0067$ & R4 & 0.315 \\
    3 & $R$ & $+0.1230$ & R1 & 5.000 & 21 & $X_E$ & $+0.0075$ & R4 & 0.350 \\
    4 & $R$ & $+0.0150$ & R1 & 5.000 & 22 & $Z_d$ & $+0.0008$ & R6 & 0.000 \\
    5 & $C_{\Pi}R$ & $-0.0080$ & R2 & 25.000 & 23 & $Z_d$ & $+0.0007$ & R6 & 0.000 \\
    6 & $C_{\Pi}R$ & $-0.0090$ & R2 & 25.000 & 24 & $Z_d$ & $+0.0005$ & R6 & 0.000 \\
    7 & $C_{\Pi}R$ & $-0.0067$ & R2 & 25.000 & 25 & $Z_d$ & $+0.0006$ & R6 & 0.000 \\
    8 & $C_{\Pi_z}D$ & $+0.0082$ & R3 & 2.626 & 26 & $C_{\Pi_z}D$ & $+0.0034$ & R3 & 1.076 \\
    9 & $C_{\Pi_z}D$ & $+0.0073$ & R3 & 2.344 & 27 & $C_{\Pi_z}D$ & $+0.0090$ & R3 & 2.889 \\
    10 & $C_{\Pi_z}D$ & $+0.0090$ & R3 & 2.889 & 28 & $C_{\Pi_z}D$ & $+0.0073$ & R3 & 2.344 \\
    11 & $C_{\Pi_z}D$ & $+0.0034$ & R3 & 1.076 & 29 & $C_{\Pi_z}D$ & $+0.0082$ & R3 & 2.626 \\
    12 & $Z_d$ & $+0.0006$ & R6 & 0.000 & 30 & $C_{\Pi}R$ & $-0.0067$ & R2 & 25.000 \\
    13 & $Z_d$ & $+0.0005$ & R6 & 0.000 & 31 & $C_{\Pi}R$ & $-0.0090$ & R2 & 25.000 \\
    14 & $Z_d$ & $+0.0007$ & R6 & 0.000 & 32 & $C_{\Pi}R$ & $-0.0080$ & R2 & 25.000 \\
    15 & $Z_d$ & $+0.0008$ & R6 & 0.000 & 33 & $R$ & $+0.0150$ & R1 & 5.000 \\
    16 & $X_E$ & $+0.0075$ & R4 & 0.350 & 34 & $R$ & $+0.1230$ & R1 & 5.000 \\
    17 & $X_E$ & $+0.0067$ & R4 & 0.315 & 35 & $R$ & $+0.1245$ & R1 & 5.000 \\
    18 & $C_{\Pi_x}D$ & $-0.0015$ & R5 & 0.649 & 36 & $R$ & $+0.1237$ & R1 & 5.000 \\
    \midrule
    \multicolumn{10}{r}{\textbf{total\quad 210.5$~\mu$s}} \\
    \bottomrule
  \end{tabular}
  \caption{Per-instance timing of the qudit reference circuit, one
    second-order Trotter step at $N_c=3$. Angles are read from the
    compiled circuit. The $36$ rows are the $34$ reference ISA
    operations plus the two $Z_d$ that fuse, both virtual, so the
    total is unchanged.}
  \label{tab:perinstance_qudit}
\end{table*}

\begin{table*}[t]
  \centering
  \footnotesize
  \setlength{\tabcolsep}{4pt}
  \begin{tabular}{r l r l r @{\hspace{1.4em}} r l r l r}
    \toprule
    \# & gate & angle & rule & $T$ ($\mu$s) & \# & gate & angle & rule & $T$ ($\mu$s) \\
    \midrule
    1 & $R$ & $+0.1237$ & R1 & 5.000 & 36 & CNOT$^{\dagger}$ & -- & R8 & 0.000 \\
    2 & $R$ & $+0.1245$ & R1 & 5.000 & 37 & CD & $+0.0008$ & R8 & 0.000 \\
    3 & $R$ & $+0.1230$ & R1 & 5.000 & 38 & CNOT$^{\dagger}$ & -- & R8 & 0.000 \\
    4 & $R$ & $+0.0150$ & R1 & 5.000 & 39 & H$^{\dagger}$ & -- & R8 & 0.000 \\
    5 & $r_z$ & $+0.0005$ & R6 & 0.000 & 40 & H$^{\dagger}$ & -- & R8 & 0.000 \\
    6 & $r_z$ & $+0.0007$ & R6 & 0.000 & 41 & $r_z$ & $+1.5708$ & R6 & 0.000 \\
    7 & CR & $+0.0160$ & R2 & 25.000 & 42 & $r_z$ & $+1.5708$ & R6 & 0.000 \\
    8 & CR & $+0.0180$ & R2 & 25.000 & 43 & $r_z$ & $-1.5708$ & R6 & 0.000 \\
    9 & CR & $+0.0135$ & R2 & 25.000 & 44 & $r_z$ & $-1.5708$ & R6 & 0.000 \\
    10 & CD & $+0.0082$ & R3 & 2.626 & 45 & H$^{\dagger}$ & -- & R8 & 0.000 \\
    11 & CD & $+0.0073$ & R3 & 2.344 & 46 & H$^{\dagger}$ & -- & R8 & 0.000 \\
    12 & CD & $+0.0090$ & R3 & 2.889 & 47 & CNOT$^{\dagger}$ & -- & R8 & 0.000 \\
    13 & CD & $+0.0034$ & R3 & 1.076 & 48 & CD & $+0.0010$ & R8 & 0.000 \\
    14 & MS$_{XX}$ & $+0.0075$ & R7 & 200.000 & 49 & CNOT$^{\dagger}$ & -- & R8 & 0.000 \\
    15 & MS$_{YY}$ & $+0.0075$ & R7 & 200.000 & 50 & H$^{\dagger}$ & -- & R8 & 0.000 \\
    16 & MS$_{XX}$ & $+0.0067$ & R7 & 200.000 & 51 & H$^{\dagger}$ & -- & R8 & 0.000 \\
    17 & MS$_{YY}$ & $+0.0067$ & R7 & 200.000 & 52 & $r_z$ & $+1.5708$ & R6 & 0.000 \\
    18 & H$^{\dagger}$ & -- & R8 & 0.000 & 53 & $r_z$ & $+1.5708$ & R6 & 0.000 \\
    19 & H$^{\dagger}$ & -- & R8 & 0.000 & 54 & MS$_{YY}$ & $+0.0067$ & R7 & 200.000 \\
    20 & CNOT$^{\dagger}$ & -- & R8 & 0.000 & 55 & MS$_{XX}$ & $+0.0067$ & R7 & 200.000 \\
    21 & CD & $+0.0008$ & R8 & 0.000 & 56 & MS$_{YY}$ & $+0.0075$ & R7 & 200.000 \\
    22 & CNOT$^{\dagger}$ & -- & R8 & 0.000 & 57 & MS$_{XX}$ & $+0.0075$ & R7 & 200.000 \\
    23 & H$^{\dagger}$ & -- & R8 & 0.000 & 58 & CD & $+0.0034$ & R3 & 1.076 \\
    24 & H$^{\dagger}$ & -- & R8 & 0.000 & 59 & CD & $+0.0090$ & R3 & 2.889 \\
    25 & H$^{\dagger}$ & -- & R8 & 0.000 & 60 & CD & $+0.0073$ & R3 & 2.344 \\
    26 & H$^{\dagger}$ & -- & R8 & 0.000 & 61 & CD & $+0.0082$ & R3 & 2.626 \\
    27 & CNOT$^{\dagger}$ & -- & R8 & 0.000 & 62 & CR & $+0.0135$ & R2 & 25.000 \\
    28 & CD & $+0.0010$ & R8 & 0.000 & 63 & CR & $+0.0180$ & R2 & 25.000 \\
    29 & CNOT$^{\dagger}$ & -- & R8 & 0.000 & 64 & CR & $+0.0160$ & R2 & 25.000 \\
    30 & H$^{\dagger}$ & -- & R8 & 0.000 & 65 & $r_z$ & $+0.0007$ & R6 & 0.000 \\
    31 & H$^{\dagger}$ & -- & R8 & 0.000 & 66 & $r_z$ & $+0.0005$ & R6 & 0.000 \\
    32 & $r_z$ & $-1.5708$ & R6 & 0.000 & 67 & $R$ & $+0.0150$ & R1 & 5.000 \\
    33 & $r_z$ & $-1.5708$ & R6 & 0.000 & 68 & $R$ & $+0.1230$ & R1 & 5.000 \\
    34 & H$^{\dagger}$ & -- & R8 & 0.000 & 69 & $R$ & $+0.1245$ & R1 & 5.000 \\
    35 & H$^{\dagger}$ & -- & R8 & 0.000 & 70 & $R$ & $+0.1237$ & R1 & 5.000 \\
    \midrule
    \multicolumn{10}{l}{$J_l$, $2\times$ one MS window (rule R7), each window standing for $8$~H, $4$~CNOT and $2$~CD priced together as one MS gate \hfill $2\times200.0$} \\
    \multicolumn{10}{r}{\textbf{total\quad 2207.9$~\mu$s}} \\
    \bottomrule
  \end{tabular}
  \caption{Per-instance timing of the qubit reference circuit, one
    second-order Trotter step at $N_c=3$, compiled to $70$
    operations. The Hadamards, CNOTs and conditional displacements of
    the two $J_l$ sandwiches are listed for completeness but carry no
    separate charge. $J_l$ is priced at one M{\o}lmer--S{\o}rensen
    window each, added below the rows. $^{\dagger}$Inside a $J_l$
    sandwich, charged as part of that term's single $200~\mu$s
    M{\o}lmer--S{\o}rensen window, not separately.}
  \label{tab:perinstance_qubit}
\end{table*}

\noindent The projected ratio is $T_{\rm qubit}/T_{\rm qudit} = 2207.9/210.5 = 10.5$.
Reading the two tables side by side, every gate class except one is
priced identically in the two encodings, the dispersive gates, the
conditional displacements, the rotations and the virtual phases. The one
class that differs is the electronic exchange, which is a local carrier
pulse $X_E$ on the qudit and a two-ion M{\o}lmer--S{\o}rensen window on
the qubit. Those windows are $2000$ of the $2208~\mu$s, and that
difference accounts for the projected ratio.

\noindent Under the multitone schedule of Table~\ref{tab:gatecounts} the same step is $27$ control blocks. The two $X_E$ pairs are each driven as one flat block at the aggregate field, and the $C_{\Pi_x}D$ terms fuse into a single summed exponential. Only those two classes change. The qudit total becomes $209.5~\mu$s and the ratio $2207.9/209.5 = 10.5$.

\begin{table*}[t]
  \centering
  \begin{tabular}{ccccccc}
\toprule
$N_c$ & Qudit ops & Qubit ops & $T_{\rm qubit}$ ($\mu$s) &
MS windows & $T_{\rm qudit}$ ($\mu$s) & Ratio \\
\midrule
3  & 34  & 32  & 2208 & 10 & 210 & 10.5 \\
6  & 82  & 80  & 6045 & 28 & 453 & 13.3 \\
8  & 114 & 112 & 8603 & 40 & 615 & 14.0 \\
12 & 178 & 176 & 13718 & 64 & 939 & 14.6 \\
\bottomrule
\end{tabular}
  \caption{Projected time for one second-order Trotter step under the
  per-instance rules above, for the one-dimensional chain extension of
  Ref.~\cite{vu2025} with one high frequency mode per chromophore and
  one low frequency mode per interior chromophore, $2N_c-2$ modes in
  all. Operation counts are logical and grow linearly in $N_c$.
  At $N_c=3$, every entry is
  read from a compiled circuit. The larger $N_c$ are projected, charging
  each gate class its $N_c=3$ mean instance time at the instance count its
  closed form gives, since this work compiles the qudit circuit only at
  $N_c=3$. The ratio is $T_{\rm qubit}/T_{\rm qudit}$. With the multitone
  $27$-control-block qudit schedule at $N_c=3$, the qudit step is
  $209.5~\mu$s and the ratio $10.5$. All rows assume the manifold fits in
  one ion, $d_{\rm eff}\ge N_c+2$. Rows beyond $N_c=5$ are not
  $^{40}$Ca$^+$ predictions. The split case is given in
  Sec.~\ref{sec:multiion}.}
  \label{tab:runtime}
\end{table*}

Table~\ref{tab:runtime} gives the result. At $N_c=3$ the qudit step is
projected at $210~\mu$s against $2208~\mu$s for the sequential
qubit reference, a ratio of $10.5$.

The sequential runtime ratio increases from $10.5$ at $N_c=3$ to
$14.6$ at $N_c=12$ under the stated timing model. Rows beyond $N_c=5$
assume an ion with $d_{\rm eff}\ge N_c+2$ and are not $^{40}$Ca$^+$
predictions. When the manifold splits over several ions the ratio drops,
see Sec.~\ref{sec:multiion}.

The dispersive detuning behind $C_{\Pi}R$ is a chosen operating point, whose
choice and sweep are given in Appendix~\ref{app:delta}. Across that window the sequential
ratio at $N_c=3$ ranges from $6.5$ at $20$~kHz to $15.7$ at $80$~kHz, with
the high detuning end requiring about twice the measured sideband
amplitude.

\end{document}